\documentclass[lettersize,journal,onecolumn,comsoc,10pt]{IEEEtran}
\IEEEoverridecommandlockouts
\usepackage{cite}
\usepackage{amsmath,amssymb,amsfonts}
\usepackage{algorithmic}
\usepackage{graphicx}
\usepackage{textcomp}
\usepackage[T1]{fontenc}
\usepackage{xcolor}
\usepackage{subcaption}
\usepackage{booktabs}
\usepackage{makecell}

\usepackage{amsthm}
\theoremstyle{definition}
\newtheorem{definition}{\normalfont\bfseries Definition}
\newtheorem{theorem}{\normalfont\bfseries Theorem}

\newtheorem{remark}{\normalfont\bfseries Remark}
\usepackage{authblk}
\usepackage{multirow}

\renewcommand\thesection{\arabic{section}}
\renewcommand\thesubsection{\thesection.\arabic{subsection}}
\renewcommand\thesubsubsection{\thesubsection.\arabic{subsubsection}}

\renewcommand\thesectiondis{\arabic{section}}
\renewcommand\thesubsectiondis{\thesectiondis.\arabic{subsection}}
\renewcommand\thesubsubsectiondis{\thesubsectiondis.\arabic{subsubsection}}

\newcommand{\A}{a}
\newcommand{\B}{b}
\newcommand{\aaa}{a}
\newcommand{\classK}{\alpha}

\allowdisplaybreaks
\def\BibTeX{{\rm B\kern-.05em{\sc i\kern-.025em b}\kern-.08em
    T\kern-.1667em\lower.7ex\hbox{E}\kern-.125emX}}

\usepackage[section]{placeins}

\begin{document}

\title{Safety-Critical Traffic Control \\ by Connected Automated Vehicles}

\author[1,2,*]{Huan Yu\thanks{* corresponding author\\ Email: huanyu@ust.hk}}
\author[1]{Chenguang Zhao}
\author[3]{Tamas G. Molnar}

\affil[1]{Thrust of Intelligent Transportation, The Hong Kong University of Science and Technology (Guangzhou), Nansha, Guangzhou, 511400, Guangdong, China.}
\affil[2]{Department of Civil and Environmental Engineering, The Hong Kong University of Science and Technology, Clear Water Bay, Hong Kong SAR, China.}
\affil[3]{Department of Mechanical and Civil Engineering,
California Institute of Technology, Pasadena, CA 91106, USA}

\renewcommand\Authands{ and }
\maketitle

\begin{abstract}
Connected automated vehicles (CAVs) have shown great potential in improving traffic throughput and stability. Although various longitudinal control strategies have been developed for CAVs to achieve string stability in mixed-autonomy traffic systems, the potential impact of these controllers on safety has not yet been fully addressed. This paper proposes {\em safety-critical traffic control (STC)} by CAVs---a strategy that allows a CAV to stabilize the traffic behind it, while maintaining safety relative to both the preceding vehicle and the following connected human-driven vehicles (HDVs). Specifically, we utilize control barrier functions (CBFs) to impart collision-free behavior with formal safety guarantees to the closed-loop system. The safety of both the CAV and HDVs is incorporated into the framework through a quadratic program-based controller, that minimizes deviation from a nominal stabilizing traffic controller subject to CBF-based safety constraints. Considering that some state information of the following HDVs may be unavailable to the CAV, we employ state observer-based CBFs for STC. Finally, we conduct extensive numerical simulations---that include vehicle trajectories from real data---to demonstrate the efficacy of the proposed approach in achieving string stable and, at the same time, provably safe traffic.

\end{abstract}


\section{Introduction}

The efficiency of transportation is largely affected by the smoothness of the traffic flow.
As such, highway traffic often suffers from traffic oscillations, in which the vehicles driving on the road undergo repeated deceleration-acceleration motions.
These oscillations may even amplify along the road, indicating the so-called {\em string instability} of traffic, that ultimately leads to stop-and-go traffic congestion.
Congestion, on one hand, negatively impacts travel time and fuel consumption.
On the other hand, they also pose additional risk of conflict or collision between leader and follower vehicles, thereby adversely affecting traffic safety.
Therefore, significant research efforts have been invested into mitigating the traffic oscillations by stabilizing traffic, and the area of {\em traffic control} has emerged as a solution.

This work focuses on providing a safety-critical solution to traffic control through regulating the motions of connected automated vehicles (CAVs) traveling in the traffic flow.
Our goal is to control the CAVs such that they facilitate string stability for the traffic behind while also guaranteeing that a safe distance is always kept between the neighboring vehicles.
To highlight the need for and the benefits of the proposed approach,  we first give a brief overview on the relevant literature on stabilizing traffic controllers and traffic safety.

\subsection{Traffic Control by Connected Automated Vehicles}

Research work on traffic control relied on conventional {\em road-based traffic control} systems, such as ramp-metering or varying speed limits, that regulate traffic on spatially-fixed road segments in a centralized manner~\cite{horowitz2005design,yu2019traffic, zhang2019pi, zhang2016combined}.
Then, the rapid development of CAV technology has drawn extensive attention to {\em vehicle-based traffic control} of transportation systems through information sharing and coordination between CAVs~\cite{vcivcic2018traffic, vcivcic2021coordinating, molnar2021delayed, stern2018dissipation, talebpour2016influence,zheng2015stability, zheng2016distributed}.
Vehicle-based traffic controllers rely on regulating the driving behavior of CAVs, through which the smoothness and throughput of traffic is positively impacted.
These distributed, vehicle-based traffic control techniques are the main focus of our work.

There exist several linear and nonlinear feedback control strategies for regulating the motions of CAVs, that rely on various topologies in vehicle-to-vehicle (V2V) communication, different levels of automation on CAVs, and various formation geometries of multiple CAVs~\cite{li2017dynamical}.
Early studies focused on {\em adaptive cruise control (ACC)}~\cite{marsden2001towards} that automatically adjusts the speed of automated vehicles to maintain a safe distance from the preceding vehicle. Extended from ACC, {\em cooperative adaptive cruise control (CACC)}~\cite{milanes2013cooperative, alam2014guaranteeing} was developed to control platoons of CAVs by leveraging V2V connectivity.
CACC frameworks include, for example,  linear feedback controllers that mitigate disturbances in the platoon~\cite{milanes2013cooperative} and model predictive controllers that also guarantee safe distances between adjacent vehicles~\cite{massera2017safe}.
Furthermore, control strategies were also proposed for a single CAV instead of an entire platoon, such as the framework of {\em connected cruise control (CCC)}~\cite{orosz2016connected}.

Despite the promising future envisioned for fully connected and automated traffic systems, a long transition period is inevitable in which CAVs and human-driven vehicles (HDVs) coexist.
Many recent research works, therefore, have addressed mixed-autonomy traffic where only a portion of the vehicles are CAVs~\cite{cui2017stabilizing,jovanovic2016controller,zheng2020smoothing,wang2020controllability,wang2021leading,wu2021flow,zhu2018analysis,stern2018dissipation}.
In particular, some of these works have been focusing on how CAVs can influence the overall traffic.
In~\cite{cui2017stabilizing} a single automated vehicle was used to stabilize the traffic flow.
The ability of CAVs to dampen stop-and-go waves in mixed traffic was also demonstrated by experiments in~\cite{Ge2018, stern2018dissipation} and by theoretical analysis in~\cite{huang2020scalable,li2014stop,yu2018stabilization, Avedisov2022}.
The controllability and reachability of mixed-autonomy traffic were studied by~\cite{zheng2020smoothing}, which proved that a single CAV can stabilize traffic in the ring-road setting.

In this work, we rely on a specific traffic control strategy by CAVs, that was proposed as {\em adaptive traffic control (ATC)} in~\cite{Molnar2020cdc,molnar2023virtual} and as {\em leading cruise control (LCC)} in~\cite{wang2021leading, wang2021data}.
Both of these approaches control the CAV such that it simultaneously adapts its motion to follow the preceding HDVs and to lead the following HDVs---ultimately allowing the mitigation of traffic oscillations.
While these approaches provide the desired string stable behavior, they lack formal guarantees about maintaining safe distances between the vehicles, which our present work seeks to address.

\subsection{Safety of Mixed-autonomy Traffic}
Safety is of paramount importance for traffic systems. Although many studies address the potential improvement of efficiency~\cite{talebpour2016influence,milanes2013cooperative,shang2021impacts,kesting2010enhanced}, energy-saving~\cite{vahidi2018energy,rios2018impact} and stability~\cite{cui2017stabilizing,zheng2020smoothing,wang2020controllability} of the mixed traffic by adopting CAVs, only a few existing works have attempted to evaluate the impact of CAVs on traffic safety.
These works include studies about safety under different CAV penetration rates~\cite{ye2019evaluating, li2017evaluation,rahman2018longitudinal} and in signalized intersections~\cite{morando2018studying,arvin2020safety}.

The traffic safety is usually guaranteed by reducing the risk of rear-end collisions.
Therefore, a number of safety indicators, also known as surrogate safety measures, have appeared to evaluate the rear-end collision risk by establishing the relationship between the longitudinal safety and car-following state~\cite{vogel2003comparison,li2017evaluation,rahman2018longitudinal}.
Some of these safety measures are {\em time-based indicators}, such as time headway (TH) and time to collision (TTC), which are constructed based on the time that elapses before a collision could occur~\cite{vogel2003comparison}.
Others measures are {\em distance-based indicators}, such as distance headway (DH) and minimum stopping distance headway (SDH), which evaluate safety based on the spacing gap between the vehicles~\cite{Treiberbook, ro2020new}. 
These metrics are essential to construct safety-critical longitudinal controllers for CAVs.

Safety-critical control is usually approached by constraint-handling control methods such as classical optimal control, model predictive control, barrier Lyapunov function (BLF)~\cite{zhu2019barrier} and control barrier function (CBF)~\cite{ames2019control,krstic2021inverse,almubarak2021hjb,xiao2022high}. 
In this work, we rely on the framework of CBFs to maintain safety, due to its ability to take advantage of existing nominal controllers such as those that stabilize traffic.
Importantly, CBFs have been applied on a wide range of safety-critical systems, including adaptive cruise control~\cite{ames2014control} and its experimental implementation on heavy-duty trucks~\cite{alan2022cbfs},
automated vehicle experiments in multi-lane traffic~\cite{gunter2022experimental},
obstacle avoidance with automated vehicles~\cite{chen2018obstacle},
multi-agent systems representing automated vehicles~\cite{jankovic2022multiagent},
traffic merging~\cite{xiao2019decentralized}, and
roundabout crossing~\cite{abduljabbar2021cbfbased}.

Although the stabilization of traffic by CAVs and the safety-critical control of CAVs have been studied separately in the literature, the integration of these two  have not yet been accomplished, to the best of our knowledge.
Some works have conducted analysis in this area, such as~\cite{li2022trade} that quantified the trade-off between stability and safety, and~\cite{shi2021empirical,makridis2019response,gunter2019modeling} that pointed out that CAVs usually leave comparable or even longer headway than HDVs for safety concerns.
Yet, a control framework that directly addresses the objective of traffic stabilization while maintaining provable safety guarantees has not yet been established.
Our present work intends to fill this gap, via proposing safety-critical traffic control by CAVs.

\subsection{Contributions}

In this paper, we propose a {\em safety-critical traffic control (STC)} strategy for mixed-autonomy traffic systems, in which a single CAV leads several HDVs such that it both guarantees safety and achieves string stability.

The proposed STC framework is illustrated in Fig.~\ref{fig:framework} for single-lane mixed traffic including a head HDV, a CAV and several following HDVs.
The CAV, who detects the head HDV by range sensors or connectivity and is connected to the following HDVs, has access to the velocities $v_i$ of the vehicles and the spacings $s_i$ between them---or at least some of these data.
This information is used in a three-layer hierarchical control structure.
In the top layer, a {\em traffic operator} can be used to command the desired steady states (spacing and velocity) for the mixed traffic.
In the middle layer, a nominal {\em stabilizing controller} regulates the motion of the CAV in order to drive the traffic towards the desired steady states and fulfill the operator's command.
In the bottom layer, a {\em safety filter} modifies the nominal stabilizing control input using control barrier functions (CBFs) to guarantee safety, and then outputs the final safety-critical controller for the CAV.
This framework not only enables flexibility to choose different designs at each layer, but also provides privacy by keeping the details of each design to its own layer only.  

The STC framework is established through the following contributions.
\begin{itemize}
    \item A quadratic program-based controller is constructed for CAVs, that incorporates a nominal stabilizing traffic controller and safety-constraints derived from control barrier functions---the end result being safety and string stability simultaneously.
    \item Safety guarantees are provided both for the CAV as hard constraint and for the following HDVs as soft constraints.
    These guarantees are based on various safety measures (TH, TTC and SDH) that are analyzed and compared.
    \item STC is extended to the case of output feedback by the help of state observers to address scenarios in which the CAV does not have access to all the states of the HDVs.
    To this end, a state observer-based CBF approach is formulated.
    \item Extensive numerical simulations are conducted to show the benefits of the proposed STC, identify its limitations, and analyze the effect of parameters.
    The simulations also incorporate real trajectory data for the head HDV.
\end{itemize}

\begin{figure}[t!]
    \centering
    \includegraphics[width=0.75\linewidth]{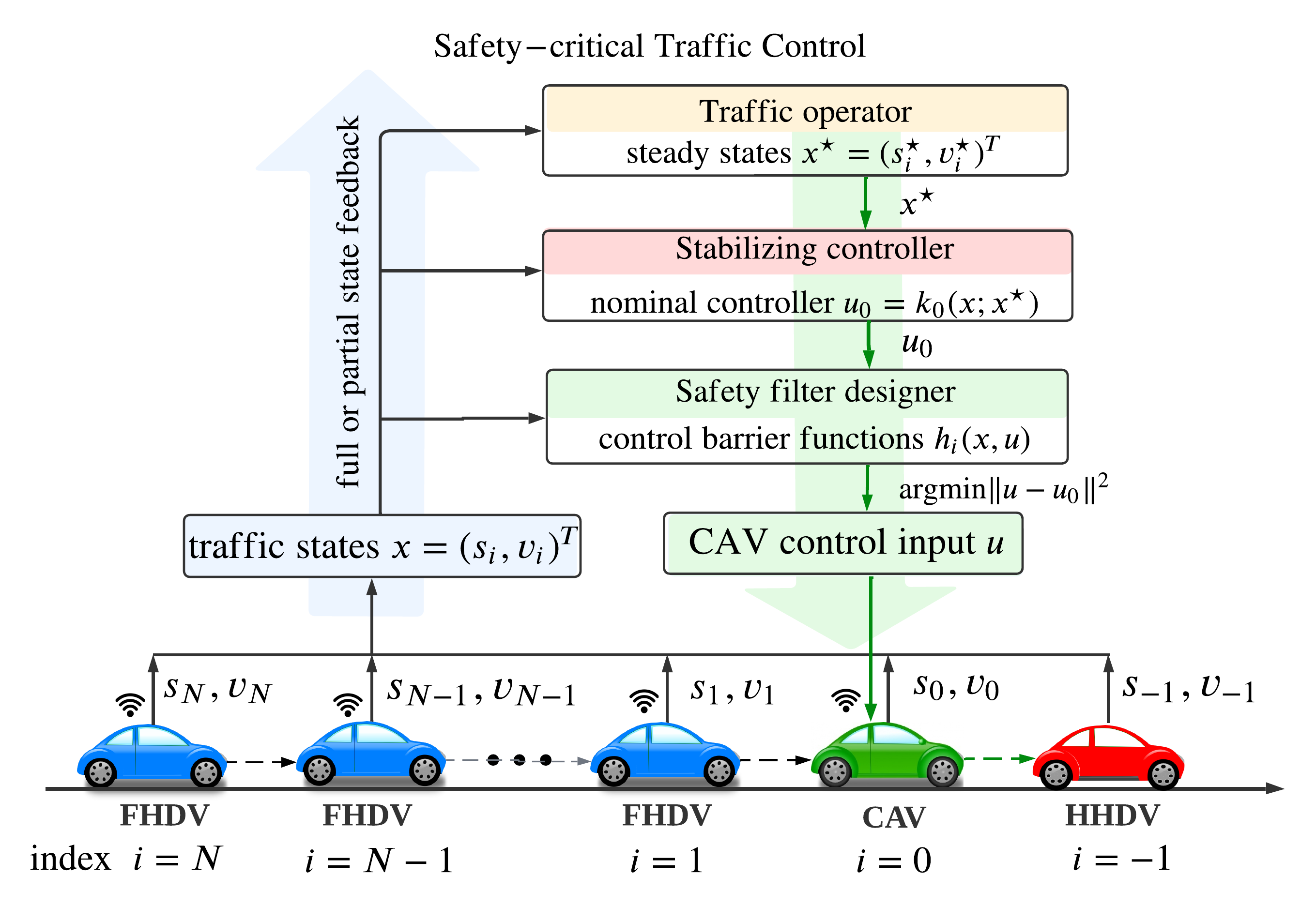}
    \caption{The proposed framework of {\em safety-critical traffic control (STC)}, wherein the traffic flow is stabilized in a provable safe fashion by regulating the motions of connected automated vehicles (CAVs) while leveraging information from vehicle-to-vehicle connectivity with neighboring connected human-driven vehicles (HDVs).
    The STC framework has hierarchical structure with a {\em traffic operator} at the top, that can be used to design the desired steady state of the traffic flow, a nominal {\em stabilizing controller} in the middle, that drives the CAV to mitigate traffic congestion without considering safety, and a {\em safety filter} at the bottom, that modifies the nominal control input to establish formal safety guarantees for collision avoidance between vehicles.}
    \label{fig:framework}
\end{figure}

The rest of this paper is organized as follows. Section~\ref{sec:smoothing} presents the modeling of mixed-autonomy traffic, and revisits a nominal control approach for CAVs to stabilize traffic.
Section~\ref{sec:safety} introduces the framework of STC, to modify the nominal controller and endow it with formal safety-guarantees.
To this end, safe spacing policies are described, the theory of control barrier functions is revisited, and state observers are utilized in case of output feedback.
Section~\ref{sec:simulation} provides extensive numerical simulations (that involve vehicle trajectories from real data) to demonstrate the efficacy of STC.
Section~\ref{sec:conclusions} concludes the paper.

\section{Smoothing Mixed-autonomy Traffic by Connected Automated Vehicles}
\label{sec:smoothing}

We first model the longitudinal dynamics of mixed-autonomy traffic.
We consider a single-lane scenario that involves a head vehicle, followed by a single connected automated vehicle (CAV) and $N$ subsequent human-driven vehicles (HDVs) that are connected to the CAV; see Fig.~\ref{fig:framework}.
We revisit the notion of string stability and design a nominal controller for the CAV that can stabilize the traffic behind it.
This controller will be endowed with provable safety guarantees in the next section.

\subsection{Modeling of Mixed-autonomy Traffic}
The mixed-autonomy traffic model constitutes two parts: a car-following model for HDVs and a model for the longitudinal control of the CAV.
The car-following dynamics of the HDVs (indexed ${i \in \{1, \ldots, N\}}$) are described as
\begin{align}
\begin{split}
\dot s_i &= v_{i-1} - v_{i}, \\
\dot{v}_{i} &=F_i\left(s_{i}, \dot{s}_{i},  v_{i}\right),
\end{split}
\label{eq:CF}
\end{align}
where $s_i$ represents the spacing (distance headway) between HDV $i$ and its predecessor vehicle $i-1$.
Function $F_i$ describes how HDV $i$ controls its acceleration for car-following as a function of its own velocity $v_i$, the spacing $s_i$ and the derivative $\dot s_i$ of the spacing (i.e., the relative speed). 
For the CAV (with index $i=0$), the longitudinal dynamics are governed by
\begin{align}
\begin{split}
    \dot{{s}}_{0} &= v_{-1}- {v}_{0}, \\
    \dot{{v}}_{0} &=u,
\end{split}
\label{eq:CAV_dyn}
\end{align}
where the acceleration of the CAV is the control input $u$ to be designed, and $v_{-1}$ is the velocity of the head vehicle (${i = -1}$).

The state of the mixed-autonomy traffic model \eqref{eq:CF}-\eqref{eq:CAV_dyn} is ${n=2N+2}$ dimensional, and it consists of the spacing and speed of the CAV and the following HDVs, given by ${x \in \mathbb{R}^n}$,
\begin{align}
x =\begin{bmatrix}
s_{0} & v_{0} & \ldots & s_{N} & v_{N}
\end{bmatrix}^\top.
\label{eq:state sv}
\end{align}
This leads to a general nonlinear mixed traffic model \begin{equation}
    \dot{x}=f(x,r) + g(x) u,  \label{system_nonlin}
\end{equation}
\begin{equation}
    f(x,r) = \begin{bmatrix}
    f_{0}(x,r) \\
    f_{1}(x) \\
    \vdots \\
    f_{N}(x)
    \end{bmatrix}, \quad
    g(x) = \begin{bmatrix}
    g_{0} \\ 0_{2 \times 1} \\ \vdots \\ 0_{2 \times 1}
    \end{bmatrix}, \quad
    f_{0}(x,r)=\begin{bmatrix}
    v_{-1}-v_{0} \\
    0
    \end{bmatrix}, \quad
    f_{i}(x)=\begin{bmatrix}
    v_{i-1}-v_{i} \\
    F_i(s_i,v_{i-1}-v_{i},v_i)
    \end{bmatrix}, \quad
    g_{0}=\begin{bmatrix}
    0 \\
    1
    \end{bmatrix},
    \label{eq:fgfunctions}
\end{equation}
where ${r=v_{-1}}$ is a time-varying reference signal, $0_{2 \times 1}$ is the two dimensional zero vector, and
${f: \mathbb{R}^n \times \mathbb{R} \to \mathbb{R}^n}$,
${g: \mathbb{R}^n \to \mathbb{R}^n}$.
Note that, alternatively, one could also formulate this model using the longitudinal positions $p_i$ of the vehicles' rear bumper, that are related to the spacings $s_i$ through the vehicle lengths $l_i$ as $s_{i} = p_{i-1} - p_{i} - l_{i}$.

For simplicity, we design a nominal controller for the CAV based on linearization of the dynamics.
Thus, we first consider the steady state of the system, given by the equilibrium speed $v^\star$ as $v_i(t) \equiv v^\star$ and the equilibrium spacing $s_i^\star$ as $s_i(t) \equiv s_i^\star$, ${i \in \{0, \ldots, N\}}$.
Note that the equilibrium spacing $s_i^\star$ of the HDVs (${i\in \{1, \ldots, N\}}$) satisfy
\begin{equation}
    F_i \left(s_i^{\star}, 0, v^{\star}\right)=0,
\end{equation}
and these may be different from the equilibrium spacing $s_0^\star$ of the CAV that is determined by its controller.
Potentially, these equilibrium states can also be assigned by a traffic operator, as described in Fig.~\ref{fig:framework}.
Using the perturbations
${\tilde{s}_{i}=s_{i}-s_i^{\star}}$ and
${\tilde{v}_{i}=v_{i}-v^{\star}}$,
the linearized dynamics can be obtained from \eqref{eq:CF}-\eqref{eq:CAV_dyn} as
\begin{align}
\begin{split}
    \dot{\tilde{s}}_{i} &= \tilde{v}_{i-1} - \tilde{v}_{i}, \\
    \dot{\tilde{v}}_{i} &= \aaa_{i1} \tilde{s}_{i} - \aaa_{i2} \tilde{v}_{i} + \aaa_{i3} \tilde{v}_{i-1},
\end{split}
\label{eq:CF_lin}
\end{align}
where $\aaa_{i1}=\frac{\partial F_i}{\partial s_i}, \aaa_{i2}=\frac{\partial F_i}{\partial \dot{s}_i}-\frac{\partial F_i}{\partial v_i}, \aaa_{i3}=\frac{\partial F}{\partial \dot{s}_i}$ are evaluated at the steady states, and
\begin{align}
\begin{split}
    \dot{\tilde{s}}_{0} &= \tilde{v}_{-1} - \tilde{v}_{0}, \\
    \dot{\tilde{v}}_{0} &= u.
\end{split}
\label{eq:CAV_lin}
\end{align}

The state of the linearized mixed-autonomy traffic model \eqref{eq:CF_lin}-\eqref{eq:CAV_lin} is
\begin{align}
x =\begin{bmatrix}
\tilde{s}_{0} & \tilde{v}_{0} & \ldots & \tilde{s}_{N} & \tilde{v}_{N}
\end{bmatrix}^\top,
\label{eq:state sv tilde}
\end{align}
hence \eqref{eq:CF_lin}-\eqref{eq:CAV_lin} can be written in the compact form
\begin{equation}
    \dot{x}=A x + B u + D r,
    \label{system}
\end{equation}
\begin{equation}
    A = \begin{bmatrix}
    P_{0} & & & \\
    Q_{1} & P_{1} & & \\
    & \ddots & \ddots & \\
    & & Q_{N} & P_{N}
    \end{bmatrix}\!, \;
    B = \begin{bmatrix}
    b_{0} \\ 0_{2 \times 1} \\ \vdots \\ 0_{2 \times 1}
    \end{bmatrix}\!, \;
    D = \begin{bmatrix}
    d_{0} \\ 0_{2 \times 1} \\ \vdots \\ 0_{2 \times 1}
    \end{bmatrix}\!, \;
    P_{0}=\begin{bmatrix}
    0 & -1 \\
    0 & 0
    \end{bmatrix}\!, \;
    P_{i}=\begin{bmatrix}
    0 & -1 \\
    \aaa_{i1} & -\aaa_{i2}
    \end{bmatrix}\!, \;
    Q_{i}=\begin{bmatrix}
    0 & 1 \\
    0 & \aaa_{i3}
    \end{bmatrix}\!, \;
    b_{0}=\begin{bmatrix}
    0 \\
    1
    \end{bmatrix}\!, \;
    d_{0}=\begin{bmatrix}
    1 \\
    0
    \end{bmatrix}\!,
    \label{eq:ABmatrices}
\end{equation}
where
$r = \tilde{v}_{-1}$ and
${A\in\mathbb{R}^{n \times n}}$,
${B, D \in \mathbb{R}^{n}}$.

\subsection{Nominal Controller Design and Head-to-tail String Stability}

Now we design a nominal control input $u_0$ that allows the CAV to follow the head vehicle while leading $N$ HDVs such that the resulting traffic flow is smooth---specficially, {\em string stable}.
In the next section, we modify $u_0$ to an input $u$ that also has formal safety guarantees for collision avoidance (without restriction on the choice of the nominal controller $u_0$).

Connectivity allows the CAV to obtain the speeds $v_i$ and positions $p_i$ of the HDVs based on GPS and basic safety messages.
This information characterizes the state of the human-driven traffic behind the CAV, hence the CAV may use these data to control and stabilize the traffic flow, ultimately mitigating traffic congestions.
This led to the idea of \emph{adaptive traffic control (ATC)} in~\cite{Molnar2020cdc, molnar2023virtual}, where the CAV responds to the state $v_i$, $p_i$ of the HDVs (obtained from connectivity) while considering the state $v_{-1}$, $p_{-1}$ of the head vehicle (provided by range sensors or connectivity), using a control law of the form
\begin{equation}
    u_0 = \bar{F}_0(p_{-1}, v_{-1}, p_0, v_0, p_1, v_1, \ldots, p_N, v_N).
    \label{eq:ATC}
\end{equation}
Note that if the CAV is not connected to an HDV and does not have access to its state, the corresponding state $p_i$, $v_i$ is omitted from the response.
Examples for ATC~\eqref{eq:ATC} and their analysis in the presence of response delays can be found in~\cite{molnar2023virtual}.

In special cases, the positions $p_i$ can be converted into the spacings $s_i$, resulting in a controller:
\begin{equation}
    u_0 = F_0(v_{-1}, s_0, v_0, s_1, v_1, \ldots, s_N, v_N).
\end{equation}
This conversion requires that the CAV is connected to all following HDVs and all longitudinal positions $p_i$ and vehicle lengths are known to the CAV.
Alternatively, if only some of the HDVs are connected, they can potentially communicate their spacing (i.e., realize cooperative perception).
Furthermore, if the CAV even has access to the perturbations $\tilde{s}_i$, $\tilde{v}_i$ (not just $s_i$, $v_i$), then a linear controller can be formulated as
\begin{equation}
    u_0 = \aaa_{1} \tilde{s}_{0} - \aaa_{2} \tilde{v}_{0} + \aaa_{3} \tilde{v}_{-1} +
    \!\!\sum_{i=1}^{N}\!\!
    \left(\mu_{i} \tilde{s}_{i} + k_{i} \tilde{v}_{i} \right),
    \label{eq:nominal controller}
\end{equation}
where $\mu_{i}, k_{i}$ are the feedback gains corresponding to the state of vehicle $i$.
These gains can be designed based on the linear model~\eqref{system}.
Controller~\eqref{eq:nominal controller} is called \emph{leading cruise control (LCC)}, and it was proposed in~\cite{wang2021leading}.

The above traffic control strategies can be compactly written as
\begin{equation}
    u_0 = k_0(x,r),
    \label{eq:nominal_nonlin}
\end{equation}
where the state $x$ and reference signal $r$ can be associated either with the nonlinear model~\eqref{system_nonlin} or with the linear model~\eqref{system}.
For simplicity, the examples of this paper are presented for the linear model~\eqref{system} and the linear LCC strategy~\eqref{eq:nominal controller} as nominal controller, under the assumption that the parameters of the HDVs are identical:
${\aaa_{i1}=\aaa_{1}}$,
${\aaa_{i2}=\aaa_{2}}$,
${\aaa_{i3}=\aaa_{3}}$,
${\forall i \in \{1, \ldots, N\}}$.
The general framework of STC, however, could accomodate any other nominal stabilizing controller of the form~(\ref{eq:nominal_nonlin}).

In order to design the gains $\mu_{i}, k_{i}$ of the nominal controller~\eqref{eq:nominal controller}, first we notice that the pair $(A,B)$ is controllable for $\aaa_1 - \aaa_2\aaa_3 + \aaa_3^2 \neq 0$, which was proven in~\cite{wang2021leading}.
Under this condition, we can design the control gains such that
speed perturbations decay along the chain of vehicles and traffic congestions are mitigated, i.e., the closed-loop system is {\em string stable}.
While there exist various definitions of string stability~\cite{feng2019string}, suited both for nonlinear and linear dynamics, we focus on one particular notion, \emph{head-to-tail string stablity}\cite{zhang2016motif}, that is convenient for the analysis of linear systems.

Head-to-tail string stability is frequently used to evaluate the CAV's ability to attenuate velocity fluctuations, thereby mitigating congestions, smoothing and stabilizing traffic.
Given the Laplace transforms $\widetilde{V}_{-1}(s)$,  $\tilde{V}_{N}(s)$ of the velocity perturbations $\tilde{v}_{-1}(t)$, $\tilde{v}_{N}(t)$ of the head and tail vehicles, the head-to-tail transfer function \cite{zhang2016motif} is defined as
    \begin{equation}
    G(s)=\frac{\widetilde{V}_{N}(s)}{\widetilde{V}_{-1}(s)}.
    \end{equation}
This function characterizes how much speed perturbations amplify along the mixed chain of $N+2$ vehicles.
Head-to-tail string stability holds if and only if \cite{zhang2016motif}
    \begin{equation}
    |G({\rm j} \omega)|^{2}<1, \quad \forall \omega>0,
    \label{eq:string stable condition}
    \end{equation}
where ${\rm j}^2=-1$ and $|.|$ denotes the modulus.
This condition means that the speed fluctuations of the tail vehicle are smaller in steady state than those of the head vehicle for any angular frequency $\omega$.

For the linearized system~\eqref{system} with controller~\eqref{eq:nominal controller},~\cite{wang2021leading} derived the head-to-tail transfer function in the form
\begin{equation}\label{eq:transfer function gammas}
    G(s) = \frac{\phi(s)}{\psi(s)-\sum_{i=1}^{N} \left( \mu_i \left(\frac{\phi(s)}{\psi(s)} -1\right) +k_i s \right) \left(\frac{\psi(s)}{\phi(s)} \right)^i } \left(\frac{\phi(s)}{\psi(s)} \right)^N,
\end{equation}
with 
\begin{equation}
    \phi(s) = \aaa_3  s + \aaa_1, \quad \psi(s) = s^2 + \aaa_2 s + \aaa_1.
\end{equation}
This expression enables one to design $\mu_i$ and $k_i$ such that~(\ref{eq:string stable condition}) holds and head-to-tail string stability is attained; see~\cite{wang2021leading}.

\subsection{Observer Design}
\label{section observability and observer}

Notice that controllers of the form~\eqref{eq:nominal_nonlin} rely on the full state $x$ of the system.
As stated in \cite{wang2022deep}, however, it may be difficult to measure all state variables in $x$.
Instead, an output $y \in \mathbb{R}^{n_y}$ may be available, given by a function $c:\mathbb{R}^n \to \mathbb{R}^{n_y}$ of the state:
\begin{equation}
    y=c(x).
    \label{eq:output_nonlin}
\end{equation}
For example, due to the heterogeneous behaviors of HDVs, their equilibrium spacing $s_i^\star$ may not be captured accurately by the car-following model \eqref{eq:CF}, which makes it infeasible to measure the spacing errors $\tilde{s}_i$. On the other hand, the equilibrium spacing of the CAV is designed, thus both its spacing error and velocity error could potentially be measured.
In this setting, the output of the system can be written as:
\begin{equation}
\begin{aligned}
    y &=
    \begin{bmatrix}
    \tilde{s}_{0} & \tilde{v}_{0} & \tilde{v}_{1} & \tilde{v}_{2} & \ldots & \tilde{v}_{N}
    \end{bmatrix}^{\top},
\end{aligned}
\label{eq:y}
\end{equation}
that is associated with the linear model~\eqref{system} by
\begin{equation}
    y = C x,
    \label{eq:obs output y}
\end{equation}
\begin{equation}\label{eq:obs matrix C}
    \begin{array}{c}
    C=\begin{bmatrix}
    c_{0} & & & \\
    & c_{1} & & \\
    & & \ddots & \\
    & & & c_{N}
    \end{bmatrix}, \quad
    c_{0}=\begin{bmatrix}
    1 & 0 \\
    0 & 1
    \end{bmatrix}, \quad c_{i}=\begin{bmatrix}
    0 & 1
    \end{bmatrix}.
    \end{array}
\end{equation}
Note that other outputs could also be considered with a different matrix $C$ or function $c$.
This includes the case of partial connectivity where the states of non-connected HDVs are unknown, which we plan to investigate in the future.

In order to overcome that only the output $y$ is available instead of the full state $x$, we use state observers to establish an estimate $\hat{x}$ of the state using the output.
State observers are usually constructed as dynamical systems that describe the evolution of the estimated state $\hat{x}$ with corresponding estimated output $\hat{y}=\tilde{c}(\hat{x})$.
For example, control-affine observers are of the form
\begin{equation}
    \dot{\hat{x}} = \tilde{f}(\hat{x},r,y) + \tilde{g}(\hat{x},r,y) u
    \label{eq:observer_nonlin}
\end{equation}
where the dynamics $\tilde{f}$, $\tilde{g}$ are designed so that $\hat{x}$ converges to the true state $x$.
That is, the observer estimation error
\begin{equation}
e=\hat{x}-x
\label{error}
\end{equation}
needs to converge to zero with time.

For simplicity, we focus on observing the internal states of the linear model~\eqref{system} using the linear output~\eqref{eq:obs output y}.
First, we notice that the pair $(A,C)$ is observable for $\aaa_1 - \aaa_2\aaa_3 + \aaa_3^2 \neq 0$, which was proven in~\cite{wang2021leading}.
Under this condition, we design a Luenberger observer~\cite{luenberger1971introduction, callier2012linear}:
\begin{equation}
    \dot{\hat{x}} = A \hat{x} + B u + L(y-\hat{y}) + D r,
    \label{eq:obs}
\end{equation}
where
${\hat{y} = C \hat{x}}$,
${\tilde{f}(\hat{x},r,y) = A \hat{x} + D r + L(y - \hat{y})}$,
${\tilde{g}(\hat{x},r,y) = B}$,
while $L \in \mathbb{R}^{n \times n_y}$ is the constant observer gain to be designed.
Thus, we express the time derivative of~\eqref{error} as
\begin{equation}
    \dot{e}=(A-L C) e,
    \label{error_system}
\end{equation}
where $L$ is chosen using pole placement so that matrix $A-LC$ is Hurwitz and the error dynamics~\eqref{error_system} are exponentially stable.

Then, the observer error of is bounded as follows:
\begin{equation}
\|e(t)\| \leq M(t), \quad \forall t \geq 0,
\label{M(t)}
\end{equation}
where ${\|.\|}$ is Eucledian norm, while ${M(t) \geq 0}$ is a known time-varying function whose values are non-negative and that depends on $x_{0}$, $\hat{x}_{0}$.
Specifically, for the Luenberger observer~\eqref{eq:obs} we have the bound
\begin{equation}
\|e(t)\| \leq M_0 {\rm e}^{-\lambda t}, \quad \forall t \geq 0,
\label{eq:observer_error_bound}
\end{equation}
where $M_0 \geq 0$ depends on ${e_0 = \hat{x}_0 - x_0}$, whereas ${\lambda > 0}$ is associated with the real part of the rightmost eigenvalue of ${A - L C}$.
Note that one could also design observers of the form~\eqref{eq:observer_nonlin} for the nonlinear system~\eqref{system_nonlin} with the general output~\eqref{eq:output_nonlin}, such that they satisfy~\eqref{M(t)} with some $M(t)$; see for example~\cite{wang2021observer}.
However, nonlinear observer design is out of scope of this paper, and instead we focus on safety-critical control.

\section{Safety-critical Traffic Control by Connected Automated Vehicles}
\label{sec:safety}

While a well-designed controller of form \eqref{eq:nominal_nonlin} may achieve string stable behavior for the mixed vehicle chain, it may not guarantee that vehicles maintain safe distances between each other.
Na\"ive stabilization of the mixed traffic could lead to unsafe driving conditions for individual vehicles. Therefore, now we propose the strategy of \emph{safety-critical traffic control (STC)}, in which we synthesize a safety-critical controller for the CAV using control barrier functions and a nominal input that stabilizes traffic.
Firts, two main safety requirements are identified, and various safe spacing policies are discussed.
Then, control barrier functions are used to endow the controlled traffic system with safety for a wide range of nominal controllers.
Finally, the case of output feedback is addressed.
The end result is string stability with formal safety guarantees.

\subsection{Safety Requirements and Safe Spacing Policies}
\label{sec:subsec:safety requirement}
We address two main requirements that are needed to avoid safety-critical failures for the STC system.
The first requirement is the {\em CAV's safety}.
This means that the CAV ($i=0$) maintains a safe distance to the head vehicle ($i=-1$), i.e., specific constraints on the spacing $s_0(t)$ must be satisfied for all time, regardless of the HDVs' driving behavior.
This guarantees the avoidance of collision between the CAV and the head vehicle, including situations, for example, when the head vehicle brakes abruptly due to an unexpected cut-in vehicle or crossing pedestrians.
The second requirement is the {\em HDVs' safety}.
This means that all of the $N$ HDVs ($i\in\{1, \ldots, N\}$) behind the CAV maintain safe distances, i.e., $s_{i}(t)$ satisfy safety constraints, provided that the HDVs drive according to a selected car-following model.
This helps avoid collision for the HDVs due to human mistakes caused by driving fatigue or distractions.

To satisfy these requirements, we provide a set of \emph{safe spacing policies} that constrain the spacings $s_i$ ahead of the CAV and HDVs.
The {\em time headway (TH)} policy guarantees a minimum constant time headway ${\tau>0}$  based on the spacing and speed:
\begin{equation}
    s_i \geq \tau v_i.
    \label{eq:safety:TH}
\end{equation}
Similarly, the {\em time-to-collision (TTC)} policy maintains a minimum time-to-collision ${\tau>0}$ based on the relative speed:
\begin{equation}
    s_i \geq \tau (v_i-v_{i-1}).
    \label{eq:safety:TTC}
\end{equation}
The {\em stopping distance headway (SDH)} policy guarantees a minimum stopping distance while incorporating the follower vehicle's braking limit $a_{i,\min}$ in addition to the time-to-collision~\cite{Treiberbook}:
\begin{equation}
    s_i \geq \tau (v_i-v_{i-1}) + \frac{(v_i-v_{i-1})^2}{2 \vert a_{i,\min}\vert }.
    \label{eq:safety:SDH}
\end{equation}
Finally, the {\em distance headway (DH)} policy guarantees a minimum constant spacing ${s_{i,\min}}$ between vehicles: $s_i \geq s_{i,\min}$.
However, this choice disregards the speed of the vehicles, hence we rather focus on the first three spacing policies.
Similarly, for the linearized system \eqref{system}, we can write the safe spacing policies as
\begin{align}
    \tilde s_i & \geq \tau (\tilde v_i + v^\star) - s_i^\star, \label{lth}\\
      \tilde  s_i& \geq\tau( \tilde v_i- \tilde v_{i-1}) - s_i^\star,\label{lttc}\\
    \tilde  s_i&\geq\tau (\tilde v_i-\tilde v_{i-1}) + \frac{(\tilde v_i-\tilde v_{i-1})^2}{2|a_{i,\min}|} - s_i^\star.\label{lsdh}
\end{align}
Next, we will introduce control barrier functions to guarantee the satisfaction of these safe spacing policies in the closed-loop system.    
Furthermore, we will compare the effectiveness of these policies in safety enforcement in mixed traffic control.

The above safe spacing criteria can be stated equivalently as $x \in \mathcal{C}_i$, i.e., the state $x$ remains inside a safe set $\mathcal{C}_{i}$, given by
\begin{equation}
    \mathcal{C}_{i}=\left\{x \in \mathbb{R}^{n}: h_{i}(x) \geq 0\right\},
    \label{eq:safeset_LCC}
\end{equation}
where the different candidates for function $h_i(x)$ are defined by
\begin{align}
    h_{i}^{\rm TH}(x) & =  s_i -\tau  v_i, \label{eq:h_TH} \\
    h_{i}^{\rm TTC}(x) & =  s_i - \tau (v_i-  v_{i-1}), \label{eq:h_TTC} \\
    h_{i}^{\rm SDH}(x) & =  s_i- \tau ( v_i- v_{i-1}) - \frac{( v_i- v_{i-1})^2}{2|a_{i,\min}|}, \label{eq:h_SDH}
\end{align}
and by the corresponding expressions with tilde.

In order to maintain safety, we seek to achieve that if the traffic system is safe initially, then the CAV's controller prevents safety violations for all time.
That is, for any initial condition ${x_0\in \mathcal{C}_{i}}$  the state of the closed-loop system must satisfy ${x(t)\in \mathcal{C}_{i}}$ for all ${t\geq0}$ and for all ${i\in \{0,\ldots,N\}}$.
In other words, the controller must render the set $\mathcal{C} = \bigcap_{i=0}^{N} \mathcal{C}_{i}$ forward invariant under the closed-loop dynamics.
We achieve this goal by using control barrier functions.

\begin{figure}[t!]
    \centering
    \includegraphics[width=18cm]{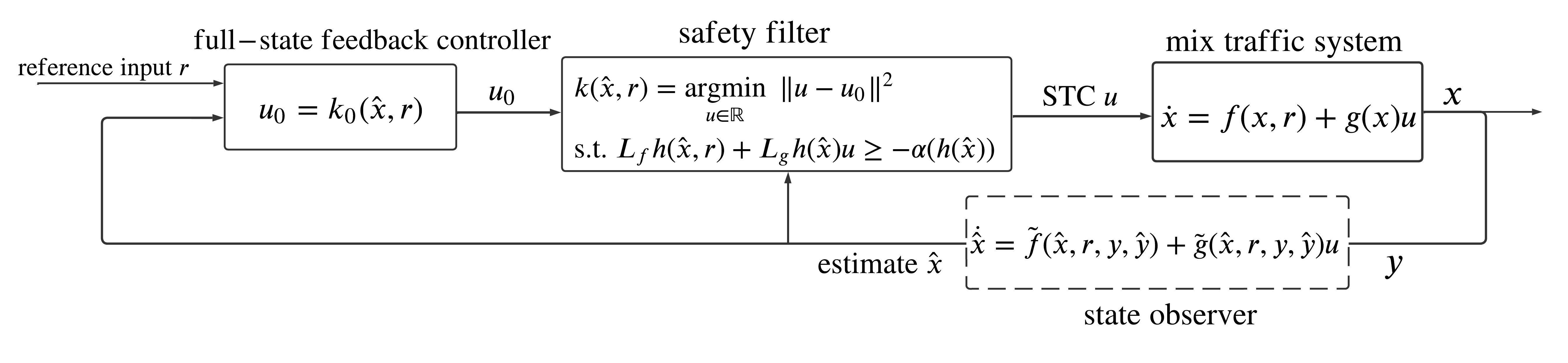}
    \caption{Block diagram of full-state feedback safety-critical traffic control (STC) and observer-based output feedback STC.
    STC modifies a nominal controller that stabilizes mixed traffic but is not necessarily safe to a safety-critical controller by using a safety filter that includes control barrier functions (CBFs).
    The dashed block represent the observer design for state estimation---this block can be omitted in case of full state feedback but is included in case of output feedback.
    }
    \label{diagram}
\end{figure}

\subsection{Guaranteed Safety via Control Barrier Functions and Quadratic Programs}
\label{cbf_subsection}
To design a safety-critical controller, we revisit the theory of control barrier functions.
We consider affine control systems:
\begin{equation}
    \dot x=f(x,r)+g(x)u,
    \label{affine}
\end{equation}
where $x \in \mathbb{R}^{n}$, $r \in \mathbb{R}^{l}$ and $u \in \mathbb{R}^{m}$, while $f: \mathbb{R}^{n} \times \mathbb{R}^{l} \to \mathbb{R}^{n}$ and $g: \mathbb{R}^{n} \to \mathbb{R}^{n\times m}$ are locally Lipschitz continuous.
This may represent the nonlinear mixed traffic system~\eqref{system_nonlin}, or the linear system \eqref{system} by the substitution $f(x,r)=A x + D r$, $g(x) = B$.
We seek to keep the state $x$ inside a safe set $\mathcal{C} \subset \mathbb{R}^{n}$, that is defined by a continuously differentiable function $h:\mathbb{R}^n \to \mathbb{R}$ as
\begin{equation}
    \mathcal{C}=\left\{x \in \mathbb{R}^{n}: h(x) \geq 0\right\},
    \label{eq:safeset}
\end{equation}
cf.~\eqref{eq:safeset_LCC}, where $h$ satisfies ${h(x) = 0 \implies \frac{\partial h}{\partial x} (x) \neq 0}$.
We achieve this goal via control barrier functions.
\begin{definition}[control barrier function\cite{ames2014control}]
Function $h$ is a \emph{control barrier function (CBF)} for system \eqref{affine} if there exists an extended class-$\mathcal{K}_{\infty}$ function\footnote{A continuous function $\classK: \mathbb{R} \to \mathbb{R}$ belongs to extended class-$\mathcal{K}_{\infty}$ if it is strictly monotonically increasing, $\classK(0)=0$ and $\lim_{r \to \pm \infty} \classK(r) = \pm \infty$.}  $\classK$ such that for all ${x\in \mathbb{R}^n}$:
\label{cbf}
\begin{equation}
    \sup_{u \in \mathbb{R}^m}\big[\underbrace{L_{f} h(x,r)+L_{g} h(x) u}_{\dot{h}(x,r,u)}\big] >-\classK(h(x)),
    \label{eq:CBF_condition}
\end{equation}
where
${L_{f} h(x,r) = \frac{\partial h}{\partial x}(x) f(x,r)}$
and
${L_{g} h(x) = \frac{\partial h}{\partial x}(x) g(x)}$.
\end{definition}

The expression in \eqref{eq:CBF_condition} inside the brackets is the first time derivative of $h$ along the system~\eqref{affine}.
Note that condition \eqref{eq:CBF_condition}  sufficiently holds if $L_{g} h(x) \neq 0$ for all ${x \in \mathbb{R}^n}$.
In such a case, the first derivative of $h$ is affected by the input $u$ regardless of the state $x$---this is termed as $h$ has \emph{relative degree one}. 
Given a CBF, the following result on safety was established by \cite{ames2014control}.

\begin{theorem}[guaranteed safety via CBF\cite{ames2014control}] \label{thm:CBF}
\textit{
If $h$ is a CBF for system \eqref{affine}, then any locally Lipschitz continuous controller ${k : \mathbb{R}^n \times \mathbb{R}^l \to \mathbb{R}}$, ${u = k(x,r)}$, that satisfies
\begin{equation}
    L_{f} h(x,r)+L_{g} h(x) u \geq -\classK(h(x)),
    \label{eq:safety_condition}
\end{equation}
renders the set $\mathcal{C}$ in~\eqref{eq:safeset} forward invariant (safe) under the closed-loop dynamics, such that
${x_0 \in \mathcal{C} \implies x(t) \in \mathcal{C}}$, ${\forall t \geq 0}$.
}
\end{theorem}

As such, controllers that satisfy~\eqref{eq:safety_condition} provide provable safety guarantees, while condition~\eqref{eq:CBF_condition} ensures that such controllers exist.
We use~\eqref{eq:safety_condition} for safety-critical control design.
Suppose that there exists a feedback controller, ${u_0 = k_0(x,r)}$, that is stabilizing but may not necessarily be safe.
Then, the input $u_0$ can be modified in a minimally invasive fashion to an input $u$ that satisfies \eqref{eq:safety_condition} by using the following quadratic program (QP) based controller design~\cite{ames2019control}:
\begin{align}
\begin{split}
    k(x,r) = \underset{u \in \mathbb{R}}{\operatorname{argmin}}& \;  \|u-u_0\|^{2}\\
   \text{s.t.} &\;  L_{f} h(x,r)+L_{g} h(x) u \geq -\classK(h(x)).
\end{split}
\label{eq:CBF-QP}
\end{align}
Controller $u = k(x,r)$ is safety-critical, and ensures safety with minimum deviation from the nominal controller $u_0=k_0(x,r)$. 
Note that~\eqref{eq:CBF-QP} is often called as {\em safety filter}.
The role of the safety filter in STC is illustrated in Fig.~\ref{diagram} and discussed next.

\subsection{Safety-critical Traffic Control}

Now we implement CBFs on the mixed traffic system and synthesize a STC strategy via a quadratic program.
Specifically, we rely on the CBF candidates $h_i$ in~\eqref{eq:h_TH}-\eqref{eq:h_SDH}.
To apply the CBF theory, however, first we must show that these candidates are indeed CBFs and satisfy~\eqref{eq:CBF_condition}.
To investigate condition~\eqref{eq:CBF_condition}, we consider the TH policy~\eqref{eq:h_TH} as example, although the conclusions below hold for the other spacing policies in~\eqref{eq:h_TTC}-\eqref{eq:h_SDH}, too.
Function $h_0$, that is associated with the CAV's safety, satisfies ${L_g h_0(x) = - \tau \neq 0}$.
Hence it has relative degree one and it is a valid CBF that satisfies~\eqref{eq:CBF_condition}.
This results from the fact that the CAV's safety is directly affected by the CAV's control input.
However, the remaining functions $h_i$ (${i \in \{1, \ldots, N\}}$) are associated with the HDVs' safety, and they are not directly affected by the CAV's controller.
Instead, ${L_g h_i(x) = 0}$ and~\eqref{eq:CBF_condition} does not hold for all $x \in \mathbb{R}^n$, hence $h_i$ are not valid CBFs.
In fact, the control input $u$ shows up in the $(i+1)$-st time derivative of $h_i$, making it relative degree $i+1$.
Although there exist more complex control designs for high relative degree scenarios~\cite{nguyen2016exponential, ames2019control}, we seek to avoid this complexity, especially since the relative degree of $h_i$ depends on the number of vehicles.

Instead, we propose to construct CBFs with a novel method that is compatible with all safe spacing policies.
Specifically,
we modify the CBF candidates $h_i$ to
\begin{equation}
    \bar h_i(x) = h_i(x)-h_0(x),
\label{eq:reduced_cbf}
\end{equation}
that have relative degree one irrespective of the number of vehicles.
It should be noted that $\bar h_i(x)\geq0$ and $h_0(x)\geq 0$ are sufficient conditions for $h_i(x)\geq0$.
Thus, $\bar{h}_i$ are valid CBFs that allow the satisfaction of the safe spacing policies given by $h_i$.

The CBFs $h_0$ and $\bar{h}_i$ lead to the following constraints on $u$ to satisfy the requirements of the CAV's safety and HDVs' safety.
The constraint that guarantees safe spacing between the head vehicle and the CAV becomes
\begin{equation}
    L_{f} h_{0}(x,r) + L_{g} h_{0}(x) u + \gamma_0 h_0(x)\geq 0,
    \label{CAV_cbf}
\end{equation}
where, for simplicity, we chose the extended class-$\mathcal{K}_{\infty}$ function $\classK$ in~\eqref{eq:safety_condition} to be linear, i.e., $\classK(h_0) = \gamma_0 h_0$ with a tunable constant $\gamma_0>0$.
Similarly, the constraints that yield safe spacings ahead of the HDVs are given by
\begin{equation}
    L_{f} \bar h_{i}(x,r) + L_{g} \bar h_{i}(x) u + \gamma_i \bar h_i(x) + \sigma_i \geq 0,
    \label{HDV_cbf}
\end{equation}
where $\gamma_i>0$, ${i \in \{1, \ldots, N\}}$, while the role of $\sigma_i \geq 0$ is explained below.

We remark that while the CBF condition \eqref{eq:CBF_condition} ensures that there exists a control input $u$ satisfying the safety constraint \eqref{eq:safety_condition}, there is no such existence guarantee if multiple constraints like \eqref{CAV_cbf}-\eqref{HDV_cbf} are enforced simultaneously.
In other words, the CAV's and HDVs' safety constraints \eqref{CAV_cbf}-\eqref{HDV_cbf} may conflict in some scenarios.
In such cases, we prioritize the CAV's safety over the HDVs' safety for four reasons:
(I) the CAV's controller is primarily responsible for its own safety;
(II) if a collision happened at the CAV, the following HDVs would also be affected;
(III) in practice, the HDVs do not necessarily drive according to the selected human driver model; and
(IV) if the CAV's safety constraint ${h_0(x) \geq 0}$ is not satisfied, ${\bar{h}_i(x)\geq 0}$ no longer implies ${h_i(x)\geq 0}$, i.e., the HDVs' safety may be compromised. Therefore, we resolve the conflict between multiple safety constraints by introducing the relaxation variables $\sigma_i$ into the HDVs' safety constraints.
This makes~\eqref{HDV_cbf} to be soft constraint while keeping~\eqref{CAV_cbf} as hard constraint.

Finally, we incorporate the safety constraints~\eqref{CAV_cbf}-\eqref{HDV_cbf} and the traffic controller~\eqref{eq:nominal_nonlin} that yields string stable behavior into an optimization problem to establish the proposed STC.
Specifically, while enforcing~\eqref{CAV_cbf}-\eqref{HDV_cbf}, we minimize the deviation of the control input $u$ from the nominal stabilizing input $u_0$.
This leads to the quadratic program (QP):
\begin{align}
\begin{split}
    k(x,r)=\underset{u \in \mathbb{R}, \sigma_i \geq 0}{\operatorname{argmin}} \; & |u-u_0|^{2} + \sum_{i=1}^N p_i\sigma_i^2 \\
    \text{s.t.} \;
    & L_{f} h_{0}(x,r) + L_{g} h_{0}(x) u + \gamma_{0} h_{0}(x) \geq 0, \\
    & L_{f} \bar h_{1}(x,r) + L_{g} \bar h_{1} (x) u + \gamma_{1} \bar h_{1}(x) + \sigma_{1} \geq 0, \\
    & \qquad \vdots \\
    & L_{f} \bar h_{N} (x,r) + L_{g} \bar h_{N}(x) u + \gamma_{N} \bar h_{N}(x) + \sigma_{N} \geq 0,
\end{split}
\label{eq:QP}
\end{align}
cf.~\eqref{eq:CBF-QP},
where ${p_i \gg 1}$ are penalties for relaxing the HDVs' safety constraints.
Note that when there is conflict between the safety of mixed traffic and the pursuit of string stability, the QP enforces safety and sacrifices string stability at minimum cost.

\subsection{Safety with Observer-based Control Barrier Functions}

The proposed STC strategy~\eqref{eq:nominal_nonlin}-\eqref{eq:QP} relies on full state feedback.
In practice, however, some state information may not be available to the CAV.
This motivates us to extend the STC design to output feedback.
We propose to use state observers for state estimation, that allows us to execute safety-critical controllers using an estimate of the unknown states; see Fig.~\ref{diagram}.
Below we provide a set of simple conditions under which formal safety guarantees are maintained despite observer estimation errors.
These results follow the constructions in~\cite{Agrawal2023}, where a comprehensive safety-critical control framework with observers was established.
For another approach on the combination of observers and CBFs, please also see~\cite{wang2021observer}.

Consider an observer~\eqref{eq:observer_nonlin} that satisfies \eqref{M(t)}.
Let the CBF $h$ have the Lipschitz coefficient $\mathcal{L}>0$, satisfying: 
\begin{equation}
|h(\hat{x})-h(x)| \leq \mathcal{L}\|\hat{x}-x\|,
\label{eq:h_Lipschitz}
\end{equation}
and for simplicity, let the class-$\mathcal{K}$ function $\alpha$ be linear, $\alpha(r) = \gamma r$.
Furthermore, consider the set
\begin{equation}
    \hat{\mathcal{C}}(t)=\left\{\hat{x} \in \mathbb{R}^{n}: \hat{h}(\hat{x},t) \geq 0\right\},
    \label{eq:safeset_timevarying}
\end{equation}
\begin{equation}
    \hat{h}(\hat{x},t) = h(\hat{x}) - \mathcal{L} M(t),
    \label{eq:CBF_timevarying}
\end{equation}
where $\hat{h}(\hat{x},t)$ represents the worst-case value of $h(x)$ that may correspond to an estimated state $\hat{x}$ given the estimation error bound $M(t)$.
Then, we have the following result on safety.

\begin{theorem}[guaranteed safety with observer\cite{Agrawal2023}] \label{thm:CBF_observer}
\textit{
Consider the system \eqref{affine}, the safe set \eqref{eq:safeset}, and an observer~\eqref{eq:observer_nonlin} that satisfies~\eqref{M(t)}.
If $h$ is a CBF for system \eqref{eq:observer_nonlin} with Lipschitz coefficient $\mathcal{L}$, then any locally Lipschitz continuous controller ${k : \mathbb{R}^n \times \mathbb{R}^l \times \mathbb{R}^{n_y} \to \mathbb{R}}$, ${u = k(\hat{x},r,y)}$, that satisfies
\begin{equation}
    L_{\tilde{f}} h(\hat{x},r,y) + L_{\tilde{g}} h(\hat{x},r,y) u \geq - \gamma h(\hat{x}) + \mathcal{L} \big( \dot{M}(t) + \gamma M(t) \big),
    \label{eq:safety_condition_observer_robust}
\end{equation}
renders set $\hat{\mathcal{C}}(t)$ in~\eqref{eq:safeset_timevarying}-\eqref{eq:CBF_timevarying} forward invariant under the closed-loop dynamics, such that
${\hat{x}_0 \in \hat{\mathcal{C}}(0) \implies \hat{x}(t) \in \hat{\mathcal{C}}(t)}$, ${\forall t \geq 0}$.
Furthermore, ${\hat{x}(t) \in \hat{\mathcal{C}}(t) \implies x(t) \in \mathcal{C}}$ holds, that ultimately leads to safe behavior.
}
\end{theorem}

\begin{proof}
First, we express the total time derivative of $\hat{h}$ defined in~\eqref{eq:CBF_timevarying} as
\begin{equation}
    \dot{\hat{h}}(\hat{x},t,r,y,u) = L_{\tilde{f}} h(\hat{x},r,y) + L_{\tilde{g}} h(\hat{x},r,y) u - \mathcal{L} \dot{M}(t).
    \label{eq:observer_CBF_safety_proof_1}
\end{equation}
With this,~\eqref{eq:CBF_timevarying}-\eqref{eq:safety_condition_observer_robust} lead to
\begin{equation}
    \dot{\hat{h}}(\hat{x},t,r,y,u) \geq -\gamma \hat{h}(x,t).
    \label{eq:observer_CBF_safety_proof_2}
\end{equation}
Therefore, we can apply the time-varying counterpart~\cite{Xu2018} of Theorem~\ref{thm:CBF} that yields the invariance of $\hat{\mathcal{C}}(t)$.
Finally, ${\hat{x}(t) \in \hat{\mathcal{C}}(t)} \implies x(t) \in \mathcal{C}$ is equivalent to ${\hat{h}(\hat{x}(t),t) \geq 0 \implies h(x(t)) \geq 0}$, which can be proven by
\begin{align}
\begin{split}
    \hat{h}(\hat{x},t)
    & = h(\hat{x}) - \mathcal{L} M(t) \\
    & \leq h(\hat{x}) - \mathcal{L} \|\hat{x}-x\| \\
    & \leq h(\hat{x}) - |h(\hat{x})-h(x)| \\
    & \leq h(x),
\end{split}
\end{align}
where we used~\eqref{eq:CBF_timevarying},~\eqref{M(t)} and~\eqref{eq:h_Lipschitz}.
\end{proof}

\begin{remark}
For output feedback, the safety condition modifies from~\eqref{eq:safety_condition} to~\eqref{eq:safety_condition_observer_robust}, that includes the extra robustifying term ${\mathcal{L} \big( \dot{M}(t) + \gamma M(t) \big)}$ to account for observer estimation errors.
\cite{Agrawal2023} highlighted that if the observer satisfies ${\dot{M}(t) \leq - \gamma M(t)}$, then
\begin{equation}
    L_{\tilde{f}} h(\hat{x},r,y) + L_{\tilde{g}} h(\hat{x},r,y) u \geq - \gamma h(\hat{x})
    \label{eq:safety_condition_observer}
\end{equation}
is a sufficient condition for~\eqref{eq:safety_condition_observer_robust}, i.e., the extra robustifying term can be dropped.
For example, the Luenberger observer~\eqref{eq:obs} associated with the linear system~\eqref{system} and error bound~\eqref{eq:observer_error_bound} satisfies ${\dot{M}(t) \leq - \gamma M(t)}$ if $\lambda \geq \gamma$.
This means that, if the observer is fast enough (characterized by $\lambda$) compared to how fast the system is allowed to approach the boundary of the safe set (expressed by $\gamma$),~\eqref{eq:safety_condition_observer} guarantees safety.
This idea first appeared in~\cite{molnar2021model, Singletary2022}, where safety with exponentially stable tracking controllers was established similarly to our case of safety with exponentially stable observers.
Furthermore, notice that condition ${\hat{x}_0 \in \hat{\mathcal{C}}(0)}$ must also hold for safety, that is stricter than ${x_0 \in \mathcal{C}}$.
Accordingly, the initial state must be located farther inside the safe set (${h(x_0) \gg 0}$) in case of significant initial state estimation errors.
Under such conditions, the observer converges to an accurate state estimate by the time the system may get into a safety-critical scenario, and safety is maintained.
\end{remark}

In the context of STC, Theorem~\ref{thm:CBF_observer} leads to the following controller for the case of output feedback:
\begin{align}
\begin{split}
    k(\hat{x},r,y)=\underset{u \in \mathbb{R}, \sigma_i \geq 0}{\operatorname{argmin}} \; & |u-u_0|^{2} + \sum_{i=1}^N p_i\sigma_i^2 \\
    \text{s.t.} \;
    & L_{\tilde{f}} h_0(\hat{x},r,y) + L_{\tilde{g}} h_0(\hat{x},r,y) u + \gamma_{0} h_{0}(\hat{x}) - \mathcal{L} \big( \dot{M}(t) + \gamma_{0} M(t) \big) \geq 0, \\
    & L_{\tilde{f}} \bar{h}_{1}(\hat{x},r,y) + L_{\tilde{g}} \bar{h}_{1}(\hat{x},r,y) u + \gamma_{1} \bar{h}_{1}(\hat{x}) - \mathcal{L} \big( \dot{M}(t) + \gamma_{1} M(t) \big) + \sigma_{1} \geq 0, \\
    & \qquad \vdots \\
    & L_{\tilde{f}} \bar{h}_{N}(\hat{x},r,y) + L_{\tilde{g}} \bar{h}_{N}(\hat{x},r,y) u + \gamma_{N} \bar{h}_{N}(\hat{x}) - \mathcal{L} \big( \dot{M}(t) + \gamma_{N} M(t) \big) + \sigma_{N} \geq 0,
\end{split}
\label{eq:QP_obs_robust}
\end{align}
analogously to~\eqref{eq:QP}.

\section{Numerical Results}
\label{sec:simulation}

In this section, we carry out numerical simulations to show the effectiveness of the proposed STC framework.
We first describe the simulation setup,
then show our main results which validate that the proposed STC~\eqref{eq:QP} leads to smooth and safe traffic in scenarios where the nominal controller~\eqref{eq:nominal_nonlin} could yield collisions.
Afterwards, we present an exhaustive analysis of the STC performance in various simulation settings.
The parameters used for simulations are given in Table~\ref{tab:parameters} and at the figures. 
Finally, we show the performance of STC in a real-world traffic scenario by using experimental data to drive the simulations.

\begin{table}[t]
    \centering
    \caption{Parameters used for the numerical simulations}
    \label{tab:parameters}
    \begin{tabular}{cccccc}
    \hline
    Vehicle & Variable & Symbol & Value & Unit\\
    \hline
    \multirow{4}{*}{all}
    & equilibrium velocity & $v^\star$ & $20$ & m/s \\
    & equilibrium spacing & $s^\star$ & $20$ & m \\
    & braking limit & $a_{\min}$ & $-7$ & m/s$^2$ \\
    & acceleration limit & $a_{\max}$ & $7$ & m/s$^2$ \\
    \hline
    \multirow{5}{*}{HDVs}
    & speed limit & $v_{\max}$ & $40$ & m/s \\
    & standstill spacing & $h_{\rm st}$ & $5$ & m \\
    & free flow spacing & $h_{\rm go}$ & $35$ & m \\
    & \multirow{2}{*}{model coefficients}
    & $\A$ & $0.6$ & 1/s \\
    & & $\B$ & $0.9$ & 1/s \\
    \hline
    \multirow{7}{*}{CAV}
    & \multirow{4}{*}{gains of nominal controller}
    & $\mu_1$ & $-2$ & 1/s$^2$ \\
    & & $\mu_2$ & $-2$ & 1/s$^2$ \\
    & & $k_1$ & $0.2$ & 1/s \\
    & & $k_2$ & $0.2$ & 1/s \\
    & \multirow{3}{*}{parameters of safety filter}
    & $\tau$ & $1$ & s \\
    & & $\gamma$ & $10$ & 1/s \\
    & & $p$ & $100$ & 1/s$^2$ \\
    \hline
    \end{tabular}
\end{table}

\subsection{Simulation Setup}
\label{sec:subsec:simulation setting}

We simulate an STC setup that includes one head HDV (HHDV), one CAV, and two identical following HDVs (FHDV-1 and FHDV-2).
We consider two risky traffic scenarios that could cause rear-end collision:
\begin{itemize}
    \item {\em Scenario 1}: HHDV suddenly decelerates, which may endanger the safety of all vehicles.
    This could be caused, for example, by an aggressive cut-in of a vehicle from an adjacent lane, or due to an obstacle on the road.
    \item {\em Scenario 2}: FHDV-2 unexpectedly accelerates.
    Importantly, such accelerations of follower HDVs could cause a na\"ive nominal CAV controller to respond with acceleration and collide with the HHDV in front of it.
\end{itemize}

Specifically, we prescribe the following motions for the vehicles.
We assume that initially all the four vehicles in the platoon drive at the equilibrium velocity $v^\star$ and with uniform equilibrium spacing $s^\star$.
In Scenario 1, the HHDV brakes harshly, then accelerates to recover its original speed, and finally cruises at constants speed:
\begin{equation}
    \dot{v}_{-1}(t) =
    \begin{cases}
        -a_{\rm H} & \text{if} \quad t \in [0,t_{\rm H}], \\
        a_{\rm H} & \text{if} \quad t \in (t_{\rm H},2 t_{\rm H}], \\
        0 & \text{otherwise}.
    \end{cases}
    \label{eq:HHDV_accel}
\end{equation}
In Scenario 2, the HHDV drives at constant speed, $\dot{v}_{-1}(t) \equiv 0$.
In both scenarios, the CAV uses the nominal controller \eqref{eq:nominal controller} with feedback gains
that satisfy the head-to-tail sting-stability condition given by \eqref{eq:string stable condition}-\eqref{eq:transfer function gammas}.
This nominal controller will be incorporated into and compared with STC~\eqref{eq:QP}.
In Scenario 1, both FHDVs obey the driver model~\eqref{eq:CF} with specific expression given below in~\eqref{eq:OVM}.
In Scenario 2, the motion of FHDV-2 is modified to accelerate before driving according to the model:
\begin{equation}
    \dot{v}_{2}(t) =
    \begin{cases}
        a_{\rm F} & \text{if} \quad t \in [0,t_{\rm F}], \\
        F_2(s_2, \dot{s}_2, v_2) & \text{otherwise}.
    \end{cases}
    \label{eq:FHDV_accel}
\end{equation}
The behavior of each vehicle in the two safety-critical scenarios is summarized in Table \ref{tab:acc}.

To capture the longitudinal dynamics of HDVs, we use the optimal velocity model (OVM) originating from~\cite{bando1998analysis}:
\begin{equation}
    F_{i}(s_{i},\dot{s}_{i},v_{i})= \A \left(V\left(s_{i}\right)-v_{i}\right)+ \B \dot{s}_{i} \label{eq:OVM}, 
\end{equation}
where $V(s_{i})$ is a piecewise continuous function that represents the spacing-dependent desired velocity of the human drivers.
The coefficient $\A$ characterizes how human drivers accelerate to match their current speed $v_i$ with the desired driving speed $V(s_i)$, whereas coefficient $\B$ represents how human drivers change their acceleration according to the speed difference $\dot {s_i}$ between themselves and their preceding vehicle. The OVM \eqref{eq:OVM} can be linearized as \eqref{eq:CF_lin} with parameters:
\begin{align}
    \aaa_{1} = \A V'(s^{\star}), \quad
    \aaa_{2} = \A+\B, \quad
    \aaa_{3} = \B.
\end{align}
We use a range policy $V(s)$ from \cite{zhang2016motif} as:
\begin{equation}
    V(s)=\left\{
    \begin{array}{ll}
    0, & s \leq s_{\mathrm{st}}, \\
    \frac{v_{\max }}{2}\left(1-\cos \left(\pi \frac{s-s_{\mathrm{st}}}{s_{\mathrm{go}}-s_{\mathrm{st}}}\right)\right), & s_{\mathrm{st}}<s<s_{\mathrm{go}}, \\
   v_{\max }, & s \geq s_{\mathrm{go}},
   \end{array}
   \right.
  \label{eq:Vs}
\end{equation}
where $s_{\rm st}, s_{\rm  go},v_{\rm  max}$ represent standstill spacing, free flow spacing and the speed limit, respectively.

\begin{table}[t]
    \centering
    \caption{Acceleration of vehicles in the two safety-critical traffic scenarios}
    \label{tab:acc}
    \begin{tabular}{c|c|c|c|c}
    \hline
         & HHDV & CAV &FHDV-1 &FHDV-2  \\ \hline
         Scenario 1 &
         \makecell[l]{\eqref{eq:HHDV_accel} with \\
         $a_{\rm H} = 6 \ \mathrm{m/s^2}$ \\
         $t_{\rm H} = 3.3  \ \mathrm{s}$}
         & \makecell[l]{\eqref{eq:nominal controller} for nominal controller\\\eqref{eq:QP} for STC} & \eqref{eq:OVM} &  \eqref{eq:OVM} \\\hline
         Scenario 2 & 0 & \makecell[l]{\eqref{eq:nominal controller} for nominal controller\\\eqref{eq:QP} for STC} & \eqref{eq:OVM} &
         \makecell[l]{\eqref{eq:FHDV_accel}-\eqref{eq:OVM} with \\
         $a_{\rm F} = 6 \ \mathrm{m/s^2}$ \\
         $t_{\rm F} = 2.5 \ \mathrm{s}$} \\
         \hline
    \end{tabular}
    
\end{table}

\subsection{Safety Improvement by STC} \label{sec:subsec:sim trajectory}

Now we demonstrate by our main simulation results that the proposed STC improves the safety of mixed traffic.
We highlight that while a nominal controller may lead to unsafe behavior when responding to sudden accelerations of human drivers, the STC maintains safety.
In the simulations, we use the SDH safe spacing policy with CBF candidate \eqref{eq:h_SDH} and parameters in Table~\ref{tab:parameters}.
First, we consider full state feedback control in the two risky traffic situations described above (Scenario 1 and Scenario 2). 
We evaluate safety (i.e., whether the CBF candidates $h_i$ and the spacings $s_i$ stay positive over time) and head-to-tail string stability (i.e., whether the speed perturbations of the tail vehicle are smaller than those of the head vehicle).

Figure~\ref{fig:sim:trajectory case1} shows simulation results for Scenario 1, including the spacing $s_i$, velocity $v_i$, acceleration $a_i$ and CBF candidate $h_i$ of each vehicle for the cases of using the nominal controller (top) and the proposed STC (bottom).
When the HHDV suddenly reduces its speed, the nominal controller reduces the CAV's speed significantly less than the HHDV does in order to achieve head-to-tail string stability.
However, the CAV's deceleration ends up being too small, causing a collision between HHDV and the CAV.
This is indicated by ${s_0<0}$ along the red curve in Fig. \ref{fig:sim:trajectory case1}. 
In contrast, with the proposed STC the CAV decelerates more and avoids collision with formal guarantees of safety.
We remark that $h_0$ becomes slightly negative, which is caused by designing STC based on the linearized dynamics \eqref{system} while the FHDVs drive according to the nonlinear model \eqref{system_nonlin}.
It should be noted that the FHDVs are still collision-free and safe, when their safety was encoded as soft constraints in STC.
Moreover, head-to-tail string stability is still achieved with STC, as FHDV-2 reduces its speed less than the HHDV.

Figure~\ref{fig:sim:trajectory case2} presents simulation results for Scenario 2, by comparing the nominal controller (top) and STC (bottom).
This scenario highlights that if a vehicle at the end of the chain accelerates, like FHDV-2 does, it can force the nominal controller of the CAV to accelerate and compromise the CAV's safety.
By using STC, however, safe behavior is ensured for the CAV by the safety filter that is wrapped around the nominal controller.

\begin{figure}[t]
    \centering
    Nominal Controller \\[6pt]
    \includegraphics[width=0.24\linewidth]{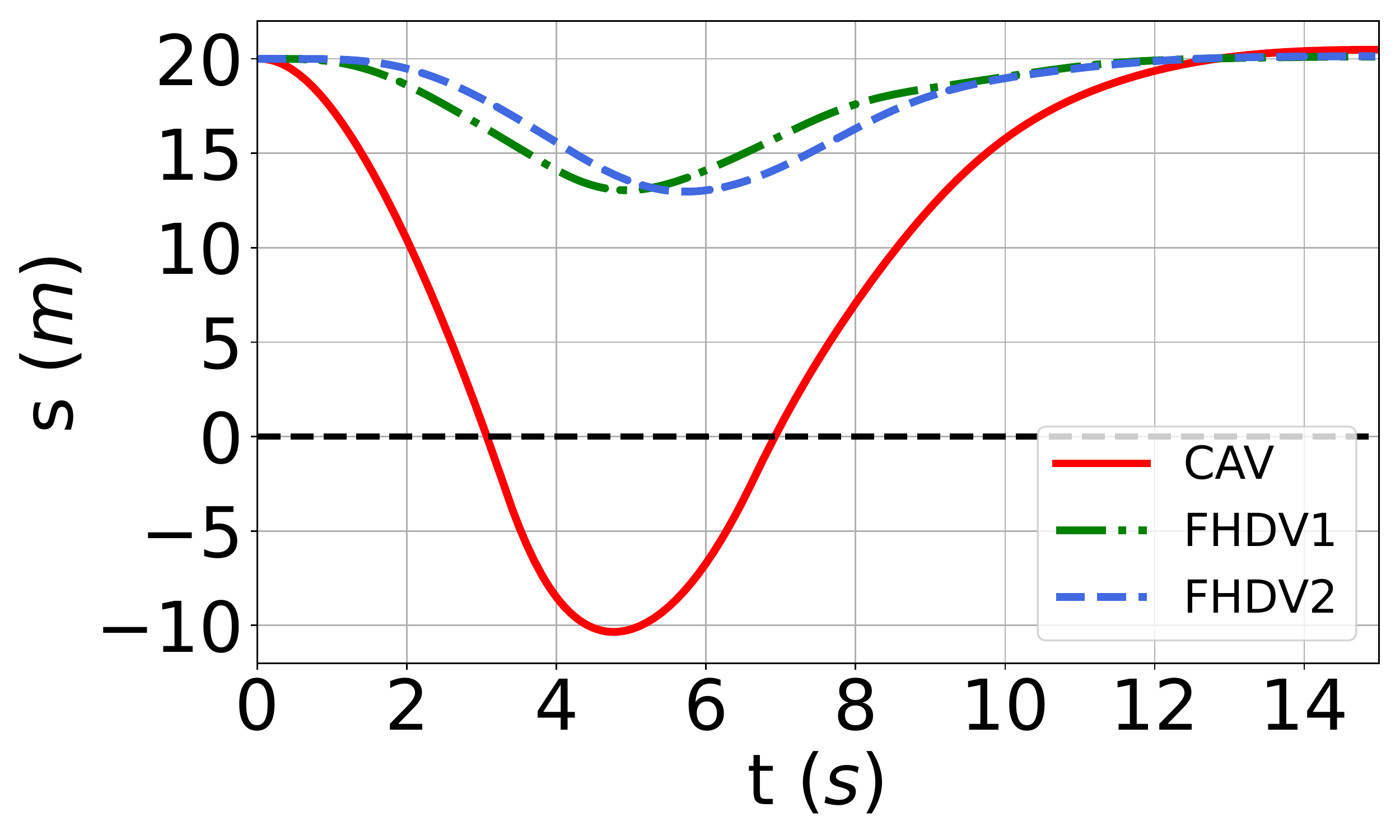}
    \includegraphics[width=0.24\linewidth]{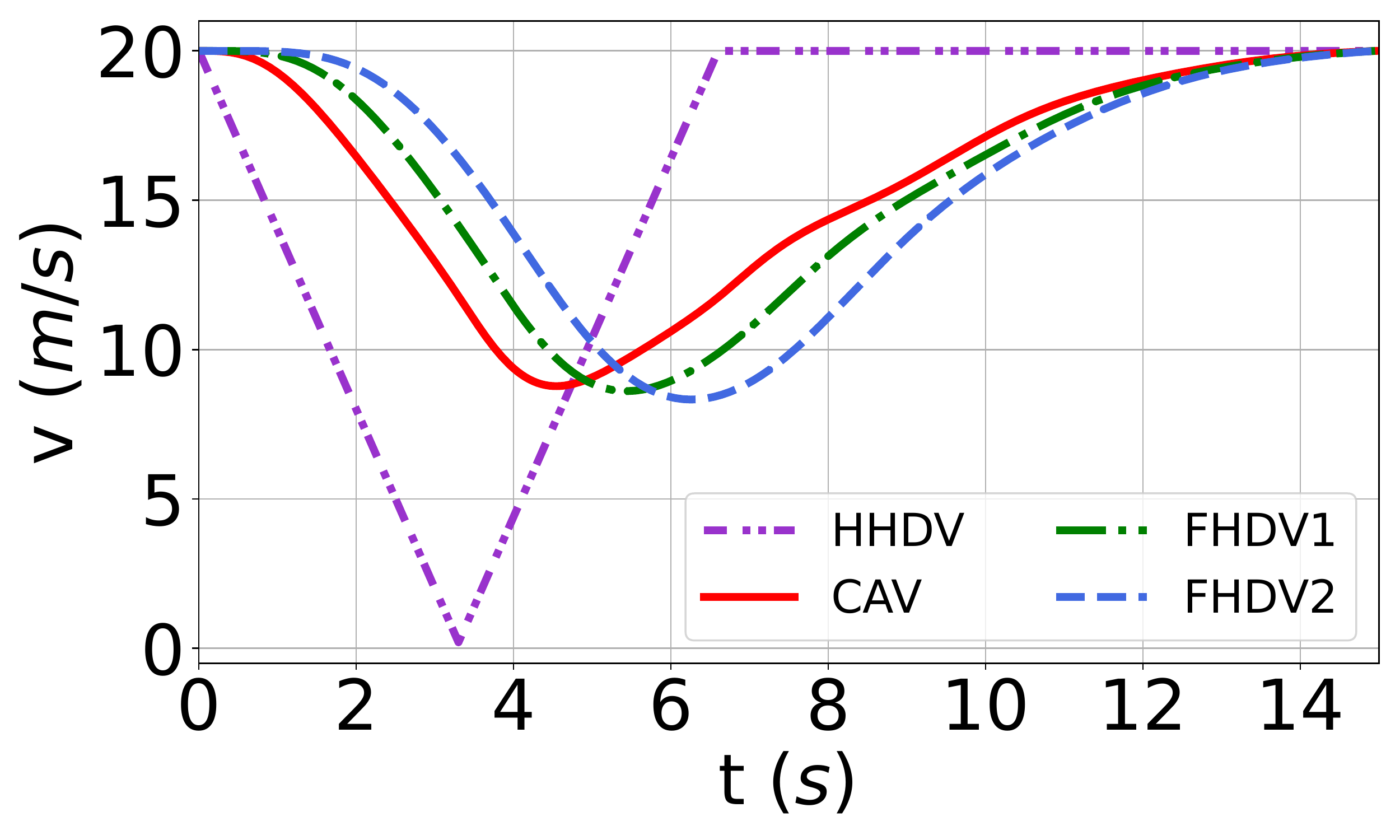}
    \includegraphics[width=0.24\linewidth]{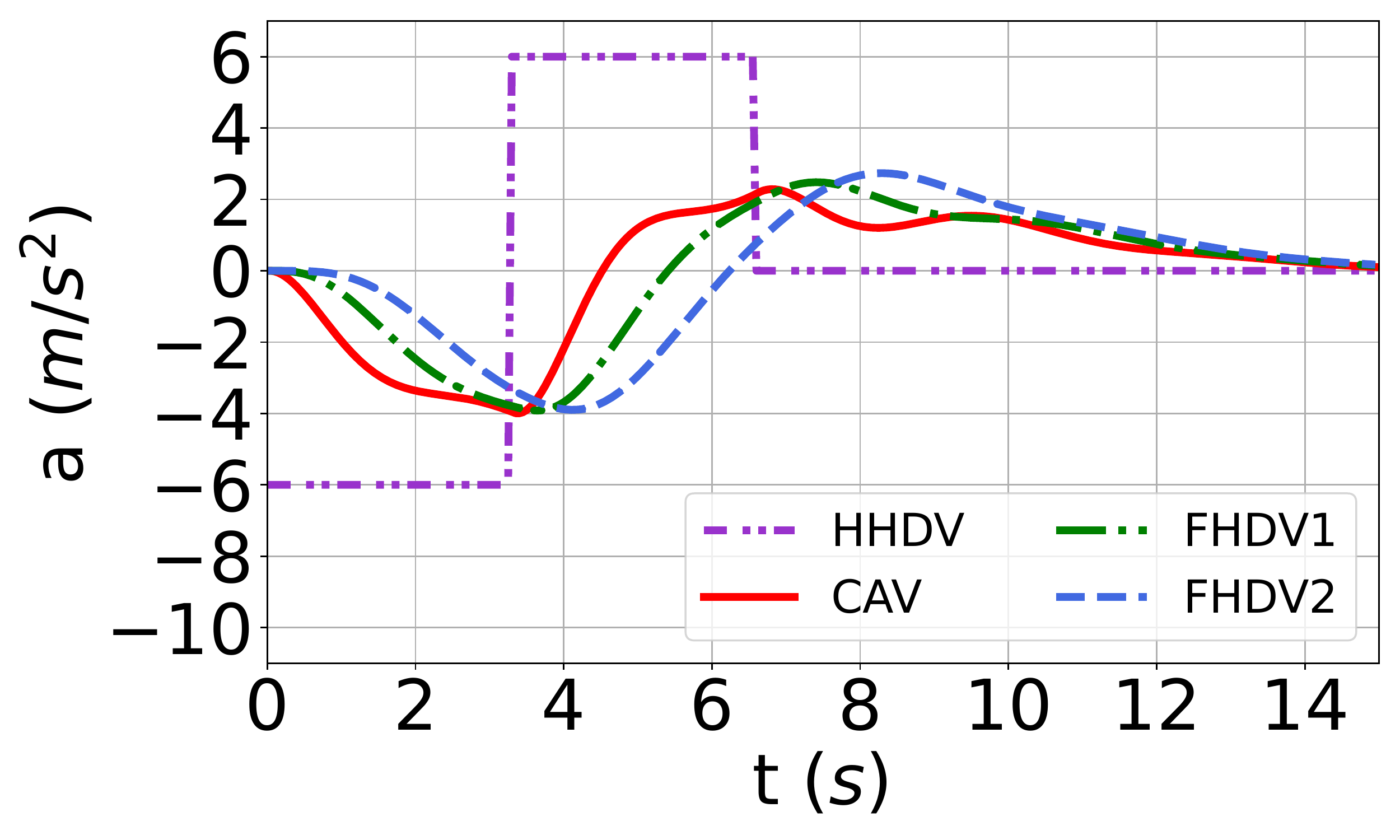}
    \includegraphics[width=0.24\linewidth]{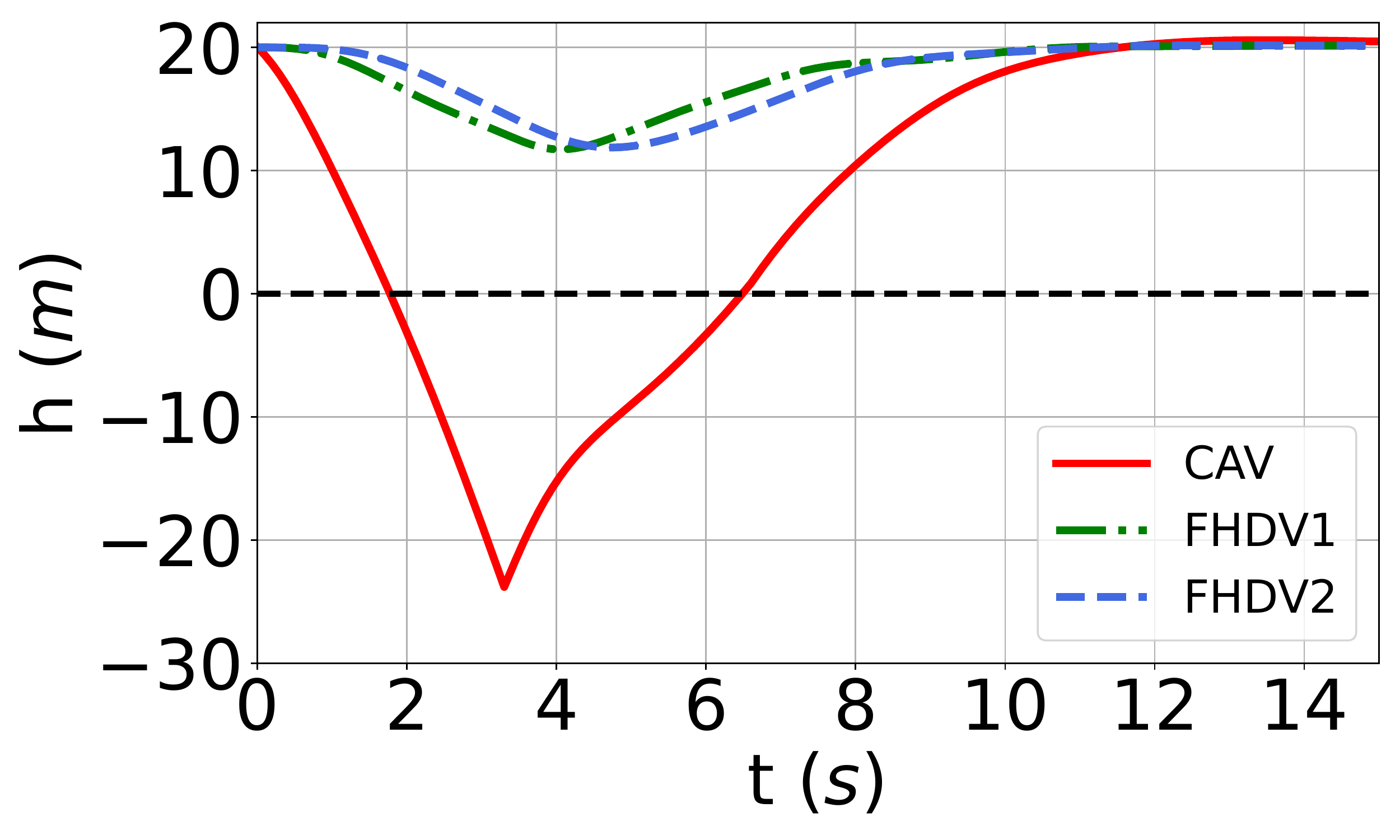}
    \\
    Proposed STC \\[6pt]
    \includegraphics[width=0.24\linewidth]{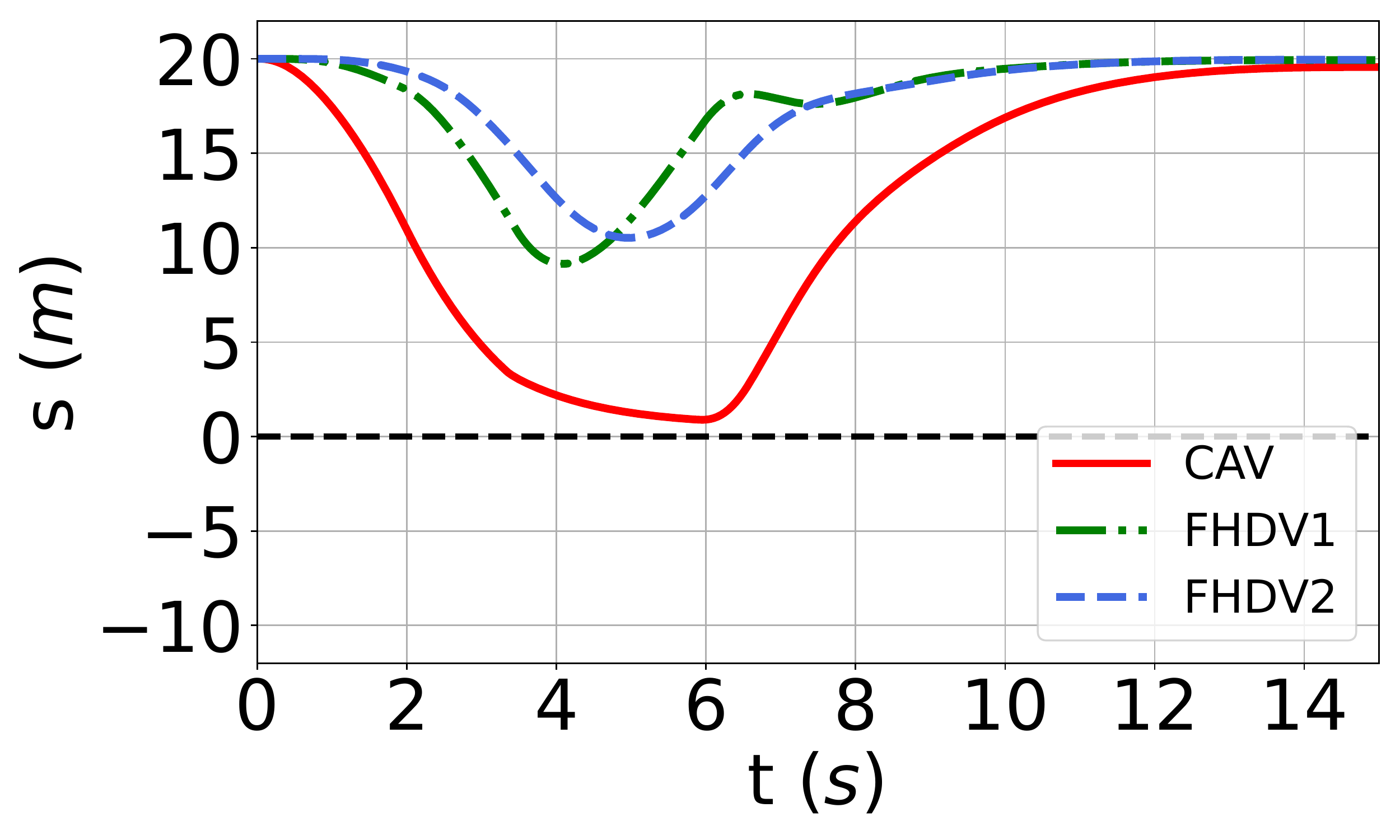}
    \includegraphics[width=0.24\linewidth]{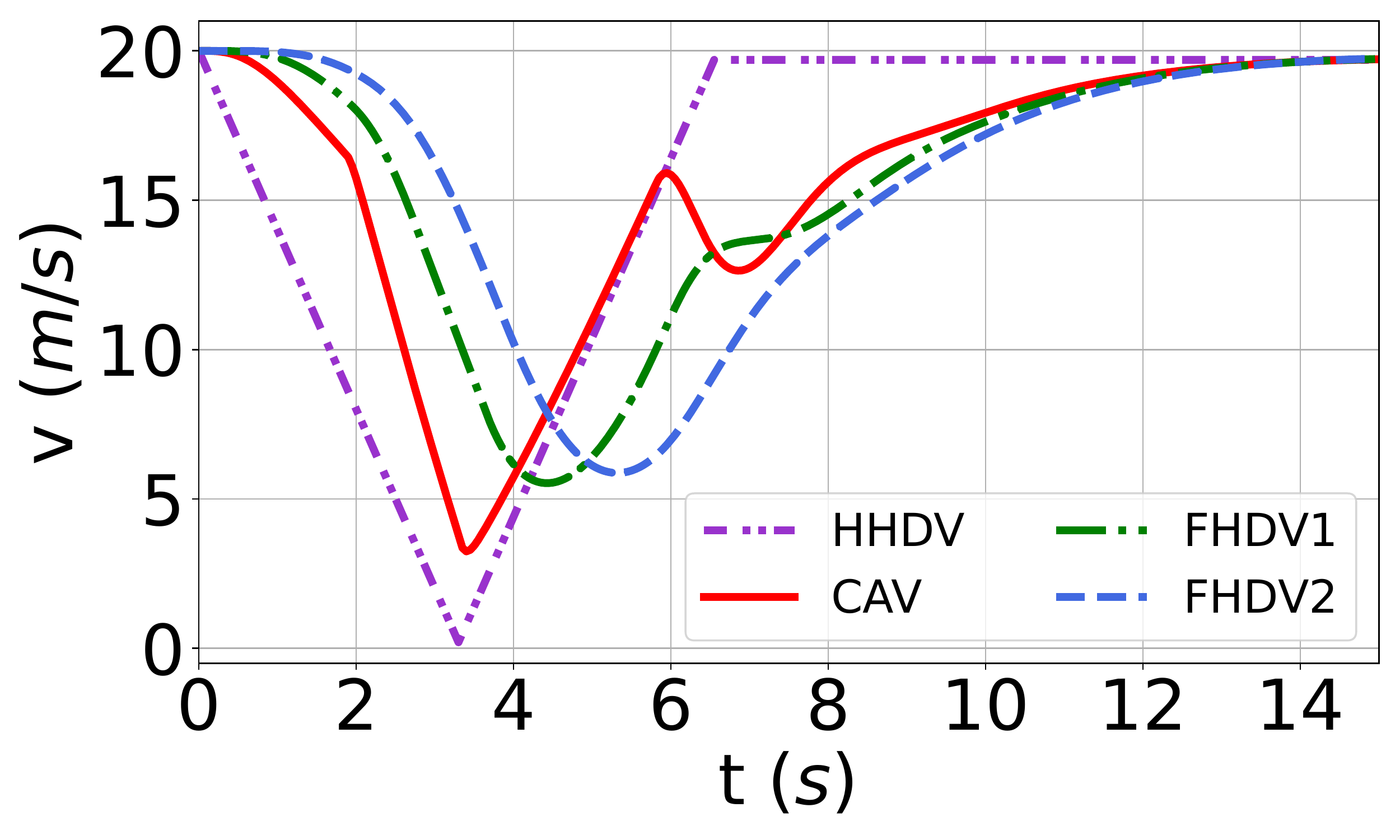}
    \includegraphics[width=0.24\linewidth]{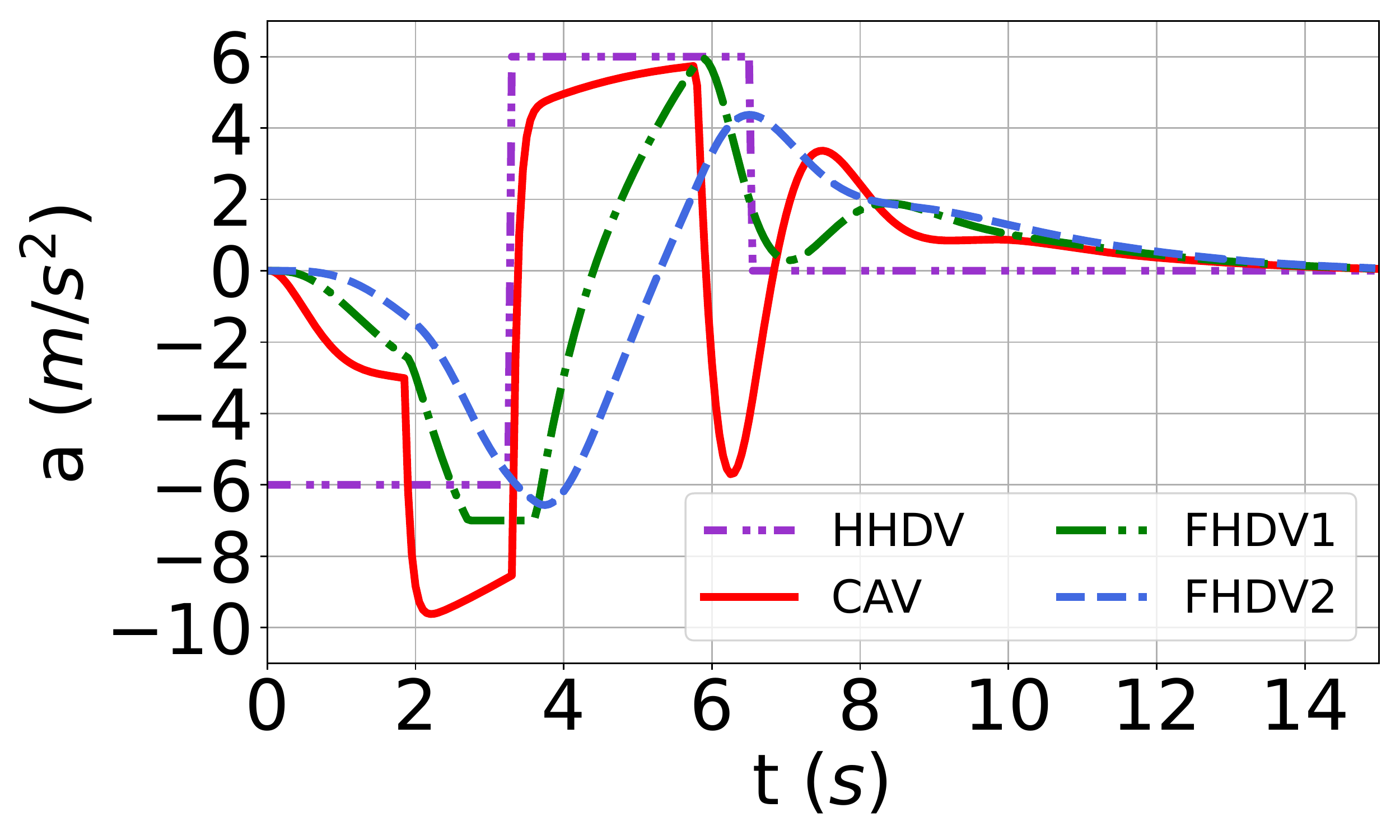}
    \includegraphics[width=0.24\linewidth]{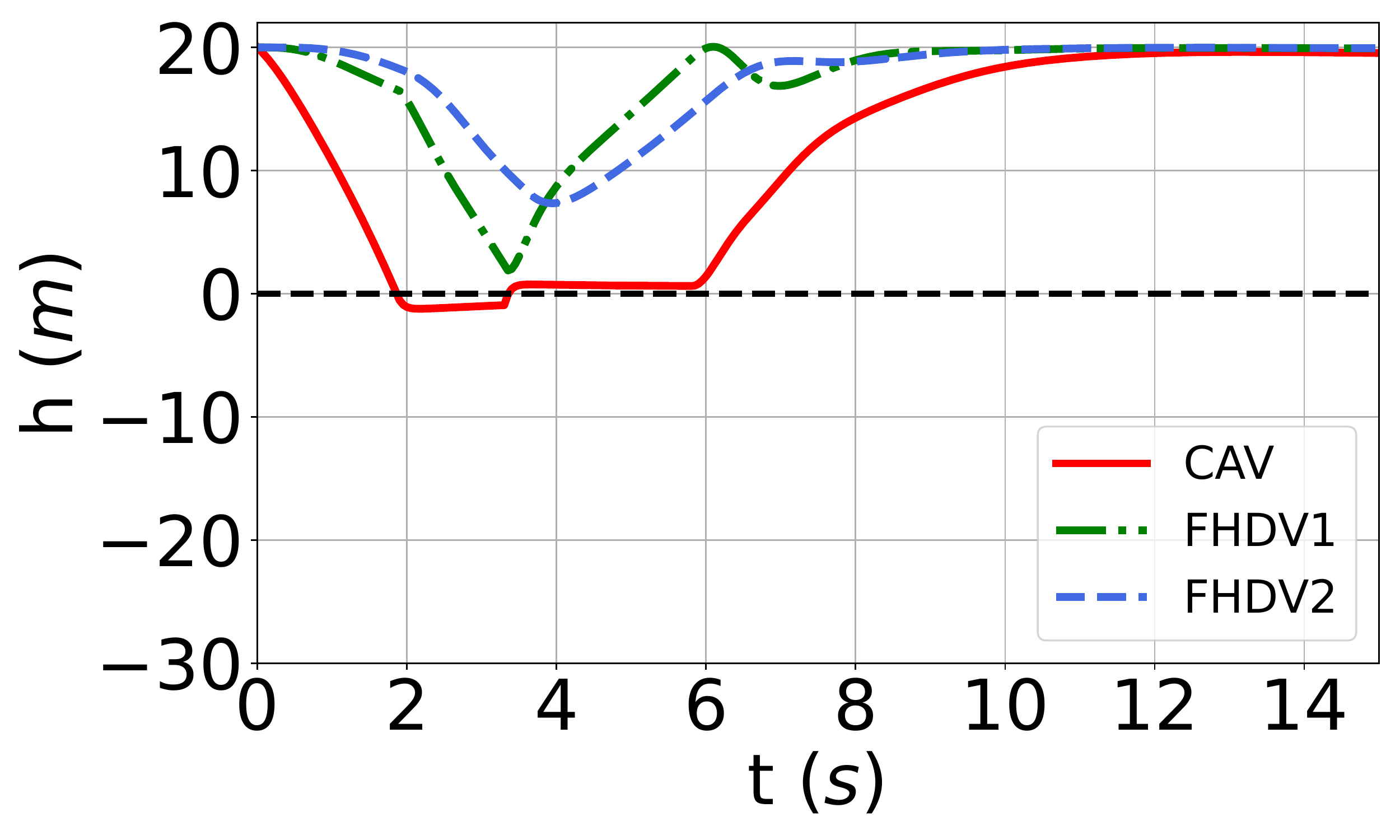}
    \caption{
    Numerical simulation of the proposed safety-critical traffic controller (STC) in Scenario 1.
    The top row shows the behavior of mixed traffic when the CAV uses the nominal stabilizing controller~\eqref{eq:nominal controller}, that achieves head-to-tail string stability (i.e., FHDV-2 reduces its speed less than the HHDV) but yields unsafe behavior (i.e., the CAV collides with the HHDV).
    The bottom row shows the proposed STC~\eqref{eq:QP}, that achieves head-to-tail string stability with provable guarantees of safety.
    }
    \label{fig:sim:trajectory case1}
\end{figure}

\begin{figure}[t]
    \centering
    Nominal Controller \\[6pt]
    \includegraphics[width=0.24\linewidth]{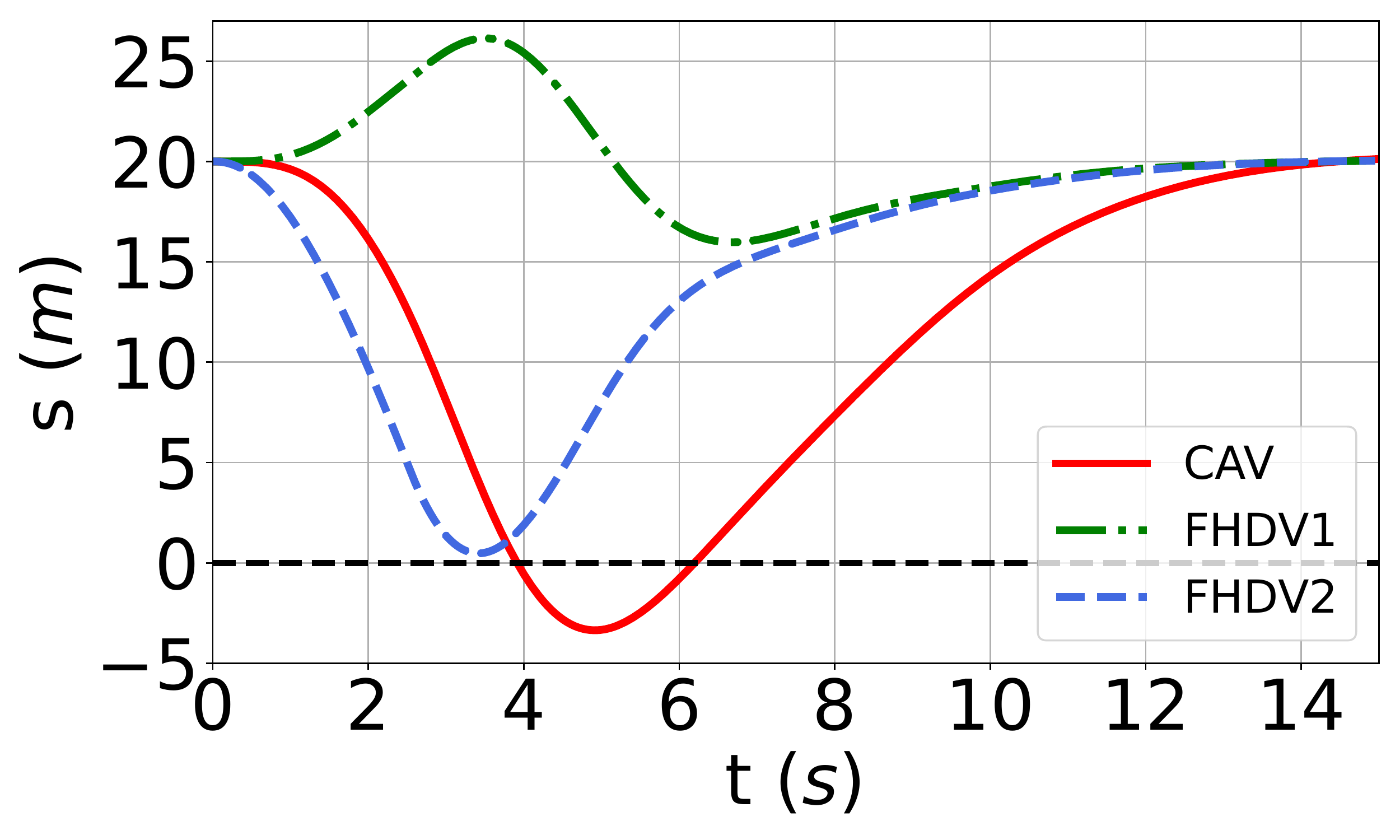}
    \includegraphics[width=0.24\linewidth]{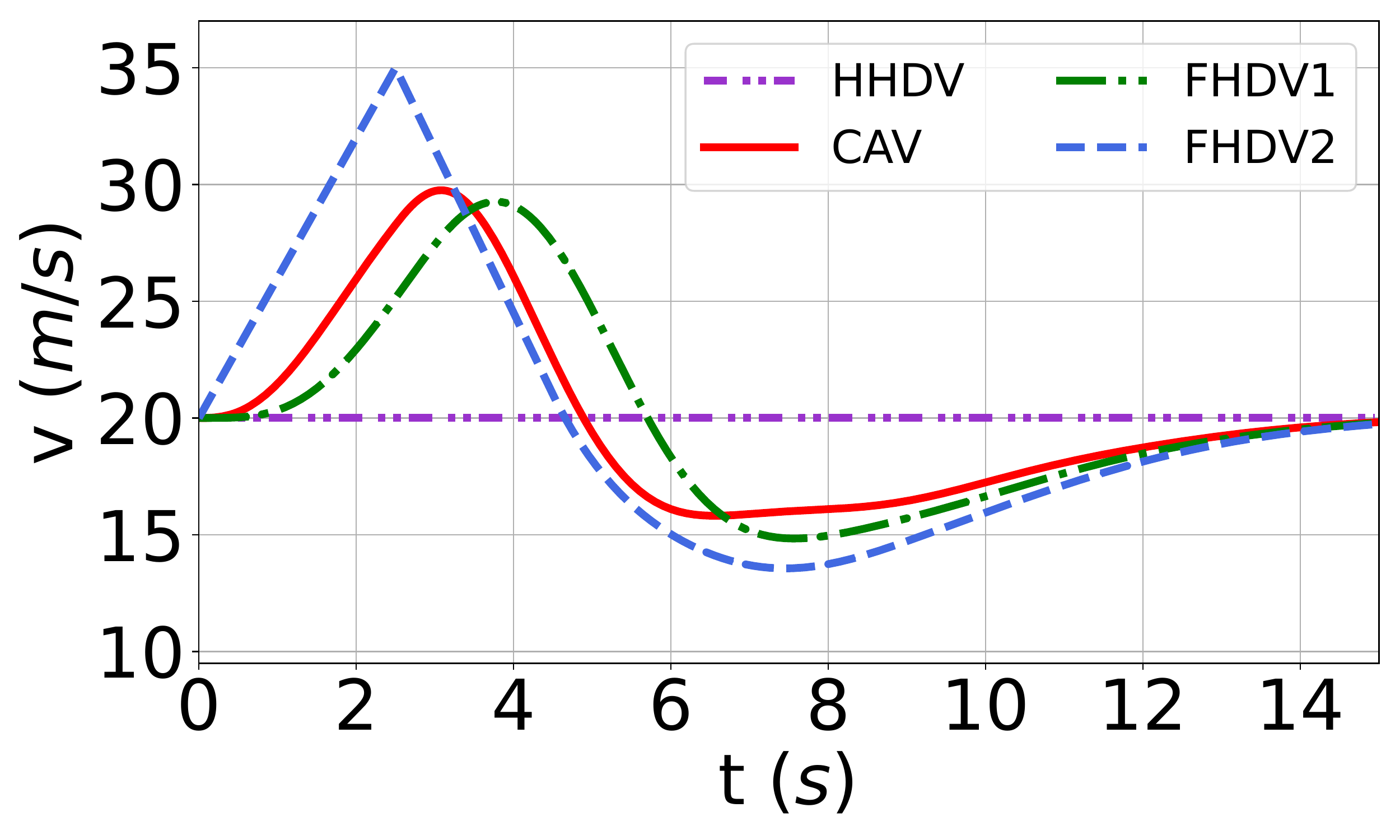}
    \includegraphics[width=0.24\linewidth]{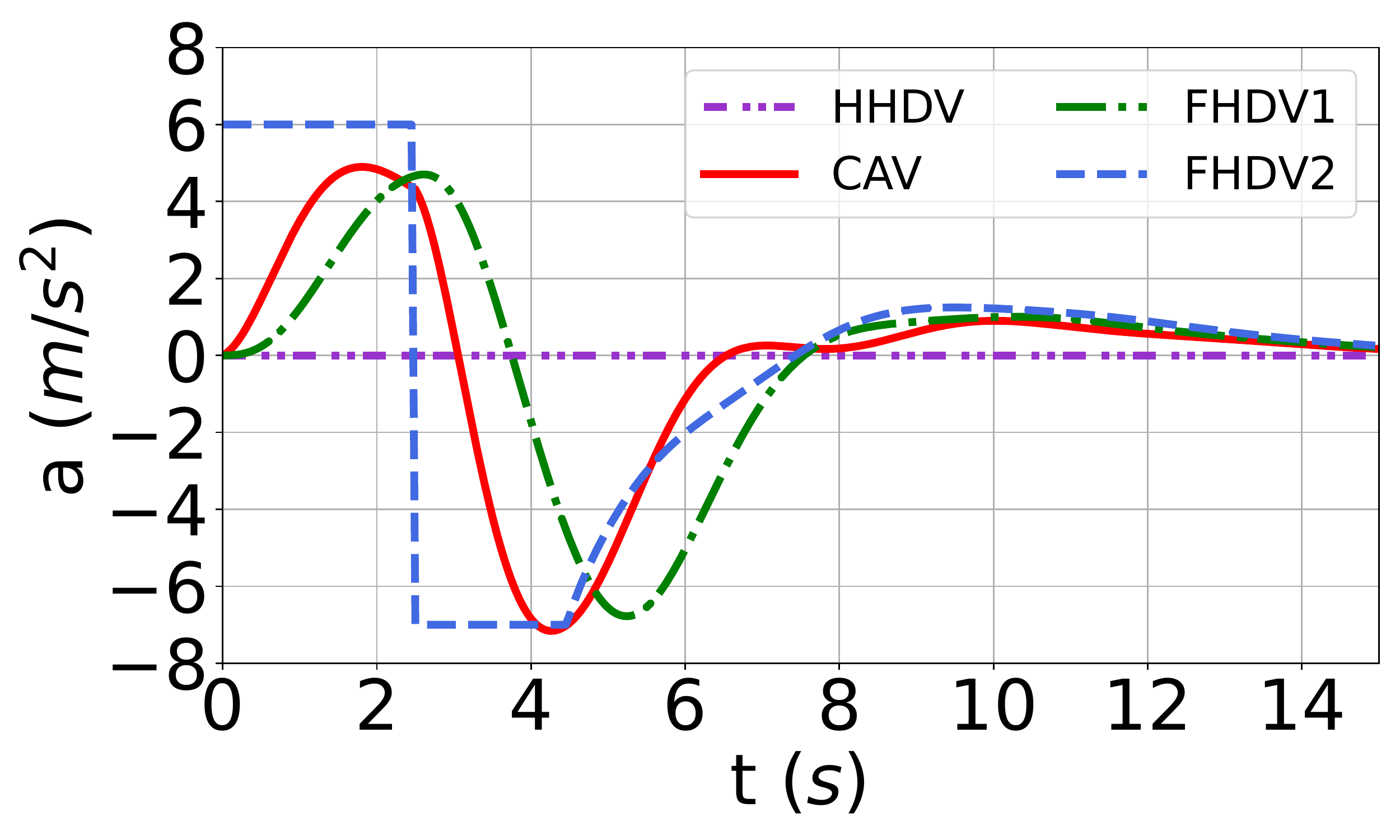}
    \includegraphics[width=0.24\linewidth]{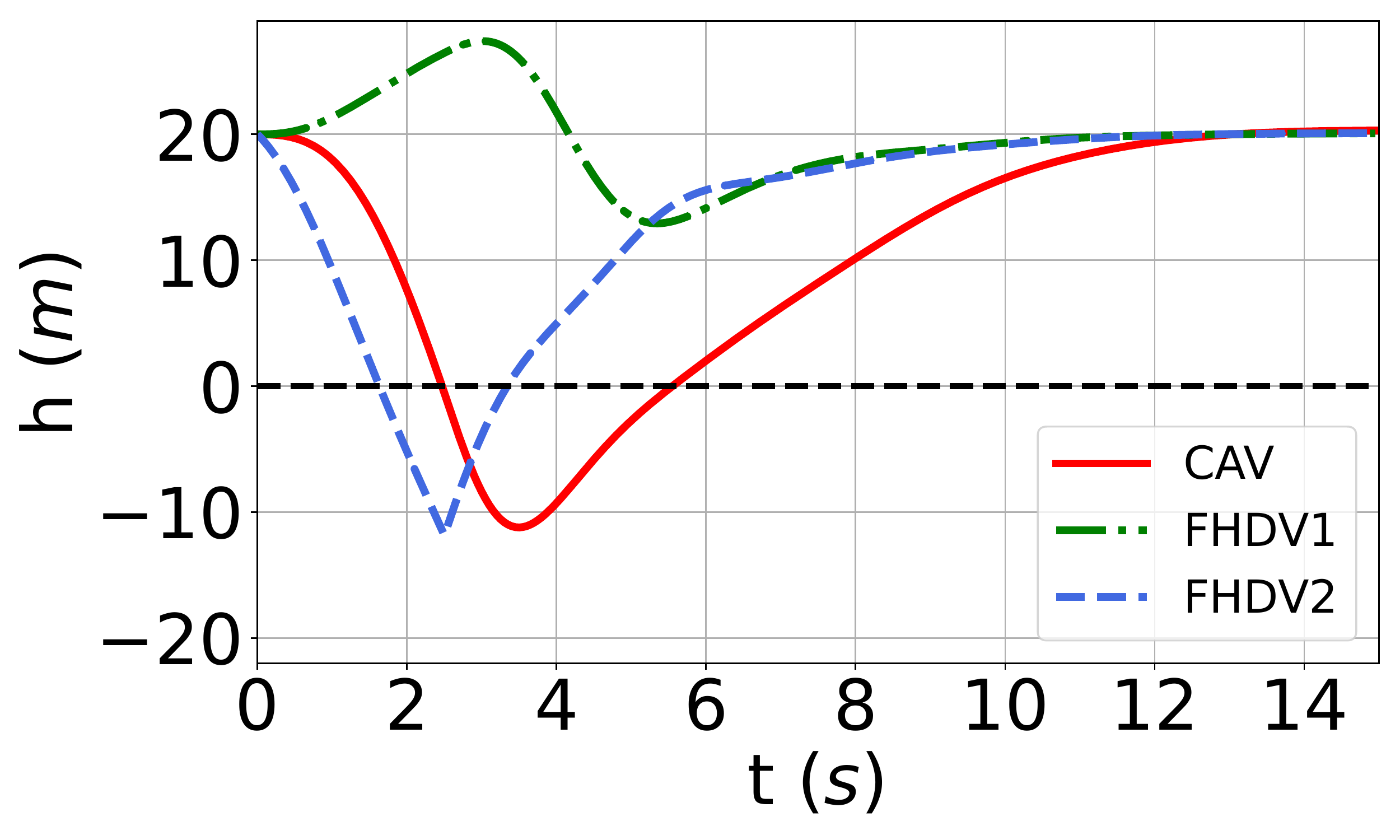}
    \\
    Proposed STC \\[6pt]
    \includegraphics[width=0.24\linewidth]{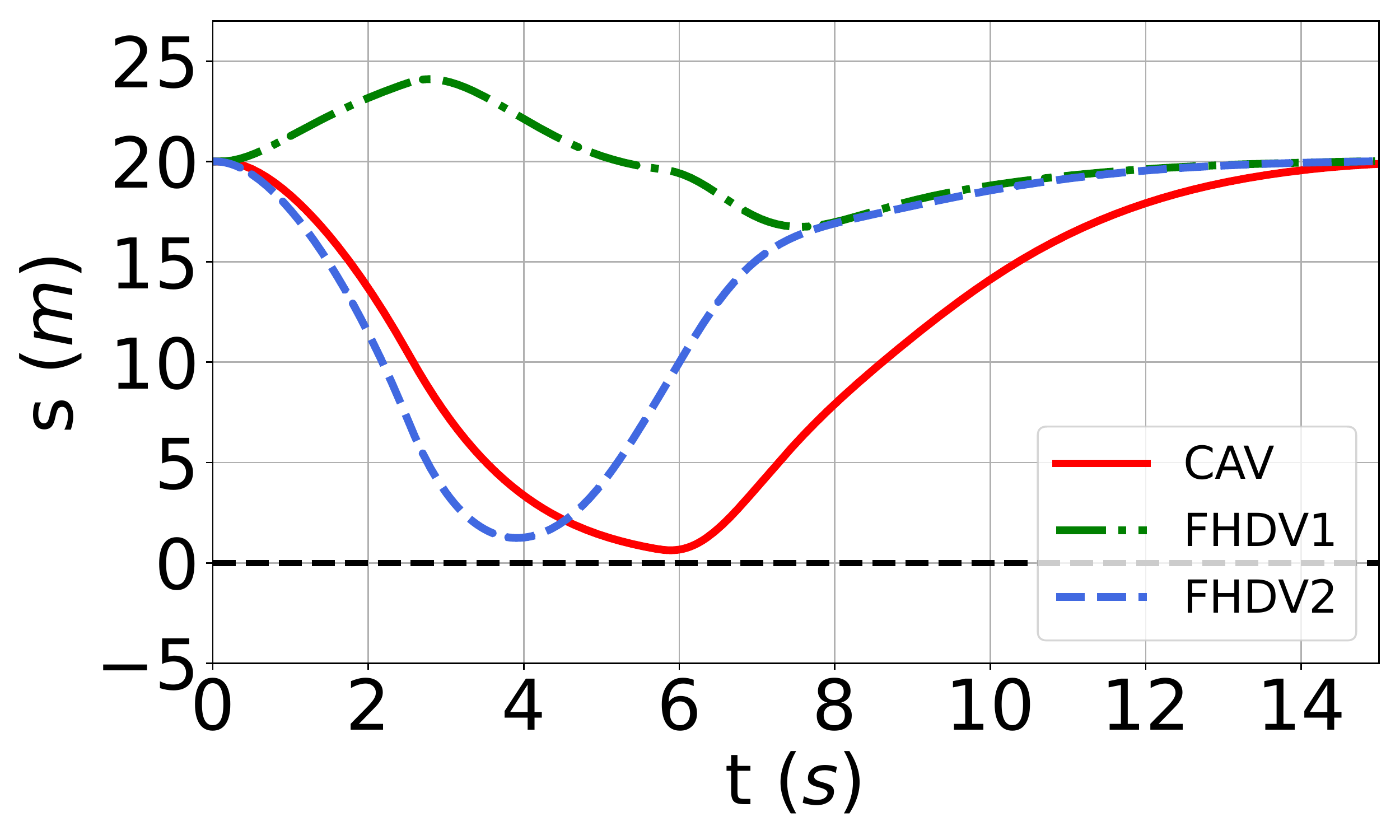}
    \includegraphics[width=0.24\linewidth]{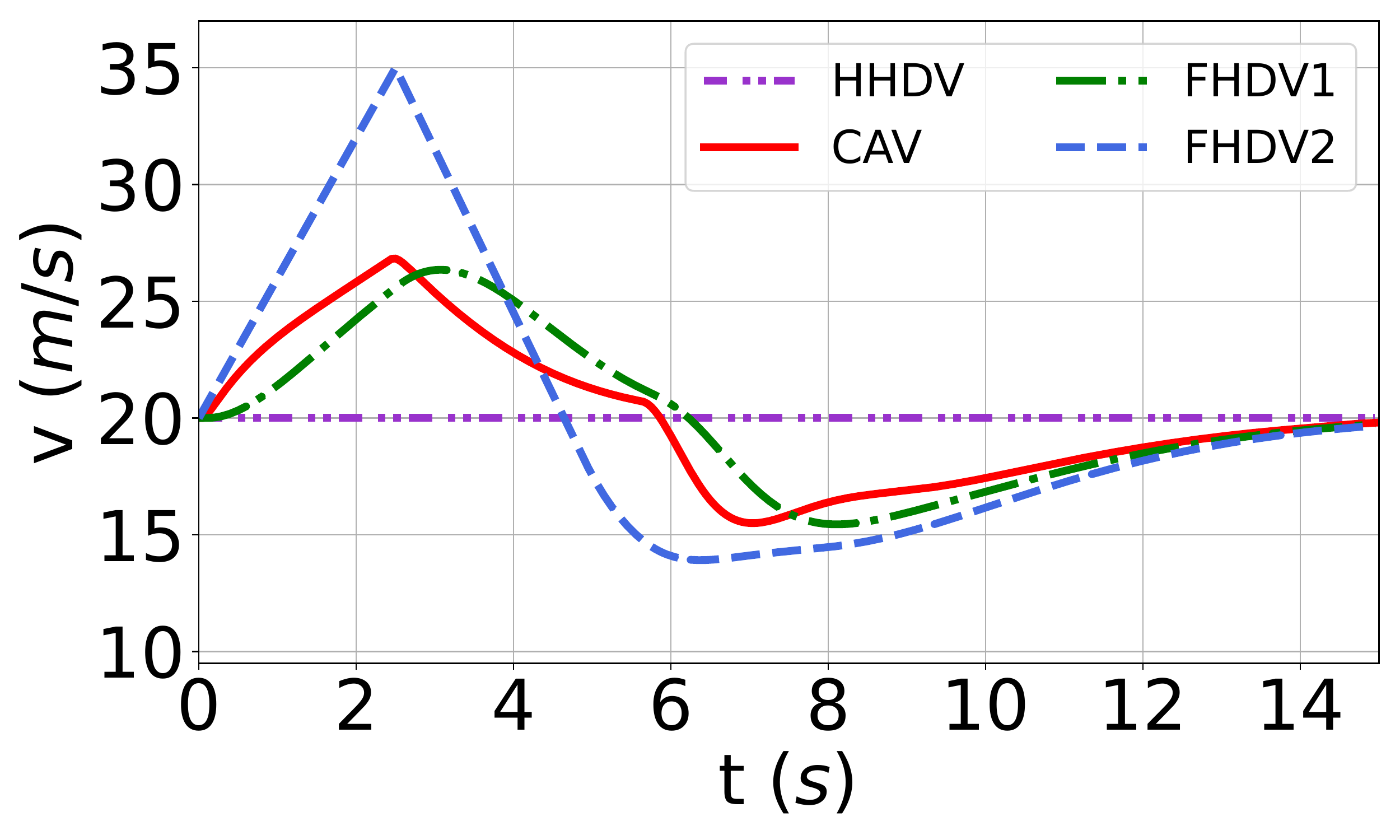}
    \includegraphics[width=0.24\linewidth]{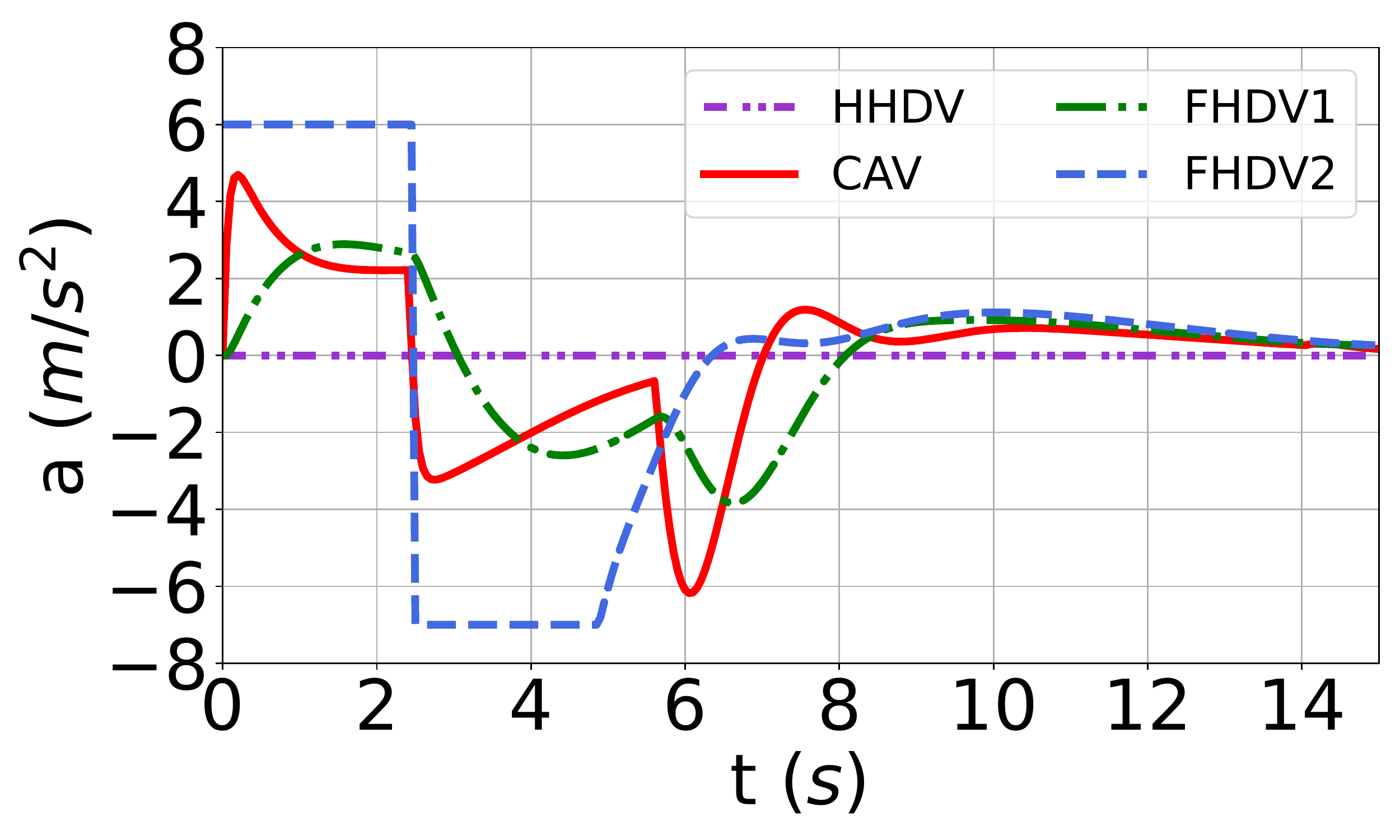}
    \includegraphics[width=0.24\linewidth]{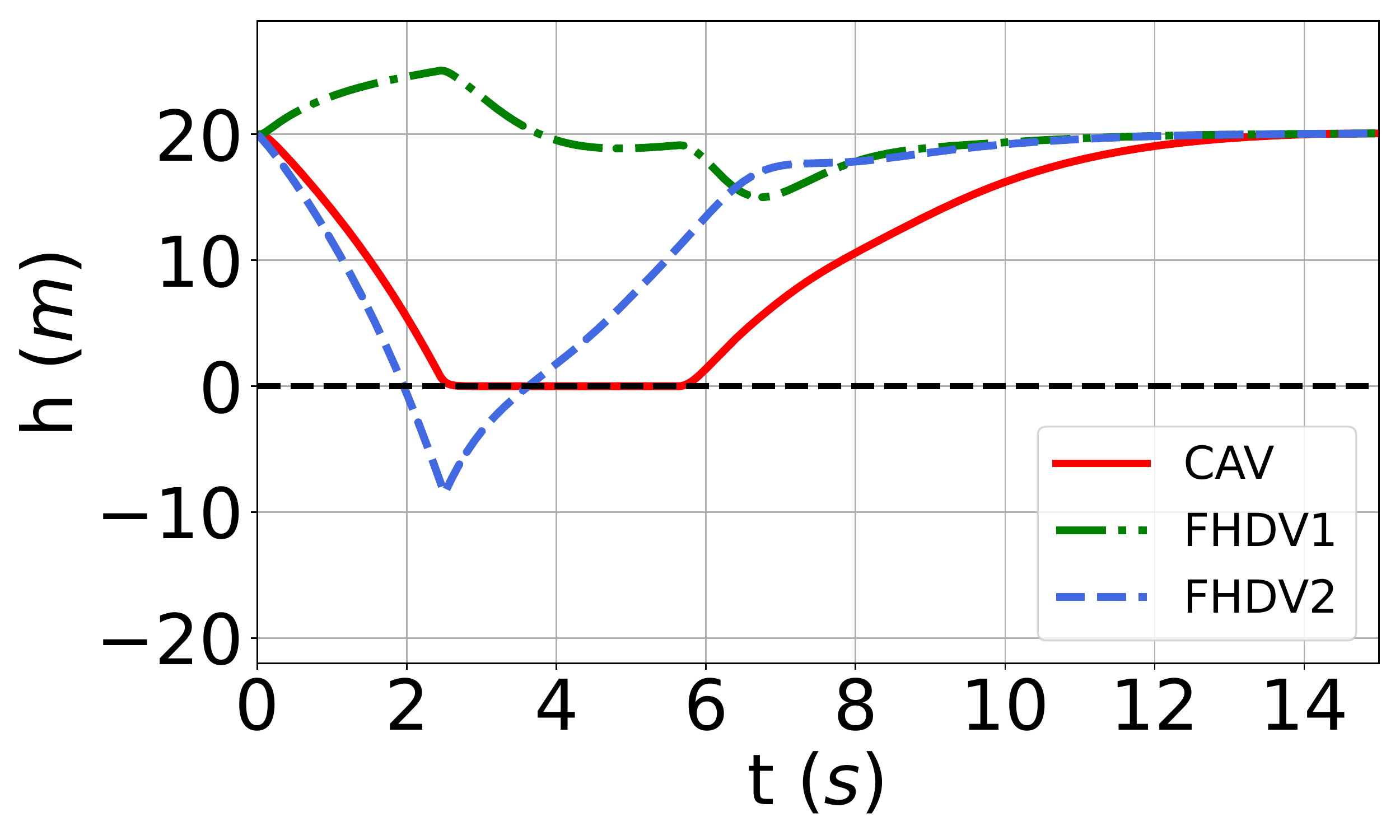}
    \caption{
    Numerical simulation of Scenario 2.
    Similar to Fig.~\ref{fig:sim:trajectory case1}, this scenario highlights that certain behaviors of HDVs can make the nominal controller of the CAV drive in an unsafe way (top), while the proposed STC recovers safe behavior (bottom).
    }
    \label{fig:sim:trajectory case2}
\end{figure}

In the results presented so far, the CAV's controller relies on the true state $x$ and uses full state feedback.
Now we address the scenario of output feedback, wherein the output $y$ in~\eqref{eq:y} is available rather than the full state $x$.
We use the state observer in~\eqref{eq:obs} 
to establish an estimate $\hat{x}$ of the state (with initial estimate
${\hat{x}_0 = [0, 0, 6, 0, 6, 0]^\top}$).
We compare two controllers:
\begin{itemize}
\item {\em ``na\"ive'' observer-based STC} that evaluates~\eqref{eq:QP} with the estimated state $\hat{x}$, without considering the estimation error ${\hat{x} - x}$;
\item {\em robust observer-based STC~\eqref{eq:QP_obs_robust}} that takes into account the estimation error.
\end{itemize}

Figure~\ref{fig:sim:observer case2} shows simulation results for STC with observer in the na\"ive (top) and robust (bottom) implementations in Scenario 2 with $a_{\mathrm{F}} = 6 \ \mathrm{m/s^2}$ and $t_{\mathrm{F}} = 2.2 $ s.
The figure highlights that the robust STC implementation leads to safer behavior (see the blue curves), thanks to taking into account the observer's estimation error bound in a provably safe fashion.
As opposed, the na\"ive STC implementation no longer has formal safety guarantees when utilized with estimated state instead of true state.

\begin{figure}[t]
    \centering
    ``Na\"ive'' Observer-based STC \\[6pt]
    \includegraphics[width=0.24\linewidth]{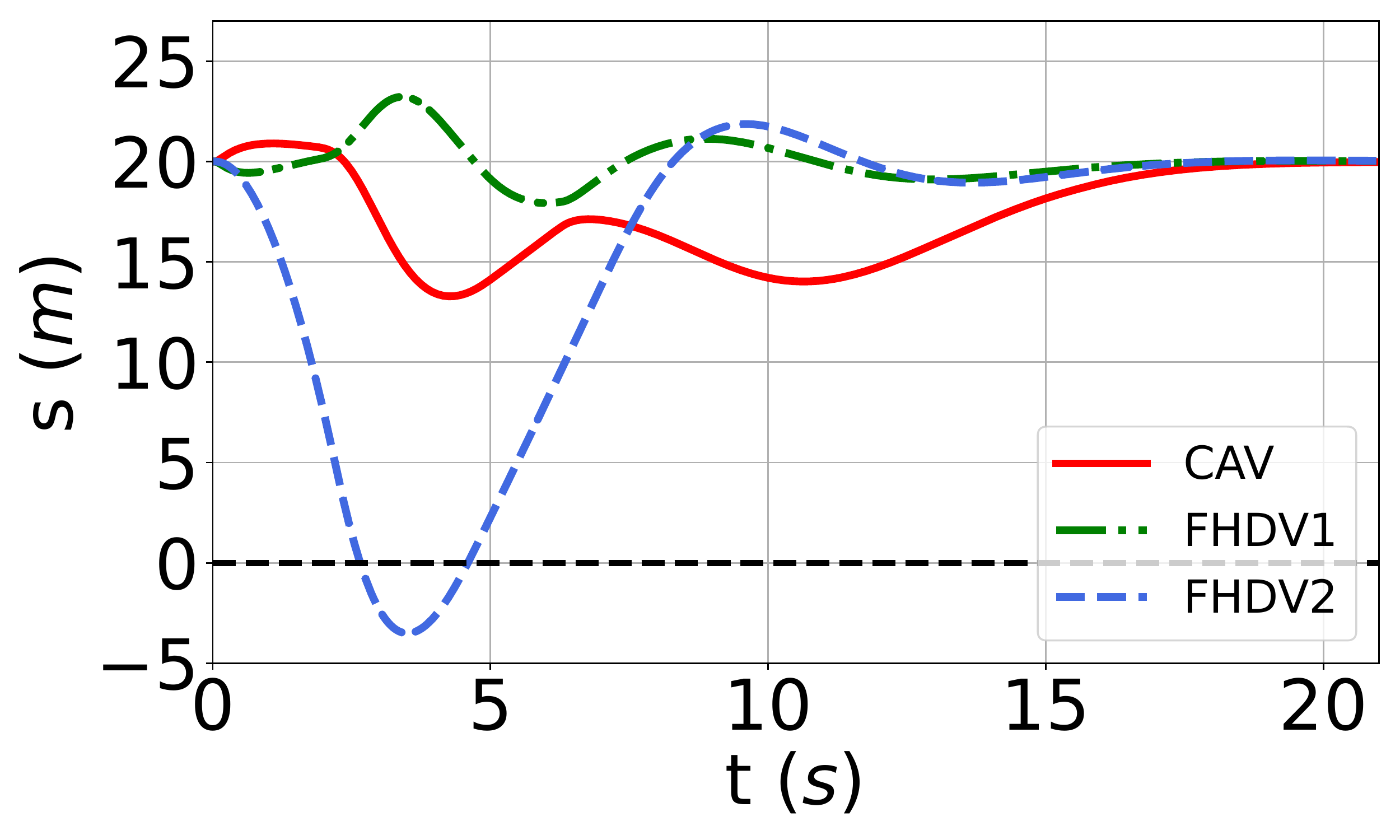}
    \includegraphics[width=0.24\linewidth]{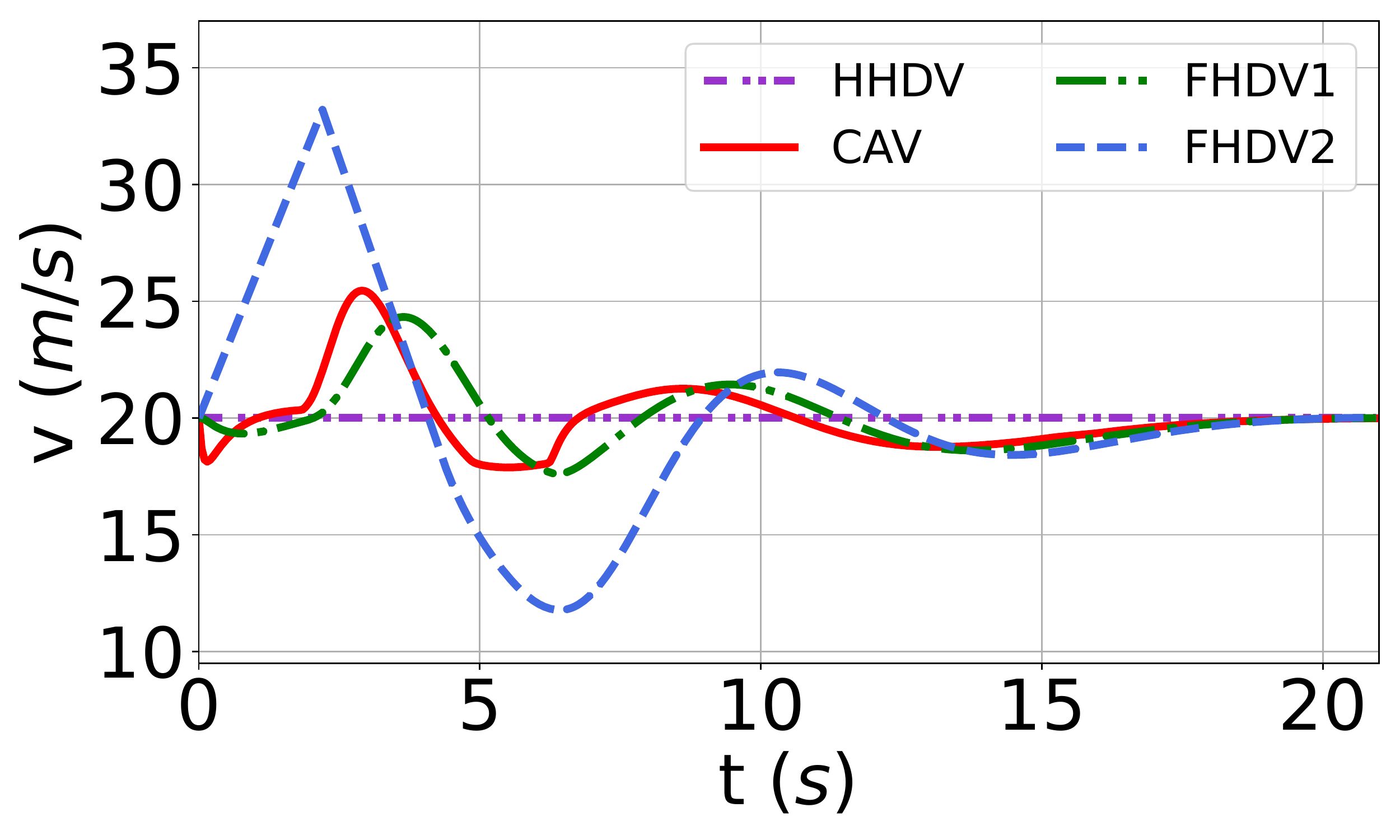}
    \includegraphics[width=0.24\linewidth]{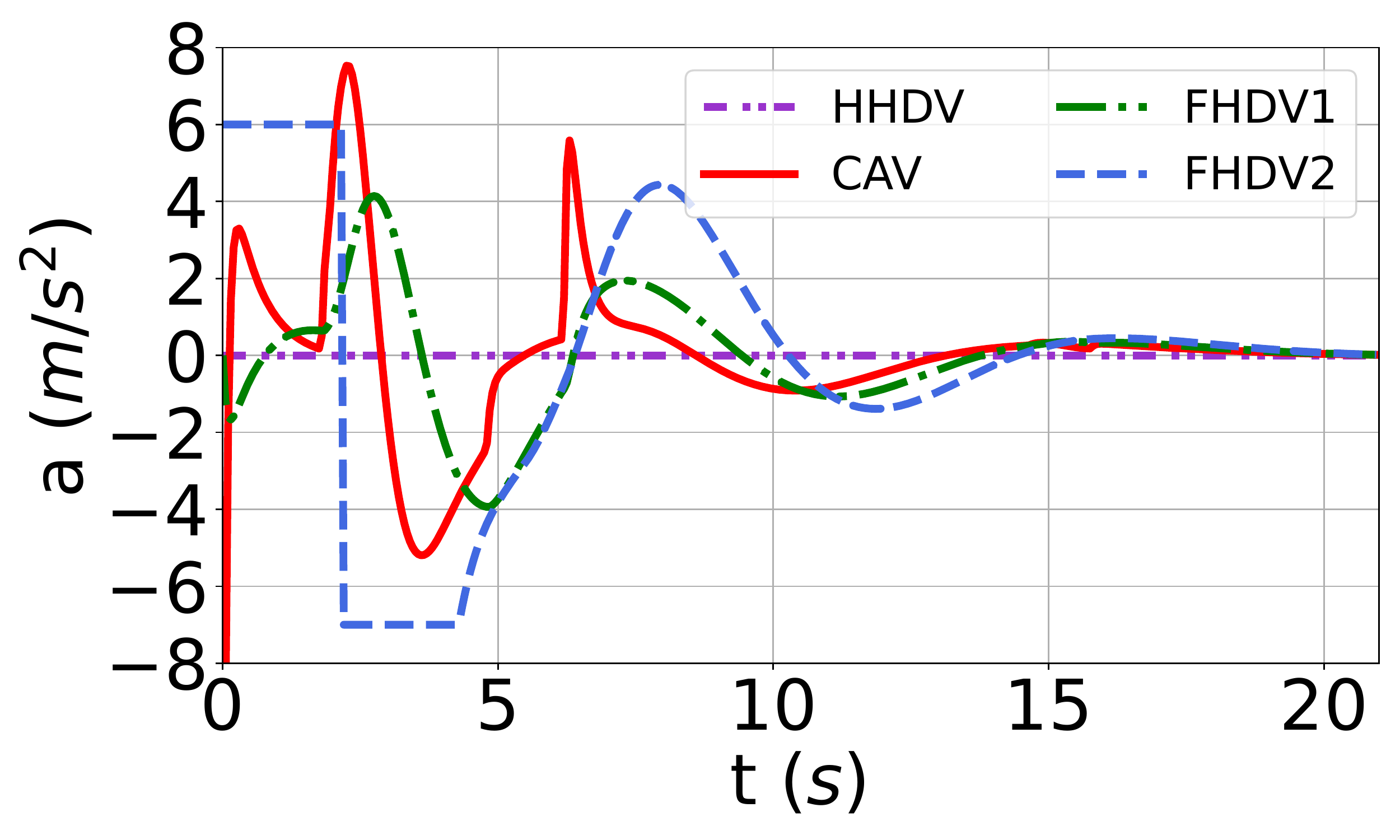}
    \includegraphics[width=0.24\linewidth]{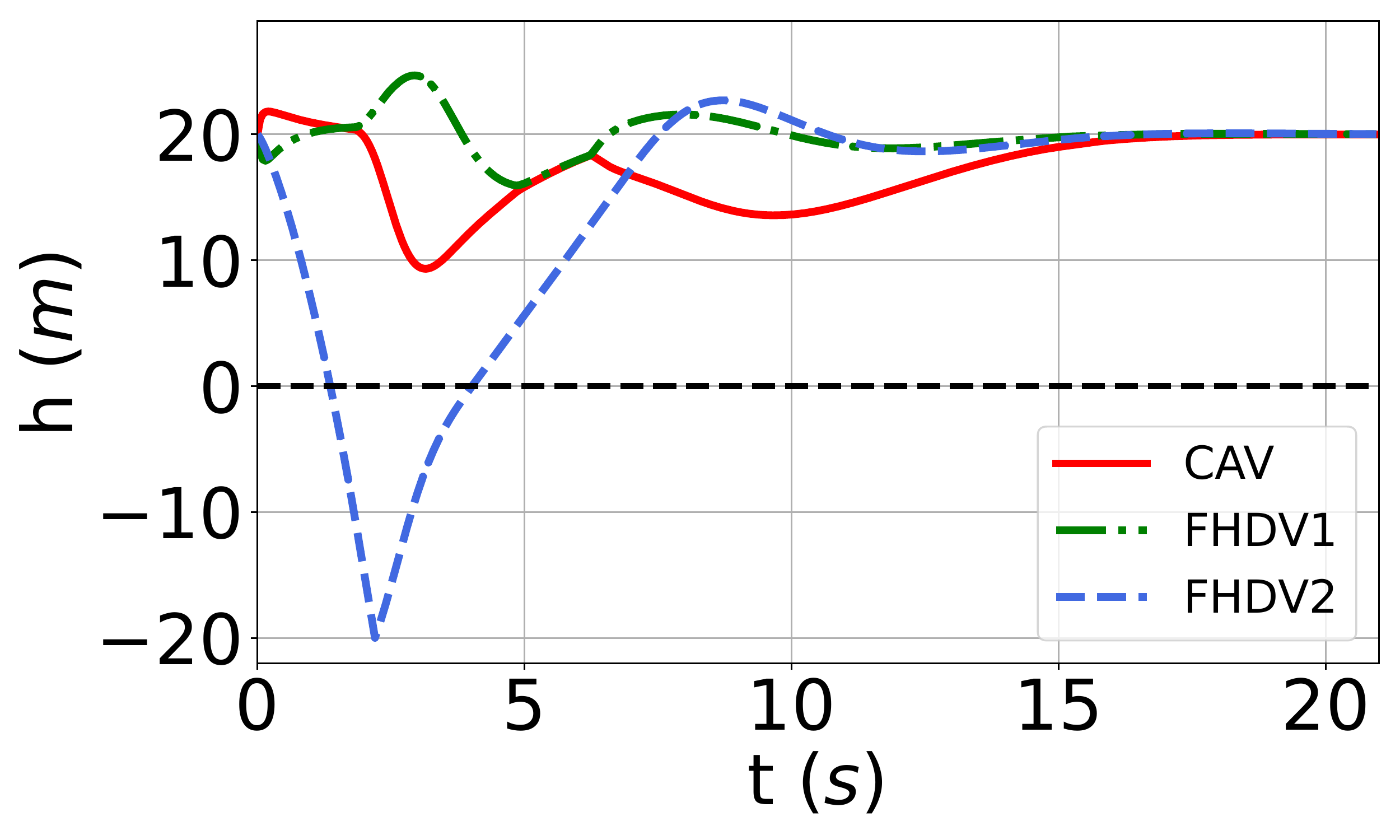}
    \\
    Robust Observer-based STC \\[6pt]
    \includegraphics[width=0.24\linewidth]{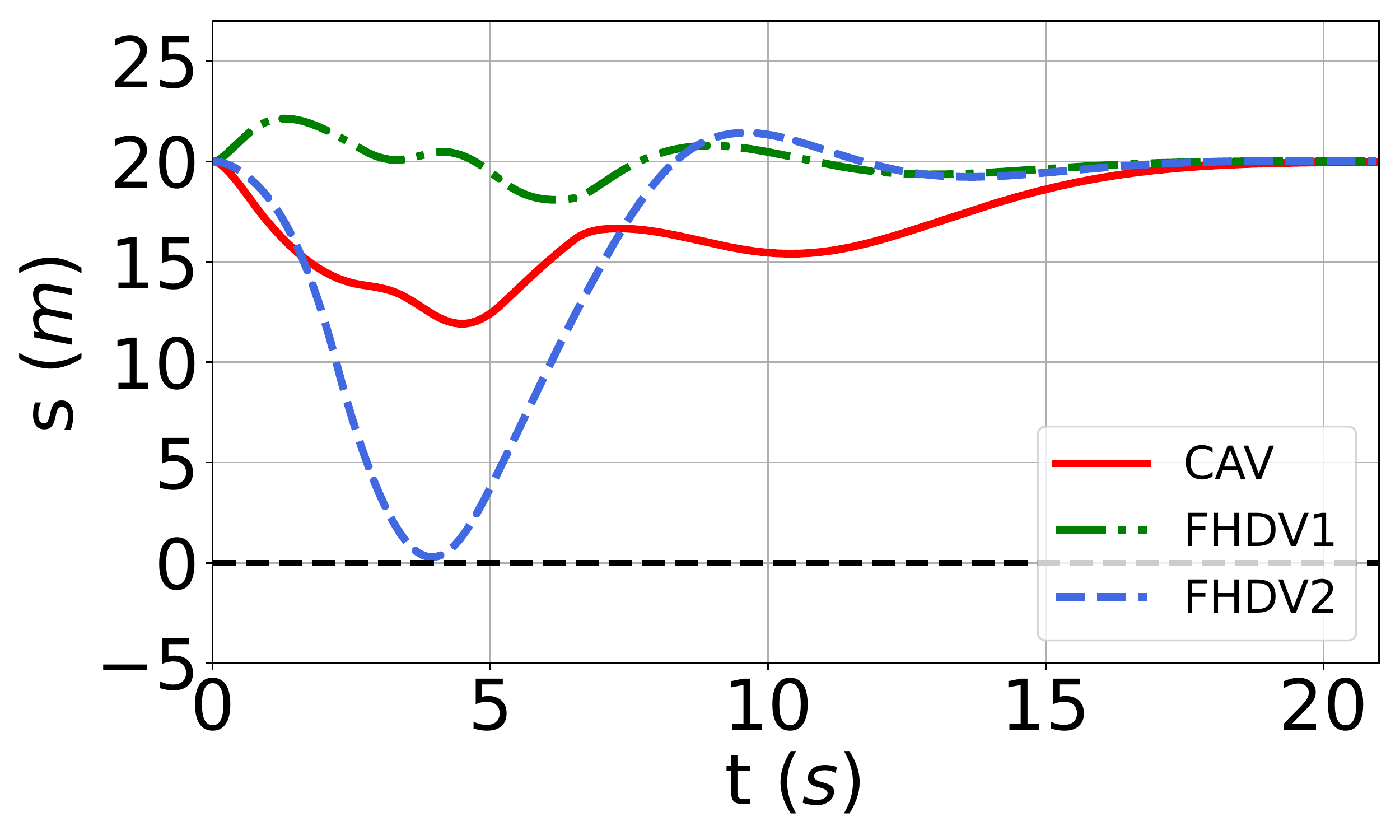}
    \includegraphics[width=0.24\linewidth]{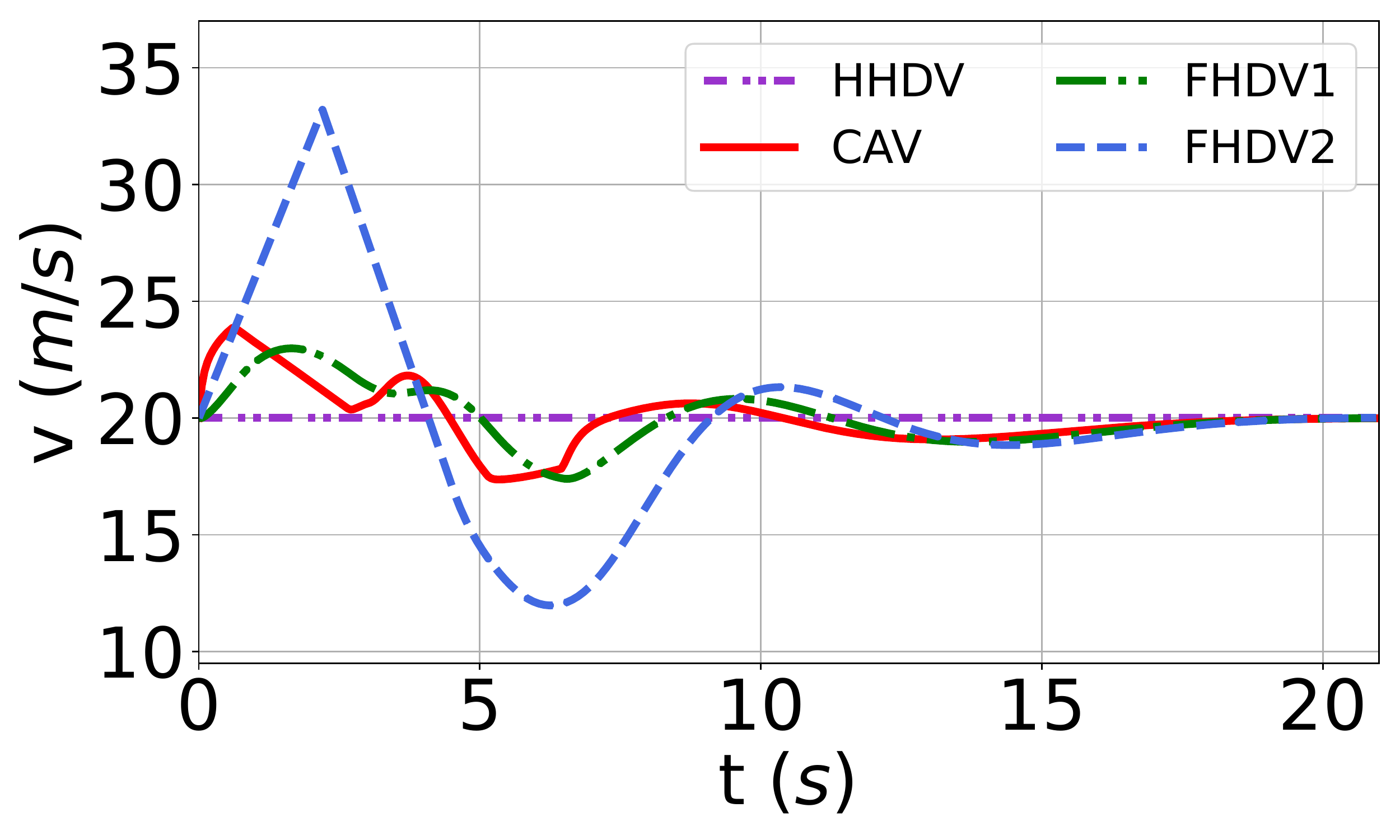}
    \includegraphics[width=0.24\linewidth]{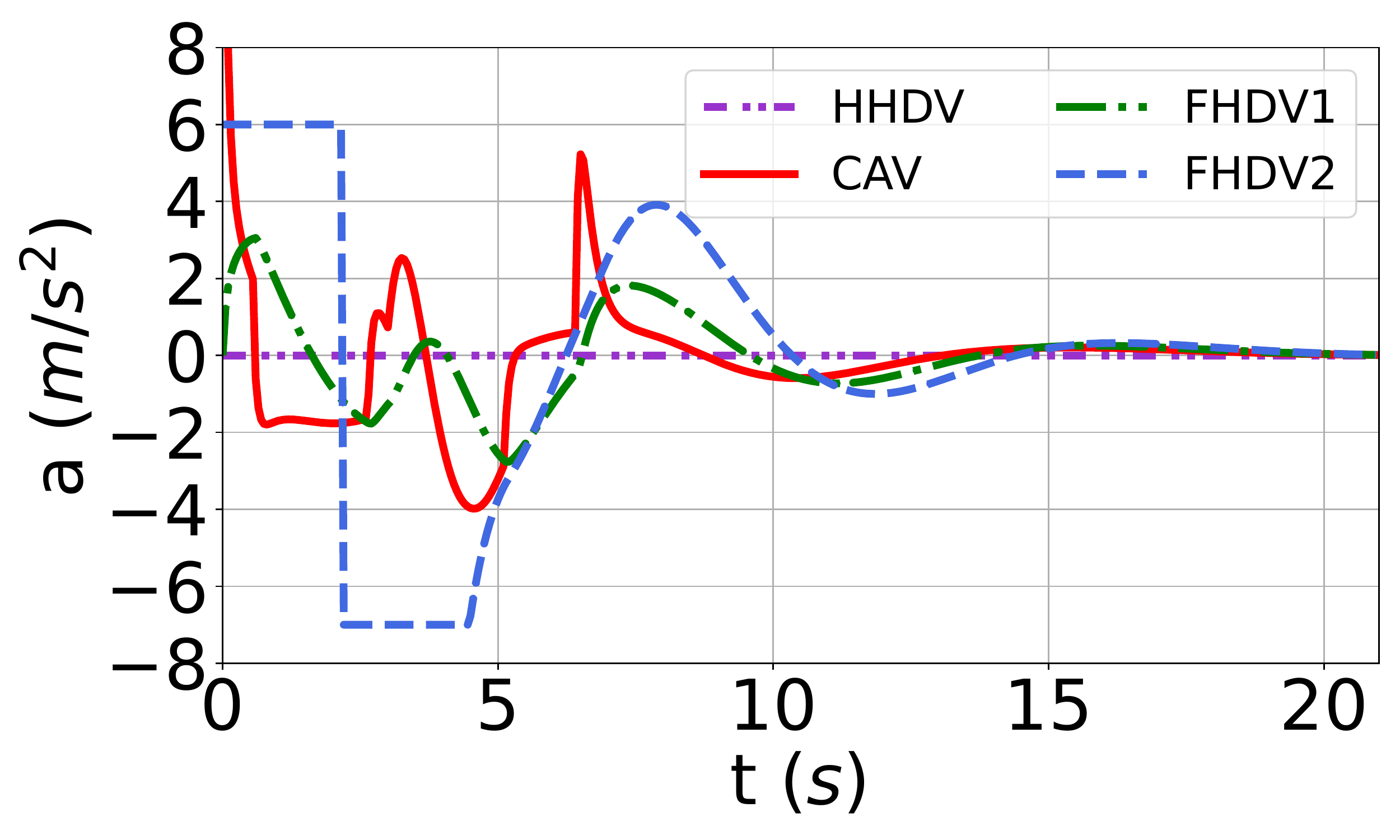}
    \includegraphics[width=0.24\linewidth]{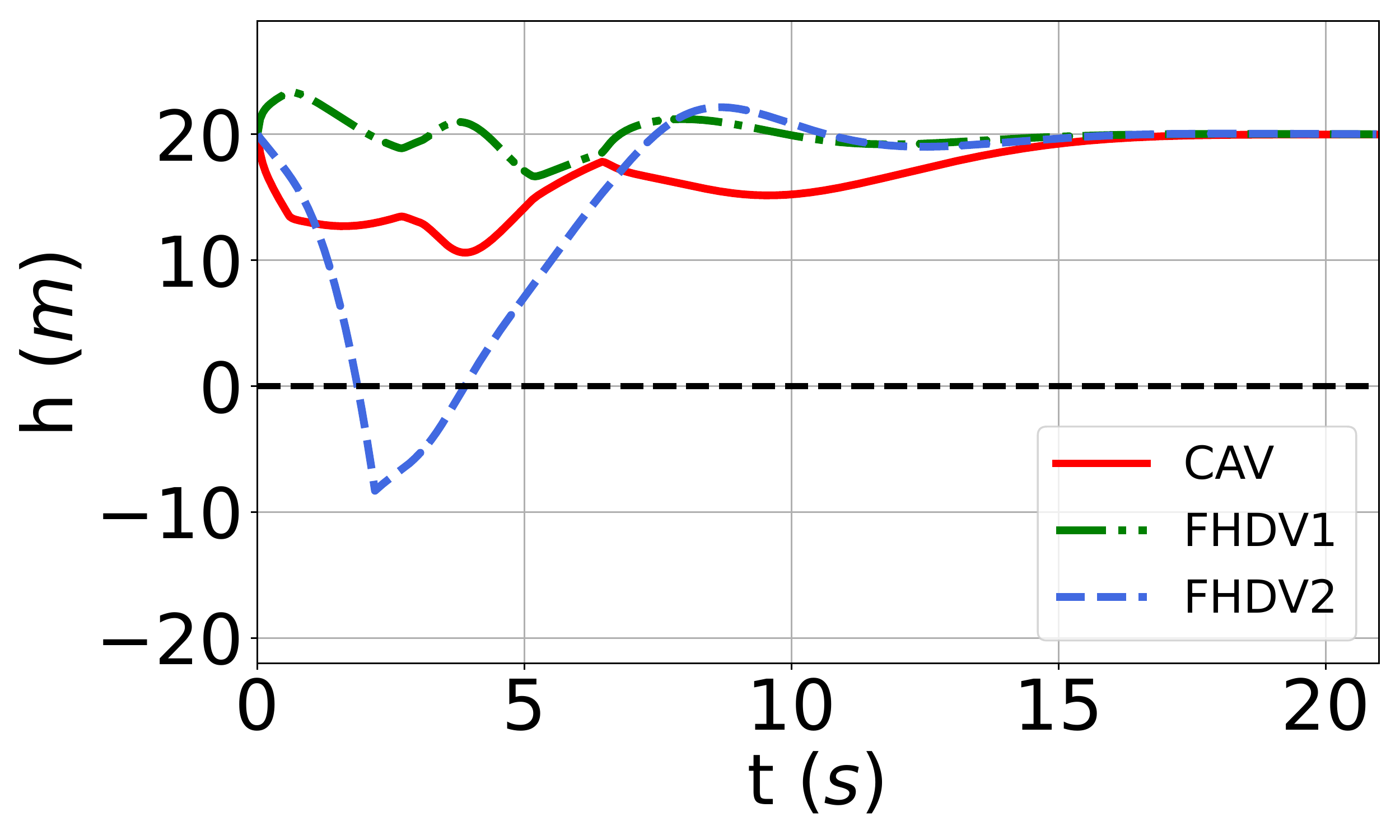}
    \caption{
    Observer-based STC for safe stabilization of mixed traffic with output feedback.
    A na\"ive implementation of observer-based STC (top), that disregards the observer's state estimation errors and directly evaluates~\eqref{eq:QP} with the state estimated by the observer~\eqref{eq:obs}, can potentially cause safety violations.
    The robust implementation of observer-based STC (bottom), that takes into account the estimation error in~\eqref{eq:QP_obs_robust}, maintains formal safety guarantees.
    }
    \label{fig:sim:observer case2}
\end{figure}

\subsection{Analysis of STC Performance} \label{sec:subsec:sim analysis}
In this section, we investigate the STC performance in more details.
We include discussion about the effects of acceleration limits, selection of parameters, and choice of spacing policy.
Through these, we highlight some limitations of the present approach and give suggestions on how to overcome them.
For simplicity, we conduct this analysis for full state feedback without incorporating observers into the control loop.

First, we highlight that safety may be affected by the acceleration capabilities of vehicles.
Notice that the control input $u$ of the proposed STC \eqref{eq:QP} and the accelerations $\dot{v}_i$ obtained from the human driver model~\eqref{eq:OVM} can take any real value.
In practice, however, vehicles have limited acceleration capabilities, hence these quantities are constrained to the range ${[a_{i,\min},a_{i,\max}]}$.
To investigate the effect of acceleration limits on safety, we conduct simulations where the accelerations are saturated as
${\dot{v}_0 = {\rm sat}(u)}$ and
${\dot{v}_i = {\rm sat}\big( F_i\left(s_{i}, \dot{s}_{i},  v_{i}\right) \big)}$, ${i \in \{1, \ldots, N\}}$, with
\begin{equation}
    {\rm sat}(u) = \max \big\{ a_{\min}, \min\{u, a_{\max} \}\big\}
\end{equation}
and uniform acceleration limits $a_{\min}$ and $a_{\max}$ for each vehicle given in Table~\ref{tab:parameters}.

Figure~\ref{fig:sim:infinite u} presents the effect of acceleration limits on safety.
Here, the simulation scenario and parameters correspond to the bottom of Fig.~\ref{fig:sim:trajectory case1}, except that the saturation function is also incorporated into the dynamics.
The red curves in the figure highlight that there exist scenarios in which acceleration limits affect safety (i.e., $s_0$  goes negative when the acceleration of the CAV saturates), since these limits are not incorporated into the STC design.
While addressing bounded control inputs (limited acceleration) is possible by CBF theory~\cite{gurriet2020scalable, agrawal2021safe, liu2022safe} and could be included in STC, it is a nontrivial task and significantly increases the complexity of control design---hence it is left for future work.
Instead, in this paper we quantify the amount of safety violations due to the saturation, and tune the controller parameters to minimize and eliminate these violations.

Specifically, we investigate the effects of three main sets of tunable parameters in STC~\eqref{eq:QP}: the time $\tau_i$ in the CBF candidates, the coefficient $\gamma_i$ in the safety constraints, and the penalty $p_i$ in the cost.
For simplicity, we keep these parameters to be the same for all vehicles, and we use the short notation $\tau$, $\gamma$, $p$.
The default values of these parameters (used in the previous figures) are listed in Table~\ref{tab:parameters}.
Now we analyze how these parameters affect the performance of STC.

In particular, we characterize how abrupt HDV motions can be handled by STC with limited acceleration as a function of parameters $\tau$, $\gamma$, $p$.
To achieve this, we conduct a large number of simulations with various $a_{\mathrm{H}}$, $t_{\mathrm{H}}$, $a_{\mathrm{F}}$, $t_{\mathrm{F}}$ values characterizing the motion of the HDVs in \eqref{eq:HHDV_accel}-\eqref{eq:FHDV_accel}.
In Scenario 1, we quantify how much the HHDV can reduce its speed without causing collision, i.e., we identify the safe range of the minimum speed ${v_{-1,\min} = v^\star - a_{\mathrm{H}}t_{\mathrm{H}}}$  for which the response of the vehicles satisfies ${s_i(t) \ge 0}$.
In Scenario 2, we quantify how much FHDV-2 can increase its speed without causing collision, i.e., we calculate the safe range of ${v_{2,\max} = v^\star + a_{\mathrm{F}}t_{\mathrm{F}}}$ for which ${s_i(t) \ge 0}$.
We calculate these by considering both the entire vehicle chain (such that ${s_i(t) \ge 0}$, ${\forall i \in \{0,1,2\}}$), and also separately for the CAV (${s_0(t) \ge 0}$), FHDV-1 (${s_1(t) \ge 0}$) and FHDV-2 (${s_2(t) \ge 0}$).

Figure~\ref{fig:sim:region SDH} shows the safe regions of $v_{-1,\min}$ in Scenario 1 (top) and $v_{2,\max}$ in Scenario 2 (bottom) for various accelerations $a_{\rm H}$ and $a_{\rm F}$, respectively.
These results are presented for fixed $\gamma$ and $p$, while varying $\tau$.
We make the following conclusions.
\begin{itemize}
    \item In Scenario 1, there is no collision for FHDV-1 and FHDV-2 for all values of $\tau$ even if the HHDV brakes to a complete stop, since the safe range of $v_{-1,\min}$ covers the whole interval $[0,v^\star]$ in Fig. \ref{fig:sim:region SDH}(c) and (d).
    The safety of the vehicle chain is therefore determined by the safety of the CAV.
    Considering the CAV's safety in Fig. \ref{fig:sim:region SDH}(b), the safe range of $v_{-1,\min}$ grows by increasing $\tau$, since larger $\tau$ corresponds to larger safe spacing.
    Ultimately, STC makes the CAV avoid collisions for any HHDV motion if $\tau$ is selected to be large enough---even with limited acceleration capability.
    \item In Scenario 2, the CAV still avoids collision with STC if $\tau$ is large enough.
    According to Fig. \ref{fig:sim:region SDH}(f), the safe range of $v_{2,\max}$ covers the whole interval $[v^\star,v_{\max}]$ for almost all choices of $\tau$ except the smallest value.
    The trend is the opposite for FHDV-2, where the safe range of $v_{2,\max}$ shrinks for the largest value of $\tau$; see panel (h).
    Meanwhile, FHDV-1 avoids collision for any $\tau$; see panel (g).
\end{itemize}
Moreover, notice that STC (in color) outperforms the nominal controller (gray) in both scenarios, since the safe ranges of $v_{-1,\min}$ and $v_{2,\max}$ are significantly larger for STC.
Finally, we performed the same analysis by varying $\gamma$ in the range $\gamma \in \{5,10,15\}$ while keeping $\tau$ and $p$ fixed, and by varying $p$ in the range $p \in \{10,100,1000\}$ while keeping $\tau$ and $\gamma$ fixed.
We found that the STC is less sensitive to $\gamma$ and $p$, as there is only slight variation in the safe ranges of velocity perturbations.
These analyses justify the parameter selection in Table~\ref{tab:parameters}.

\begin{figure}[t]
    \centering
    \includegraphics[width=0.23\linewidth]{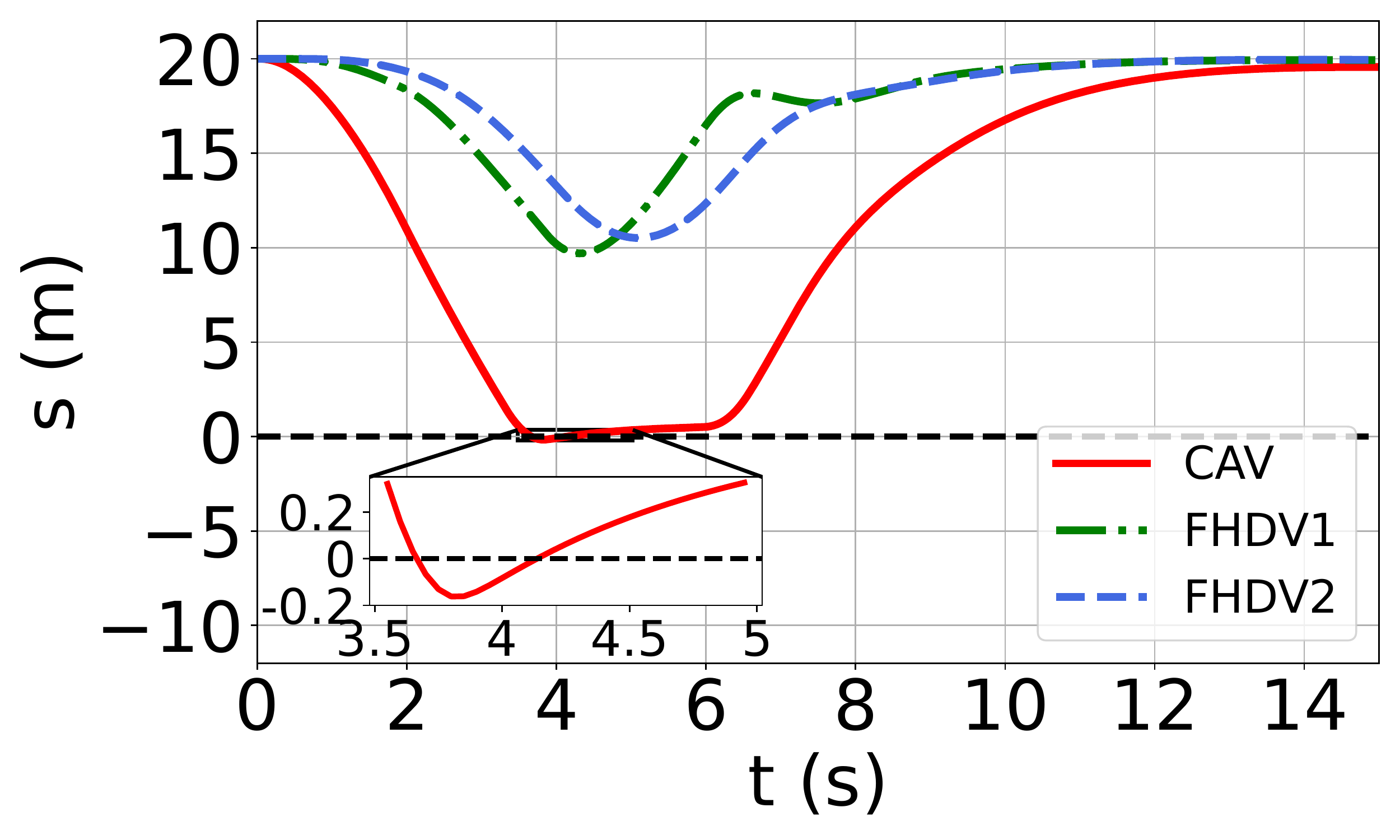}
    \includegraphics[width=0.23\linewidth]{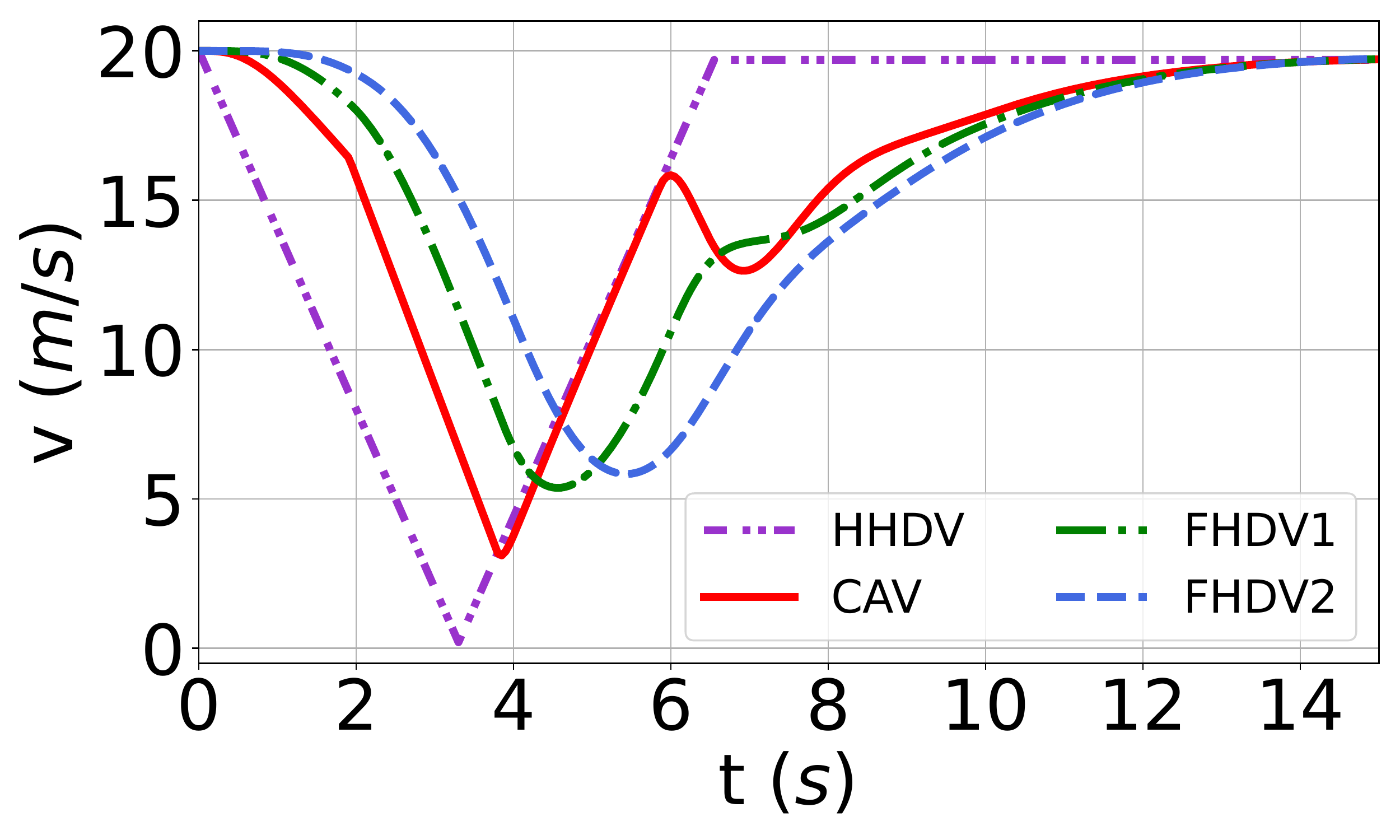}
    \includegraphics[width=0.23\linewidth]{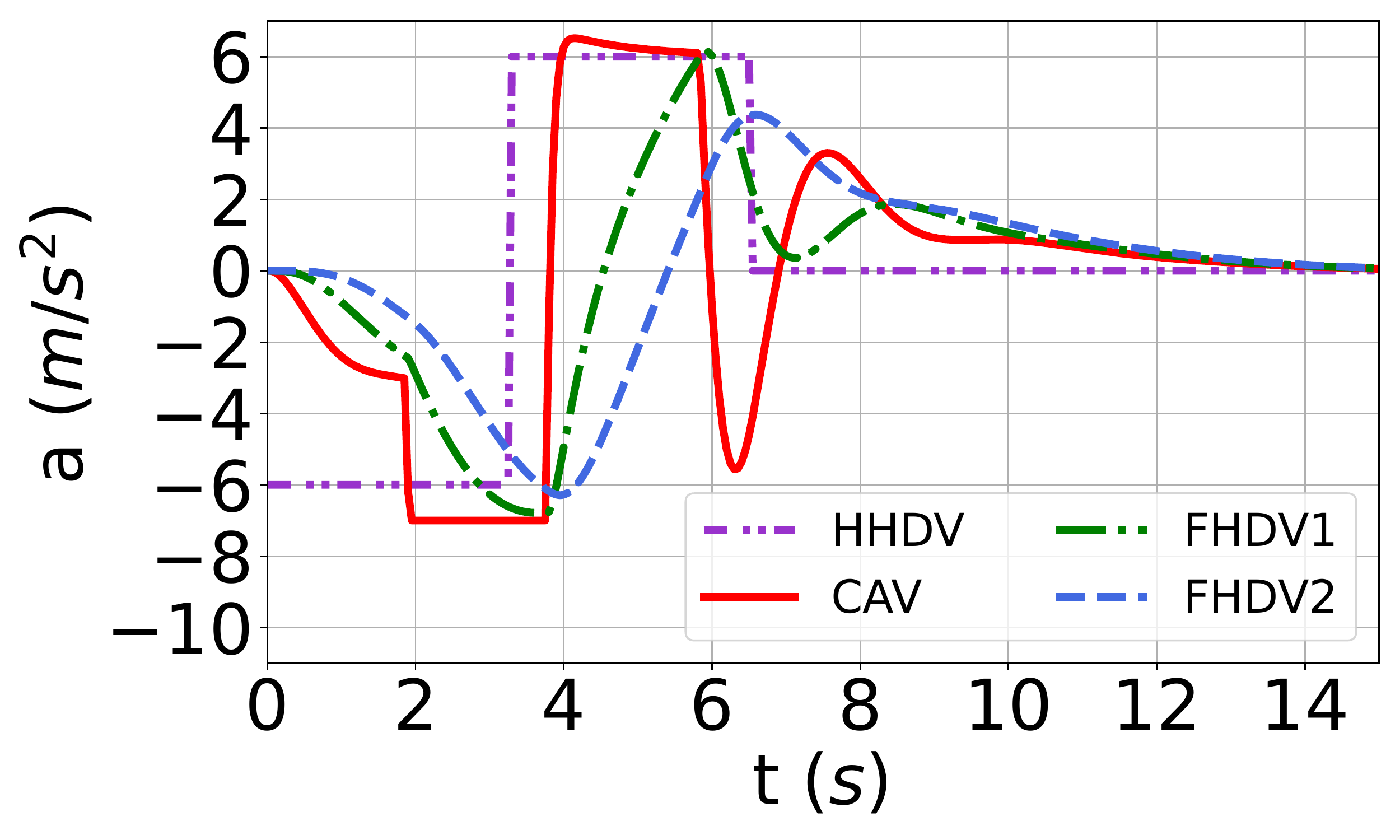}
    \includegraphics[width=0.23\linewidth]{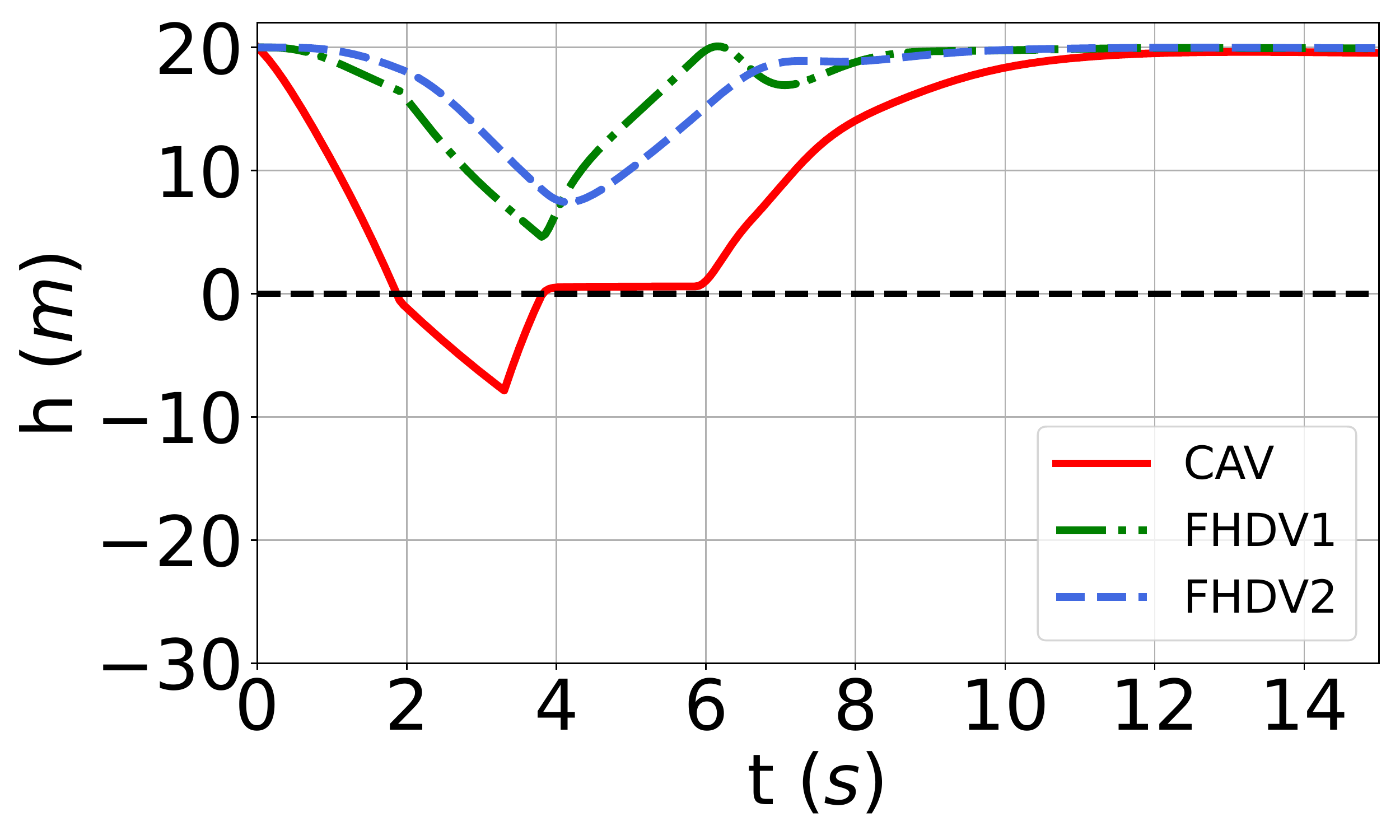}
    \caption{
    The effect of acceleration limits on the performance of STC.
    The safety of mixed traffic is negatively impacted by the acceleration limits compared to the case with no limits shown at the bottom of Fig.~\ref{fig:sim:trajectory case1}.
    }
    \label{fig:sim:infinite u}
\end{figure}

\begin{figure}[t]
\centering
\hspace*{-0.5cm}
Scenario 1 \\[6pt]
\subcaptionbox{vehicle chain}{\includegraphics[width=0.23\linewidth]{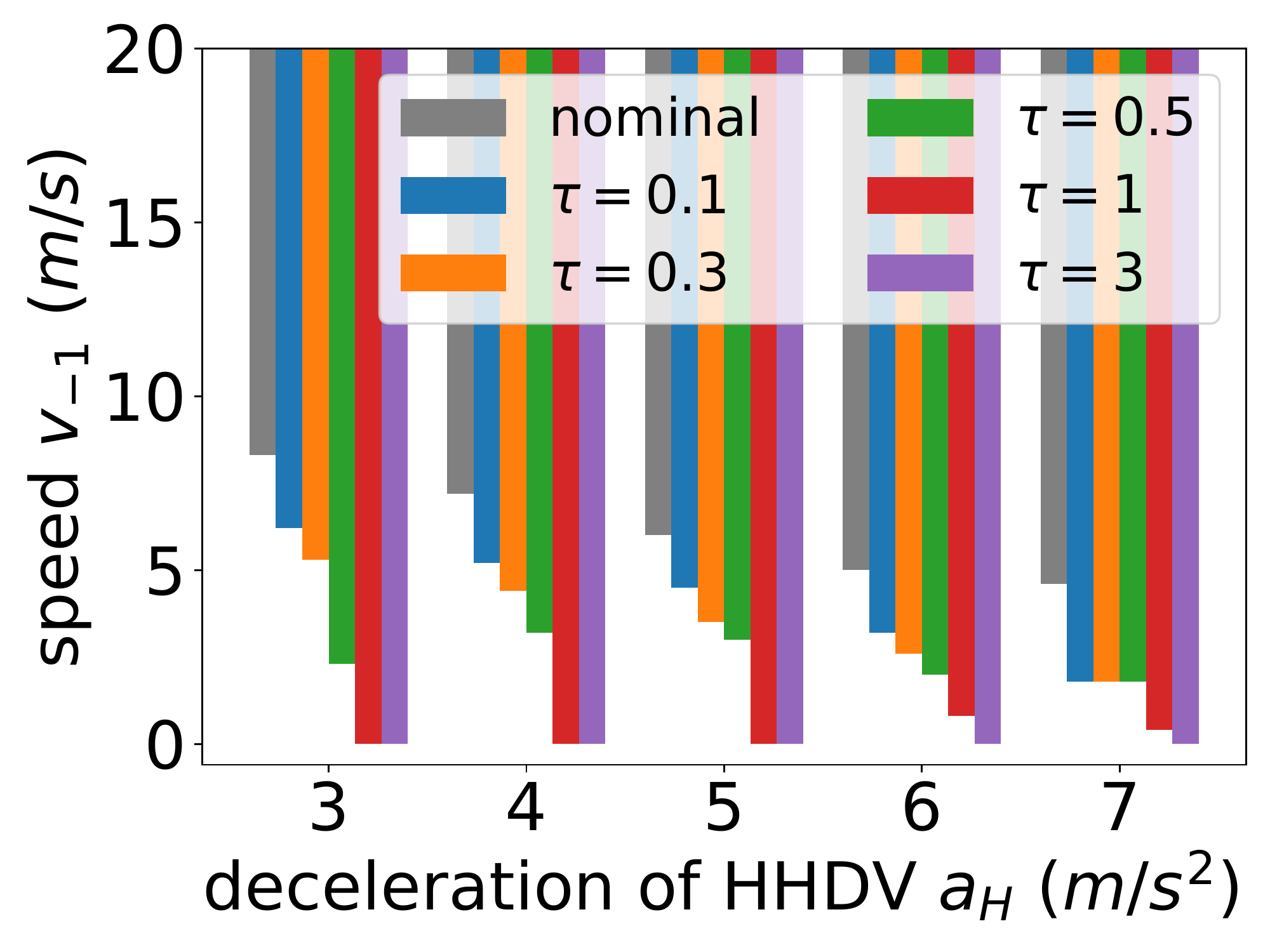}}
\subcaptionbox{CAV}{\includegraphics[width=0.23\linewidth]{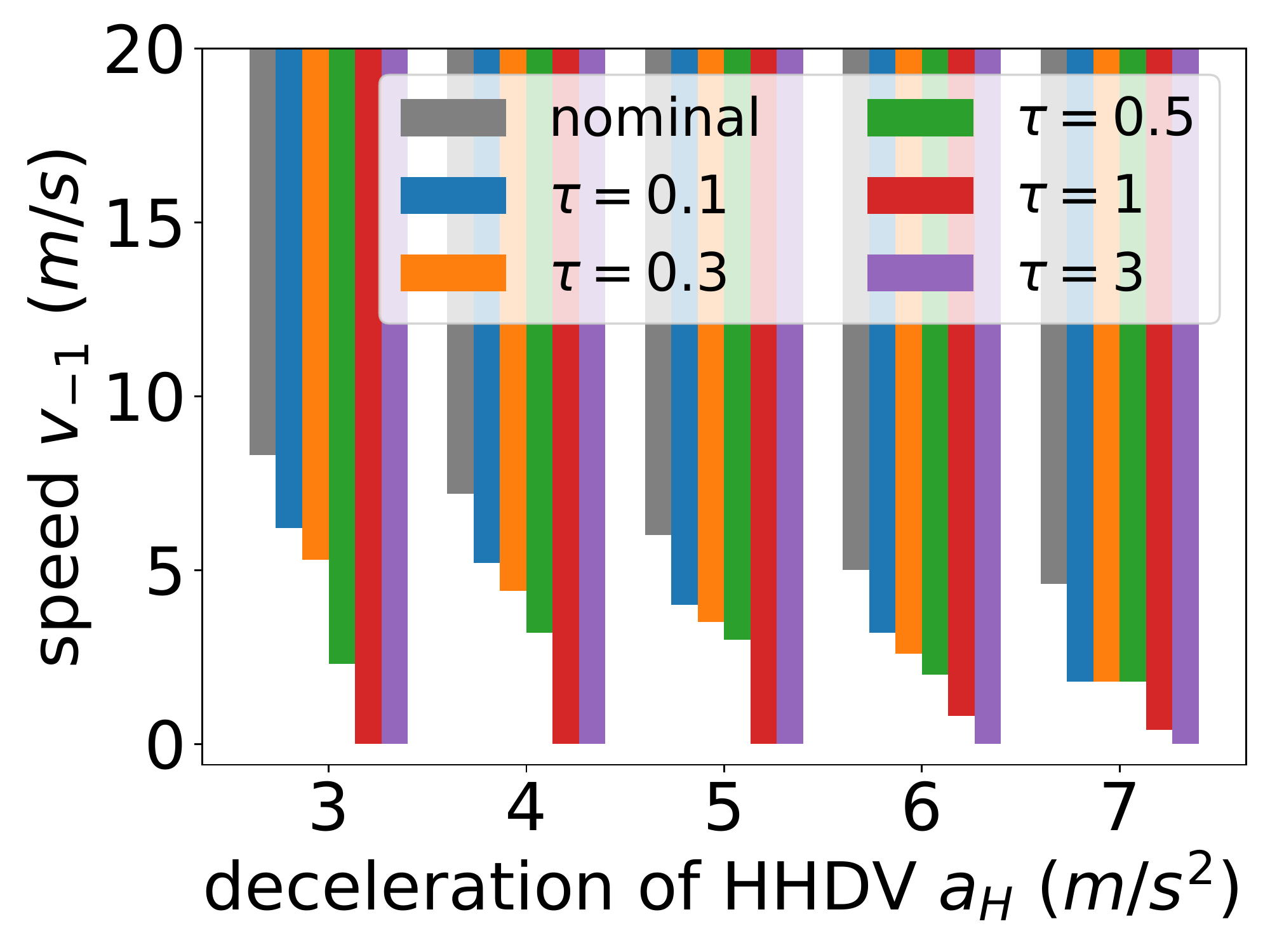}}
\subcaptionbox{FHDV-1}{\includegraphics[width=0.23\linewidth]{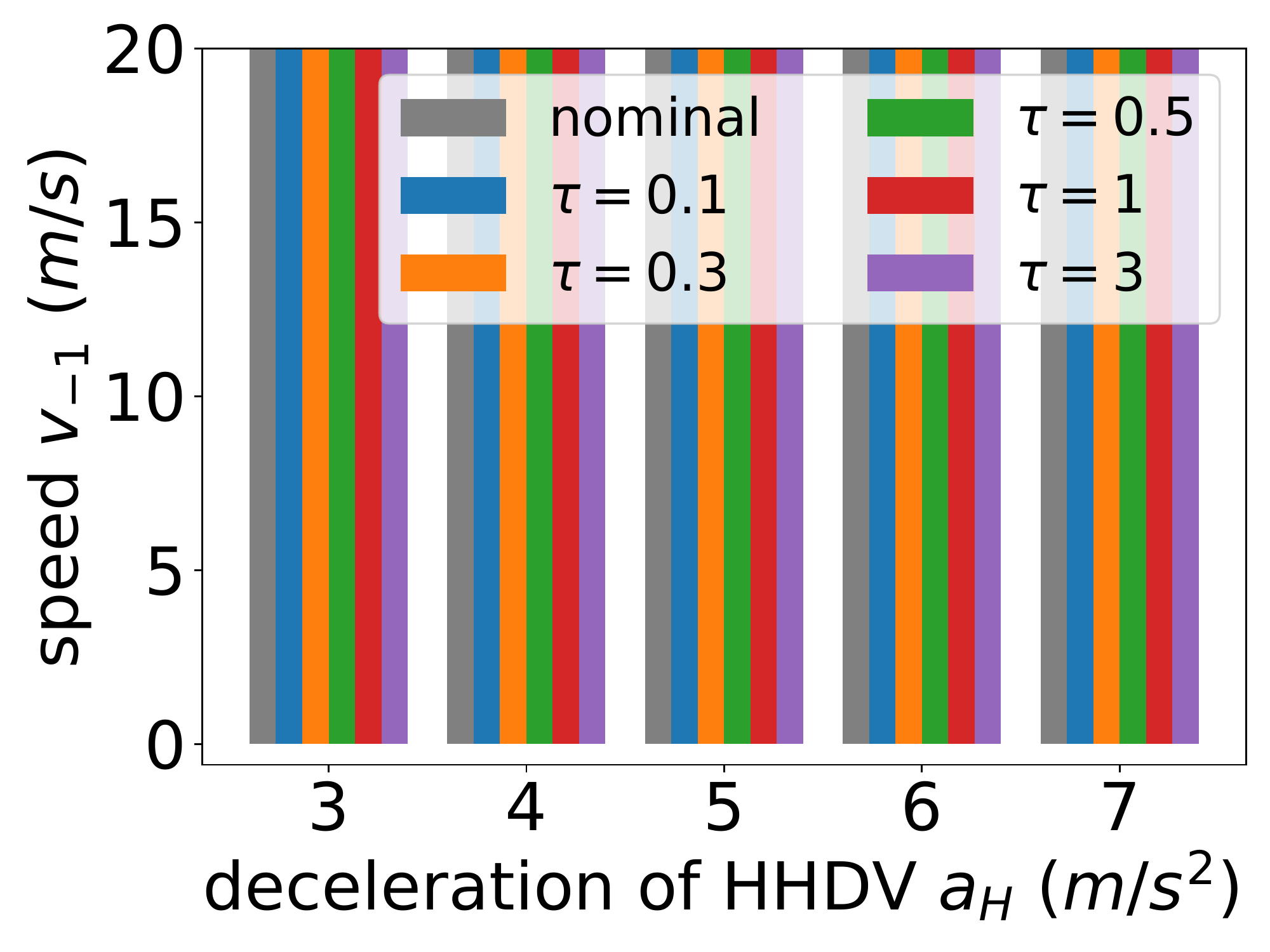}}
\subcaptionbox{FHDV-2}{\includegraphics[width=0.23\linewidth]{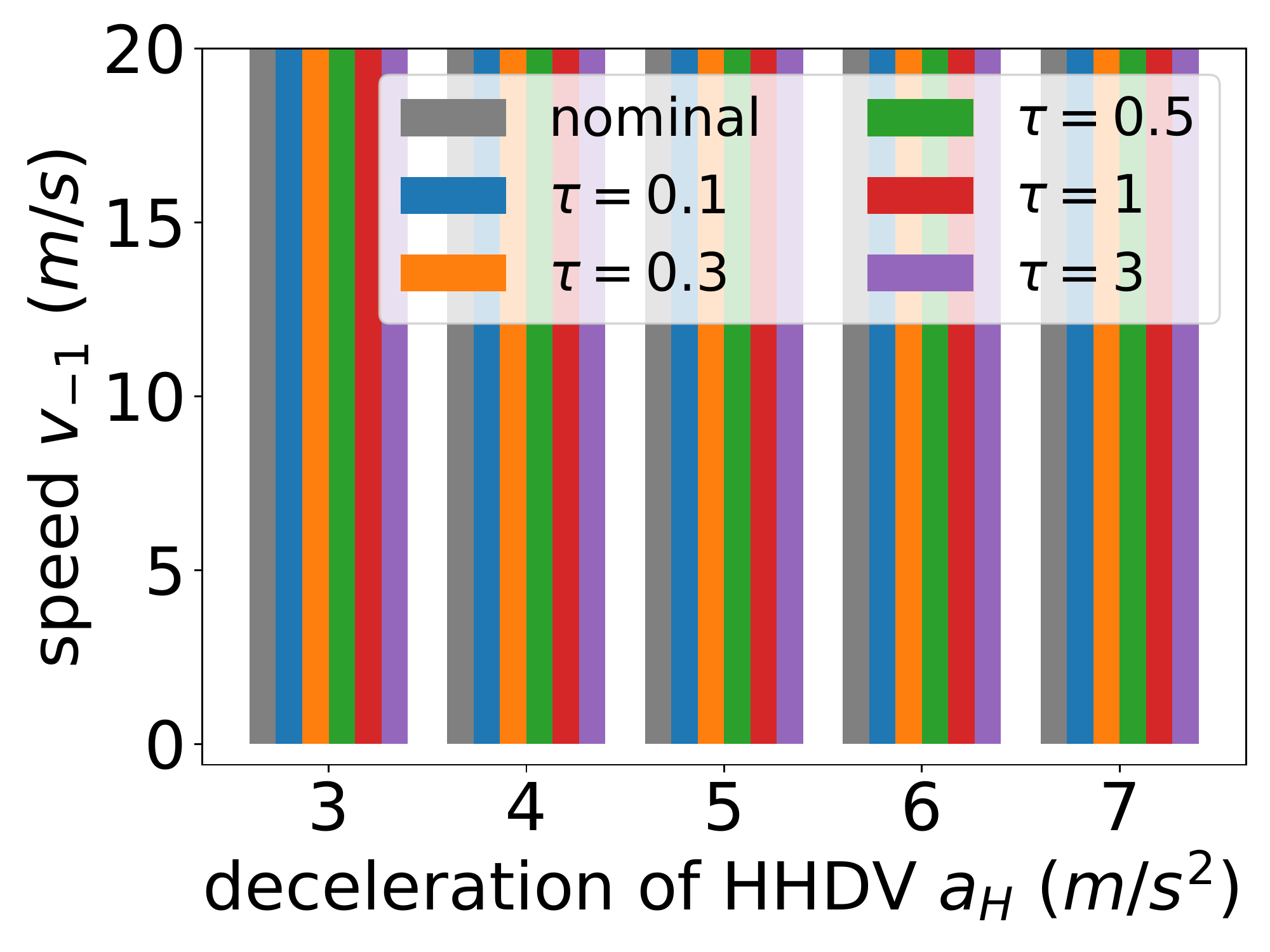}}
\\[6pt]
Scenario 2 \\[6pt]
\subcaptionbox{vehicle chain}{\includegraphics[width=0.23\linewidth]{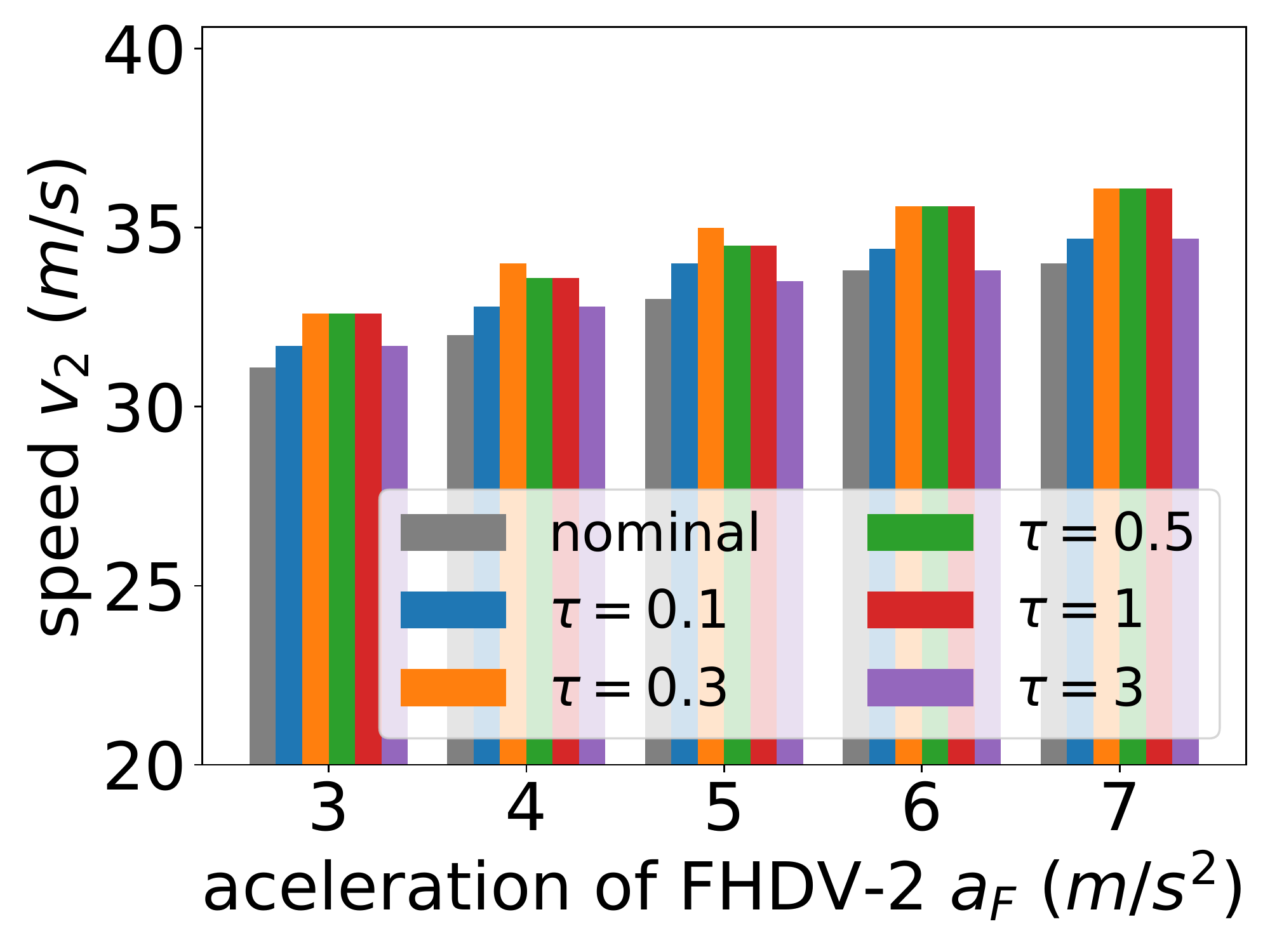}}
\subcaptionbox{CAV}{\includegraphics[width=0.23\linewidth]{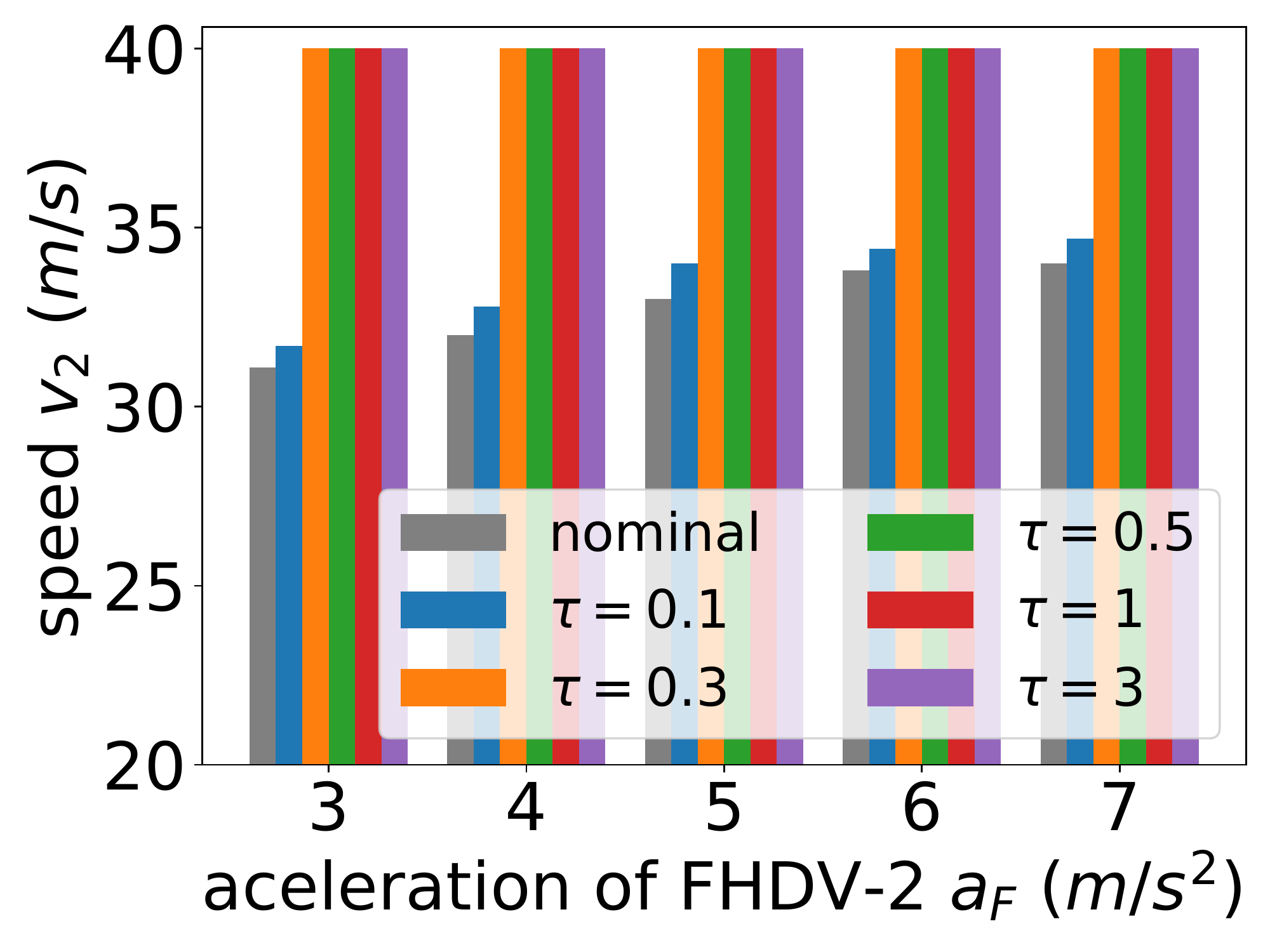}}
\subcaptionbox{FHDV-1}{\includegraphics[width=0.23\linewidth]{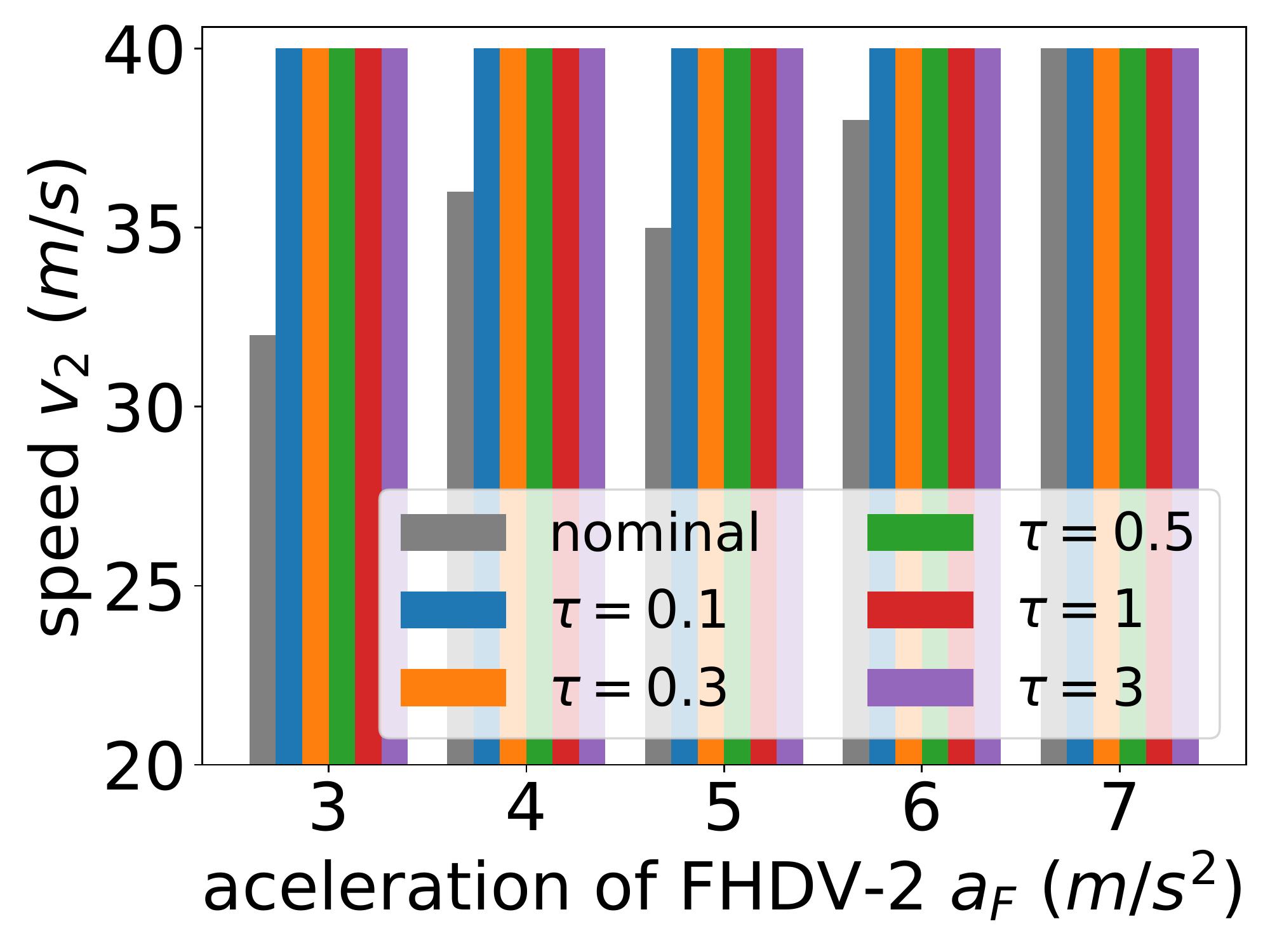}}
\subcaptionbox{FHDV-2}{\includegraphics[width=0.23\linewidth]{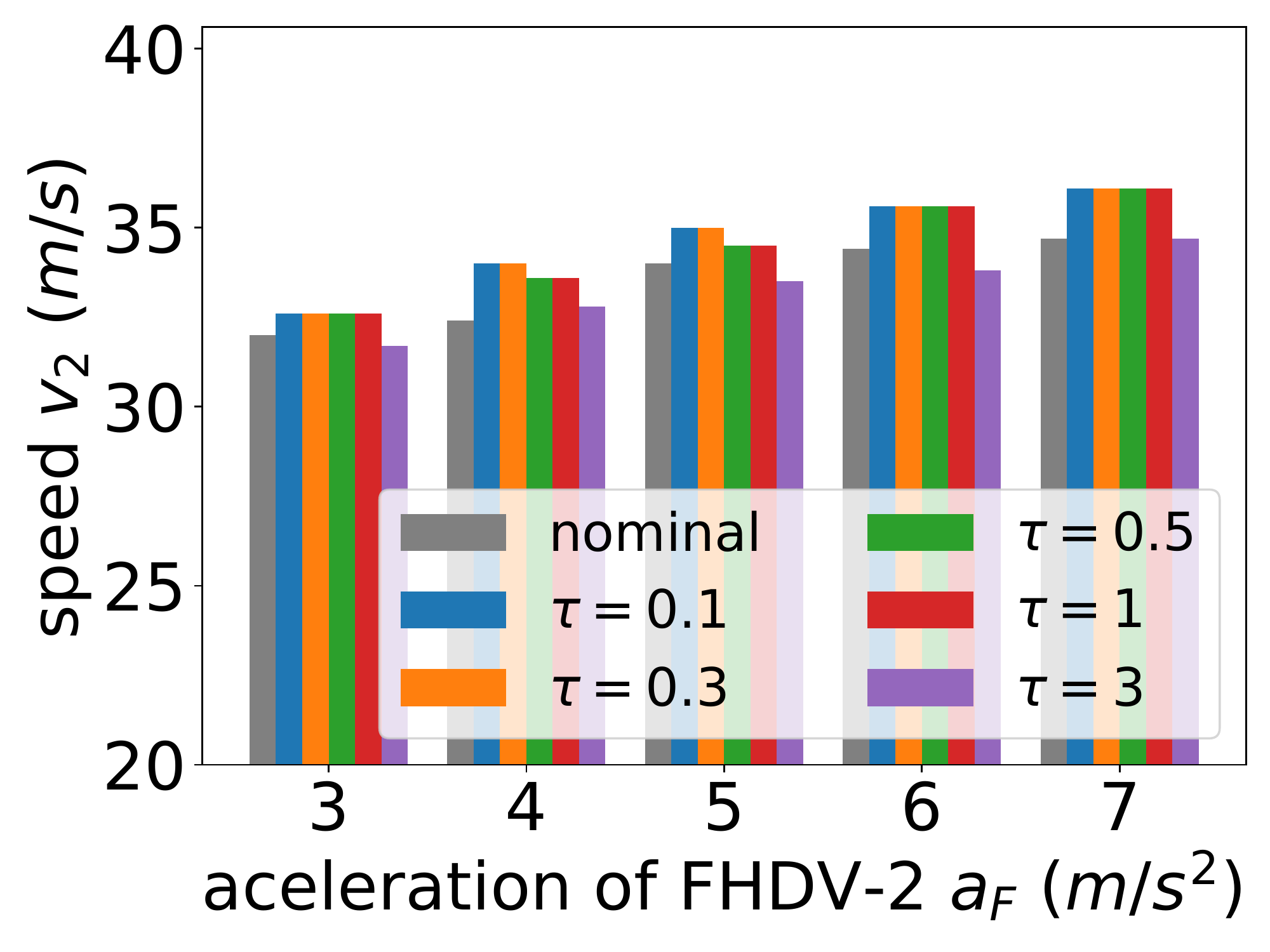}}
\caption{
The effect of HDV velocity perturbations on the safety of STC with acceleration limits, as a function of parameter $\tau$ used in the SDH safe spacing policy.
Remarkably, STC is able to keep the entire vehicle chain safe for a wide range of velocity perturbations with the right choice of $\tau$ (such as $\tau=1$ s).
}
\label{fig:sim:region SDH}
\end{figure}

Collision-free behavior not only depends on the parameters of the controller but also on the safe spacing policy implemented.
So far, we showed results for the case of the SDH policy only.
Now we compare the TH~\eqref{eq:safety:TH}, TTC~\eqref{eq:safety:TTC} and  SDH~\eqref{eq:safety:SDH} safe spacing policies.
While there has been related research comparing  the TH and TTC policies~\cite{vogel2003comparison}, a detailed comparison of these in the context of safety-critical traffic control have not yet been investigated.

Figure~\ref{fig:sim:region compare} shows the safe ranges of HDV velocity perturbations using the different safe spacing policies.
The parameters $\gamma$ and $p$ are fixed, while $\tau$ is varied in the range $\tau \in \{0.1,0.3,0.5,1,3\}$.
Considering the entire vehicle chain, we find that the safe region of the SDH policy is slightly larger than those of the other two policies in Scenario 1, and the difference between the three policies is negligible in Scenario 2.  
This showcases the ability of the proposed STC framework to accommodate various spacing policies.

\begin{figure}[t]
\centering
\hspace*{-0.5cm}
Scenario 1 \\[6pt]
\subcaptionbox{vehicle chain}{\includegraphics[width=0.23\linewidth]{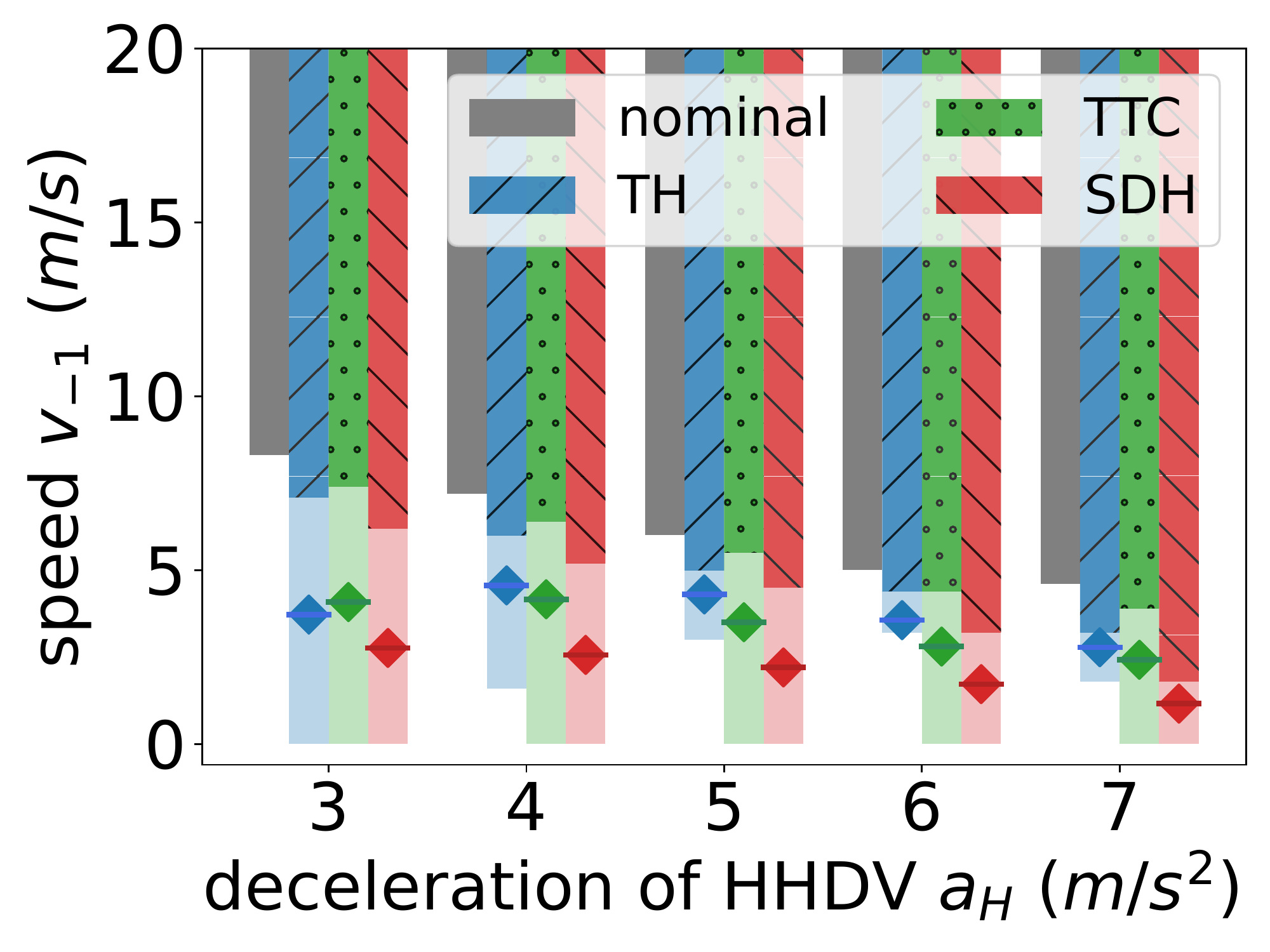}}
\subcaptionbox{CAV}{\includegraphics[width=0.23\linewidth]{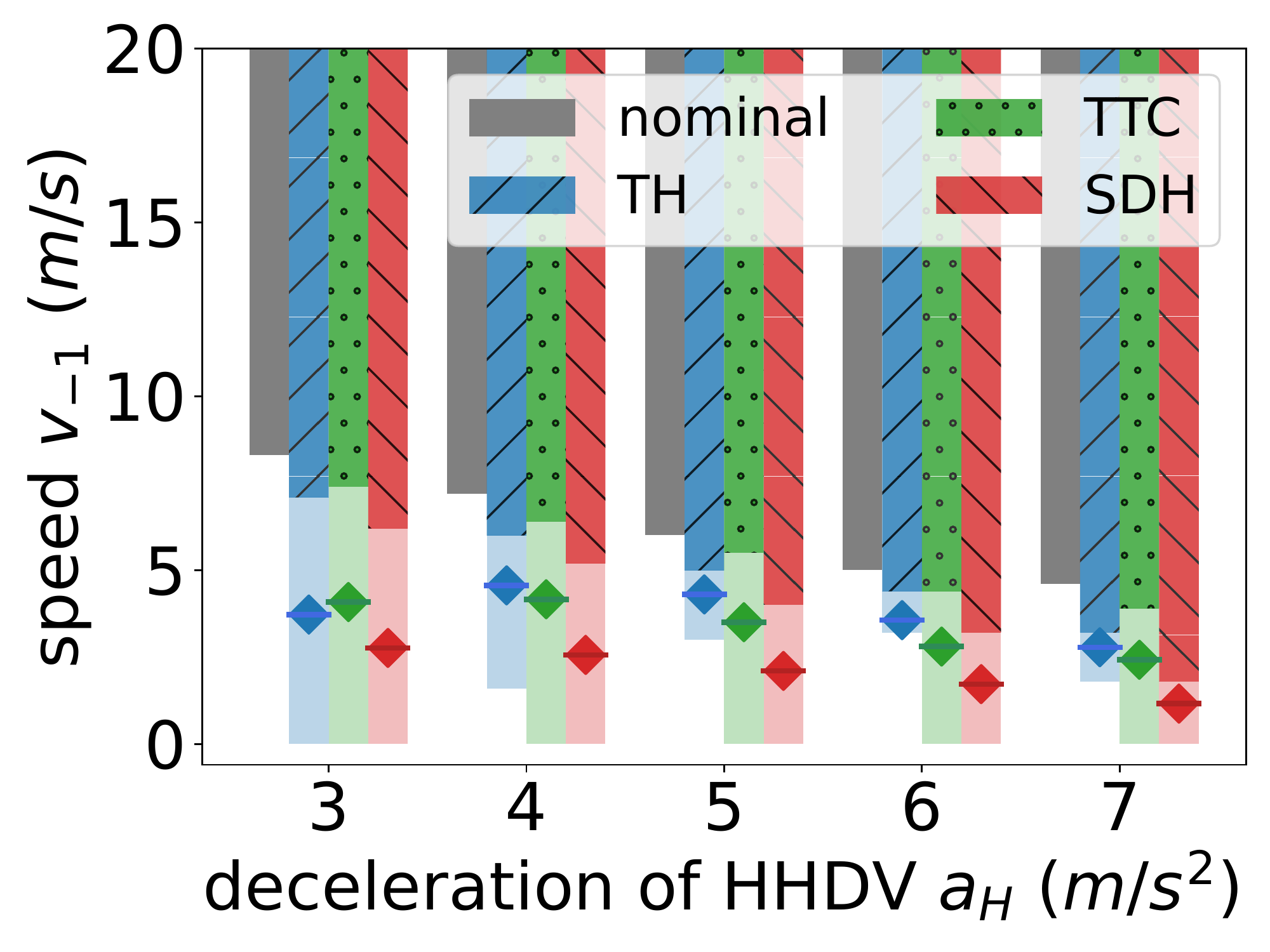}}
\subcaptionbox{FHDV-1}{\includegraphics[width=0.23\linewidth]{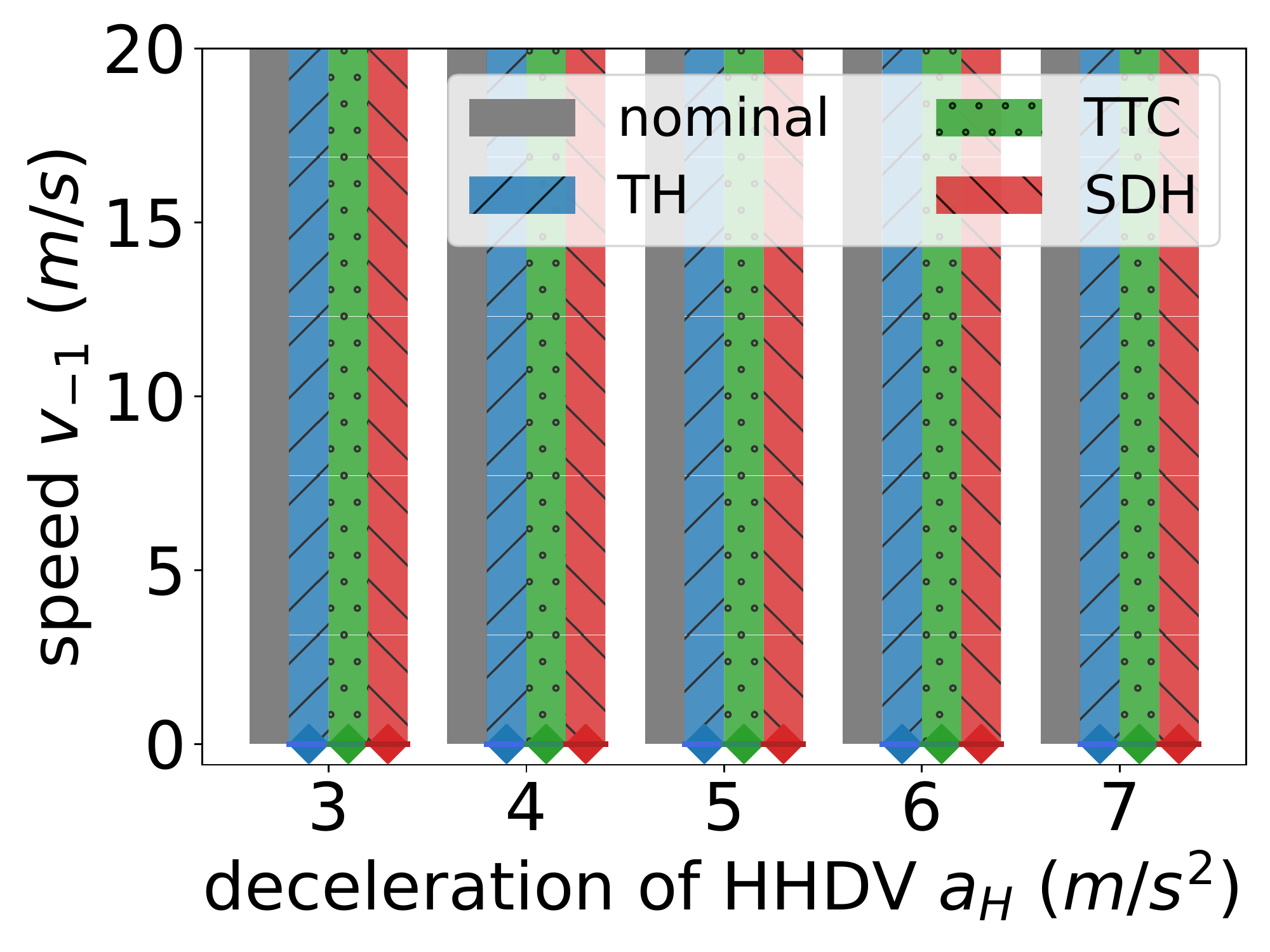}}
\subcaptionbox{FHDV-2}{\includegraphics[width=0.23\linewidth]{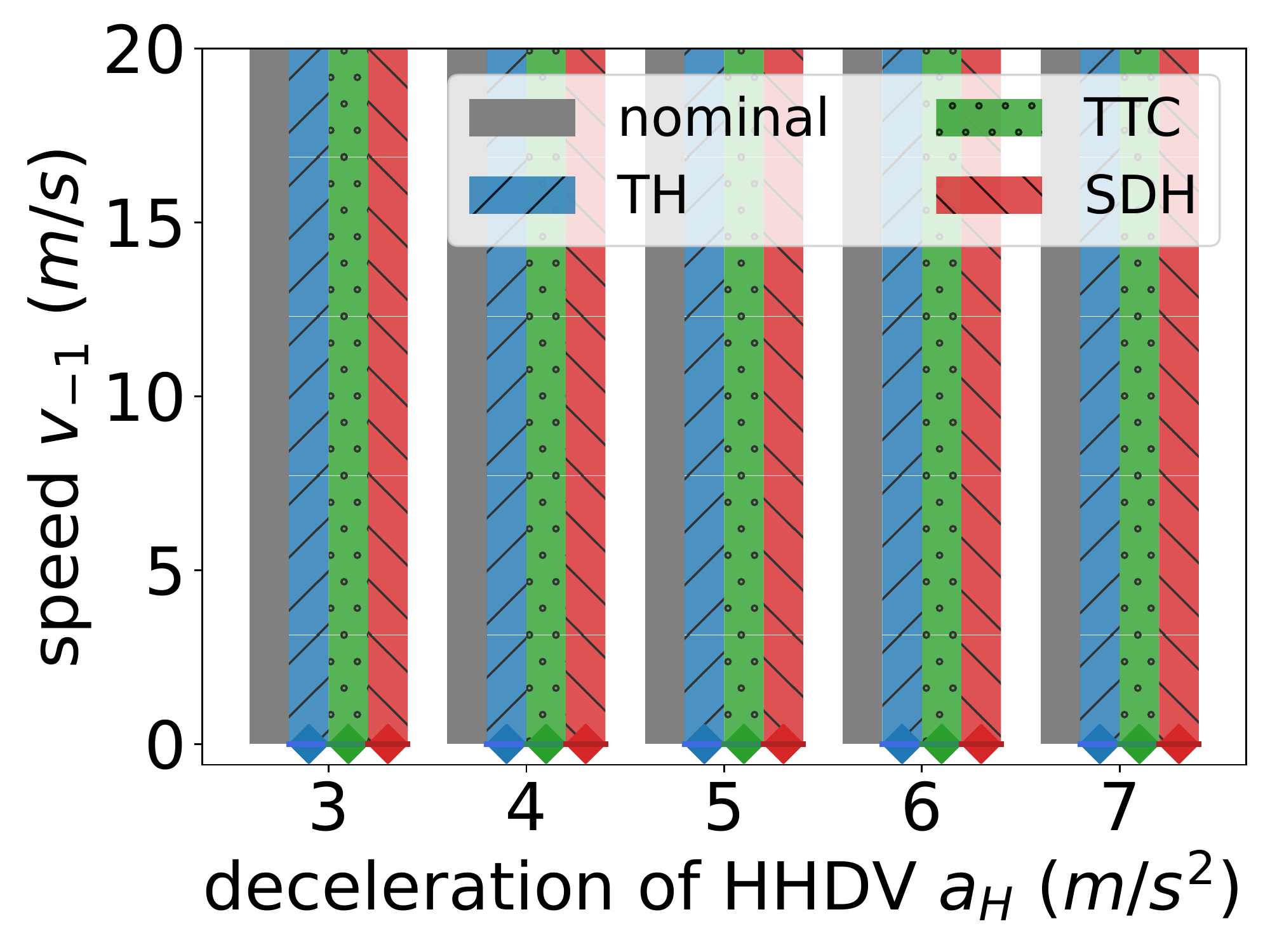}}
\\[6pt]
Scenario 2 \\[6pt]
\subcaptionbox{vehicle chain}{\includegraphics[width=0.23\linewidth]{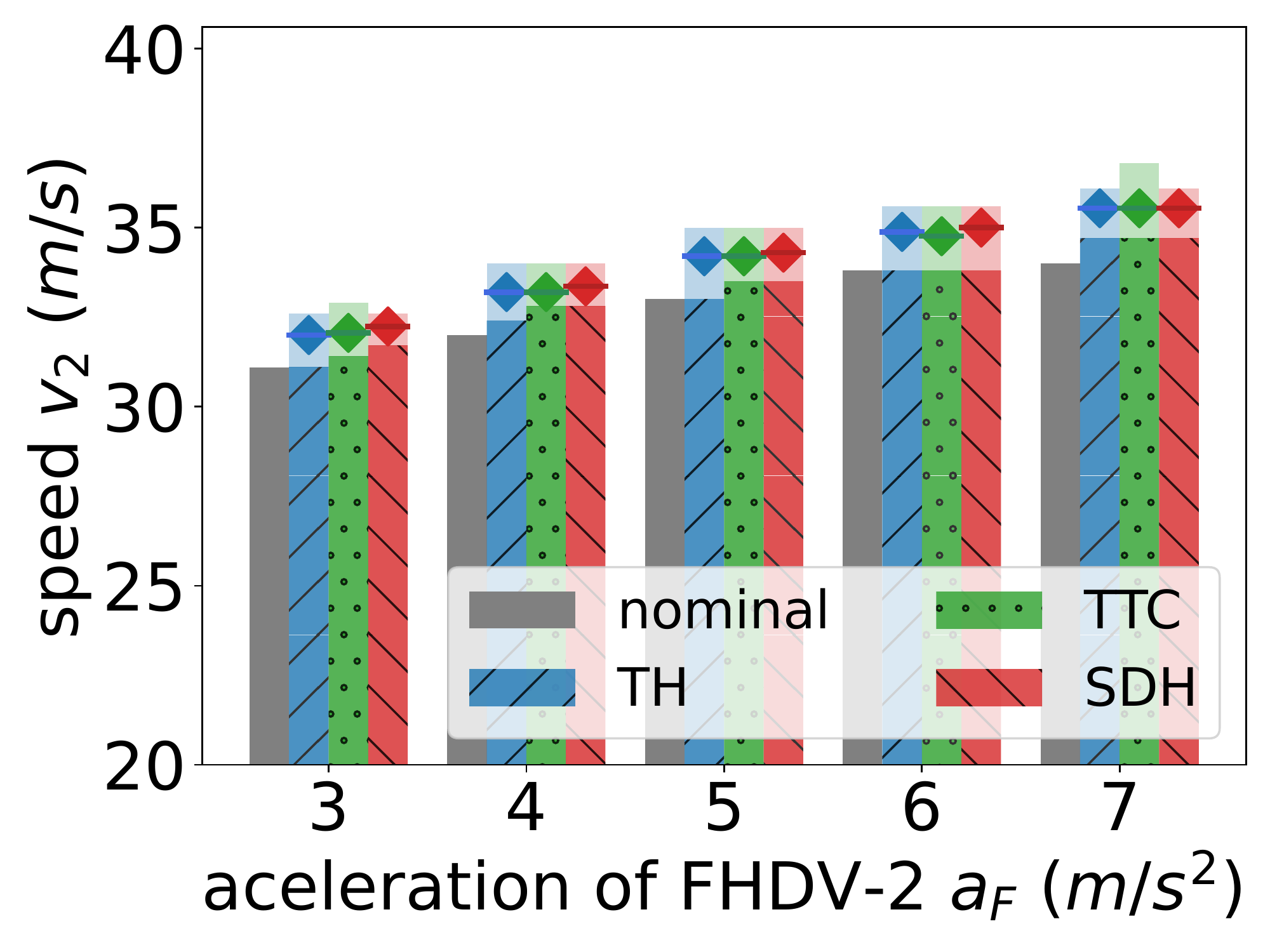}}
\subcaptionbox{CAV}{\includegraphics[width=0.23\linewidth]{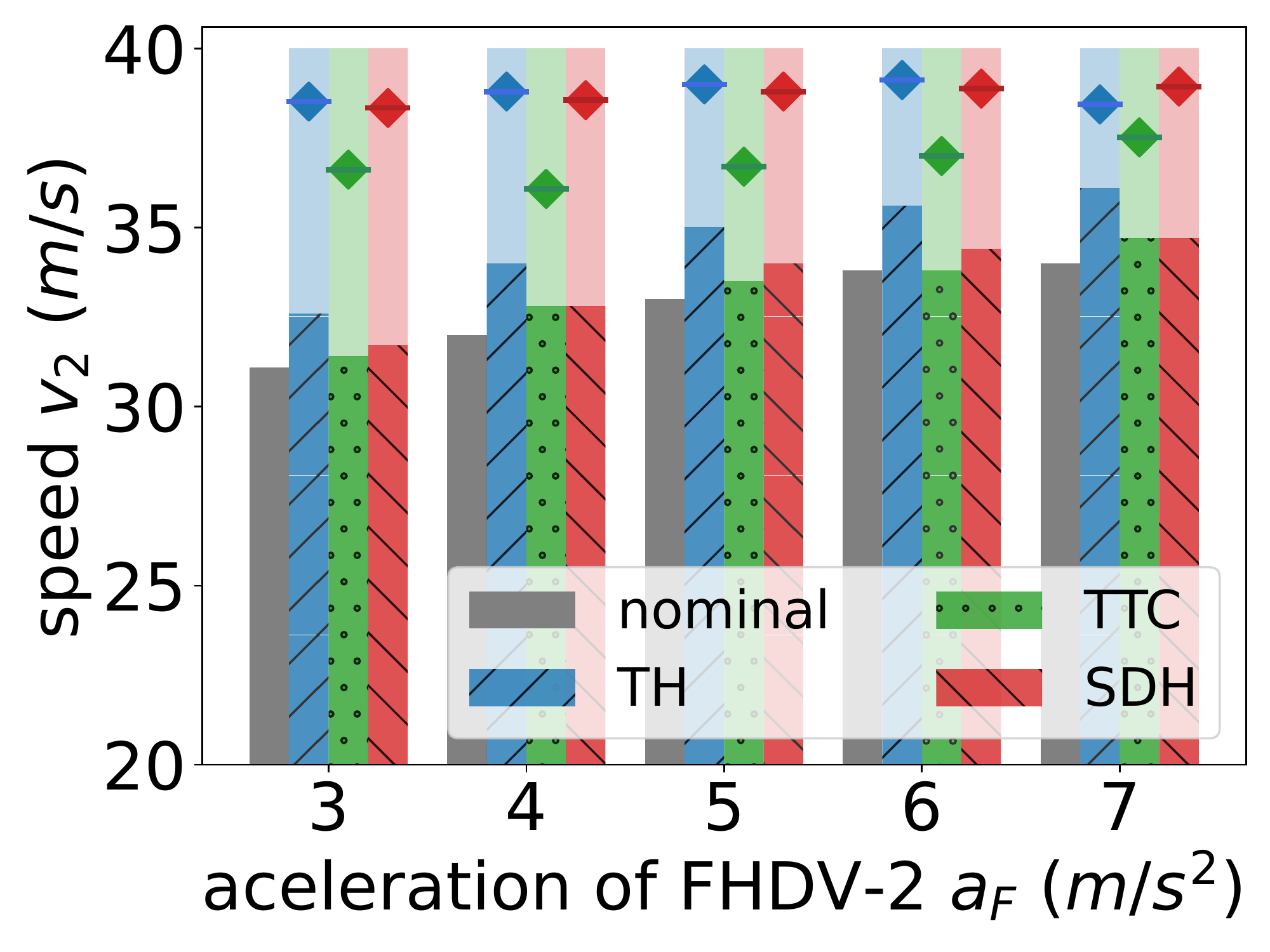}}
\subcaptionbox{FHDV-1}{\includegraphics[width=0.23\linewidth]{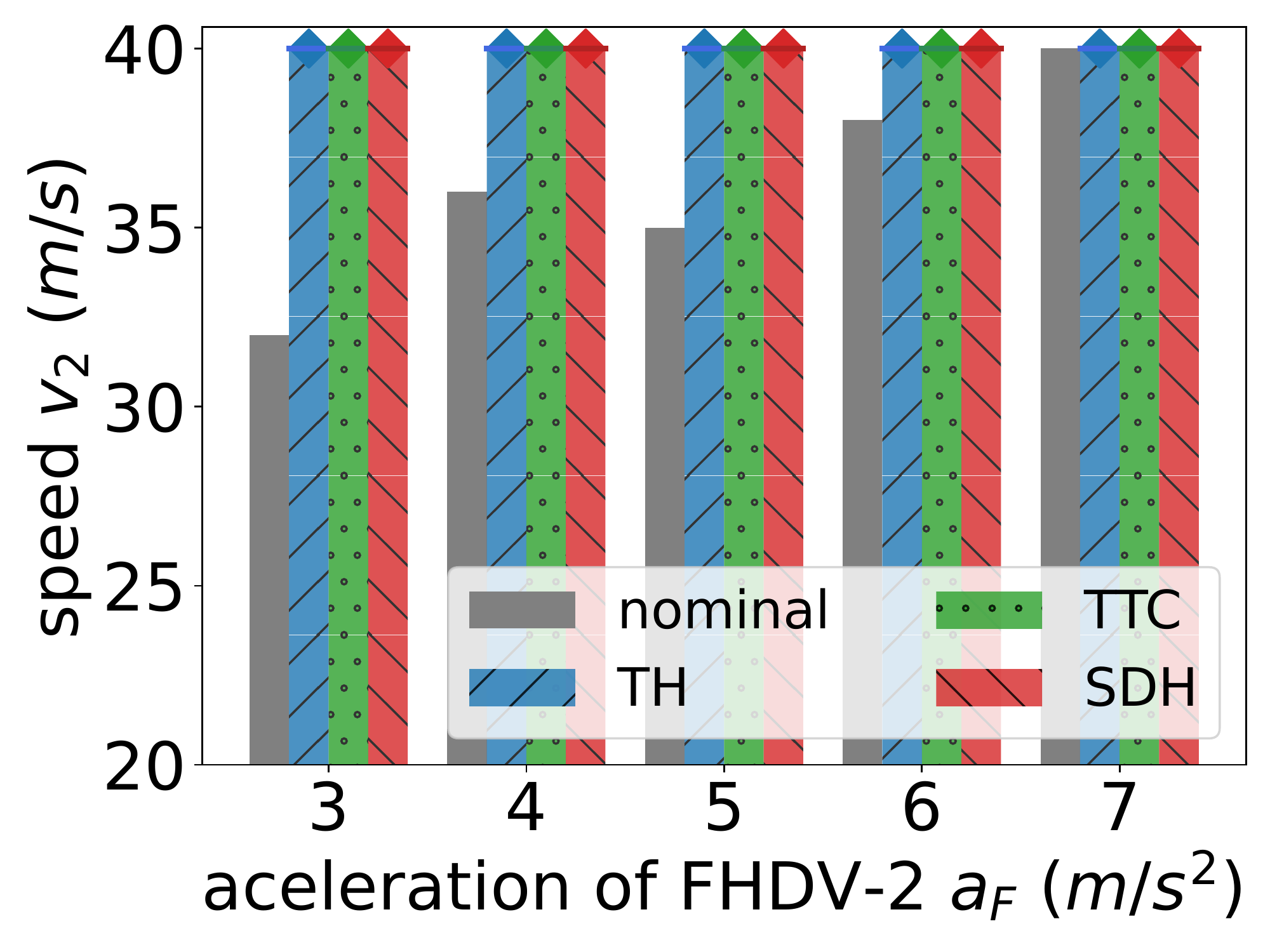}}
\subcaptionbox{FHDV-2}{\includegraphics[width=0.23\linewidth]{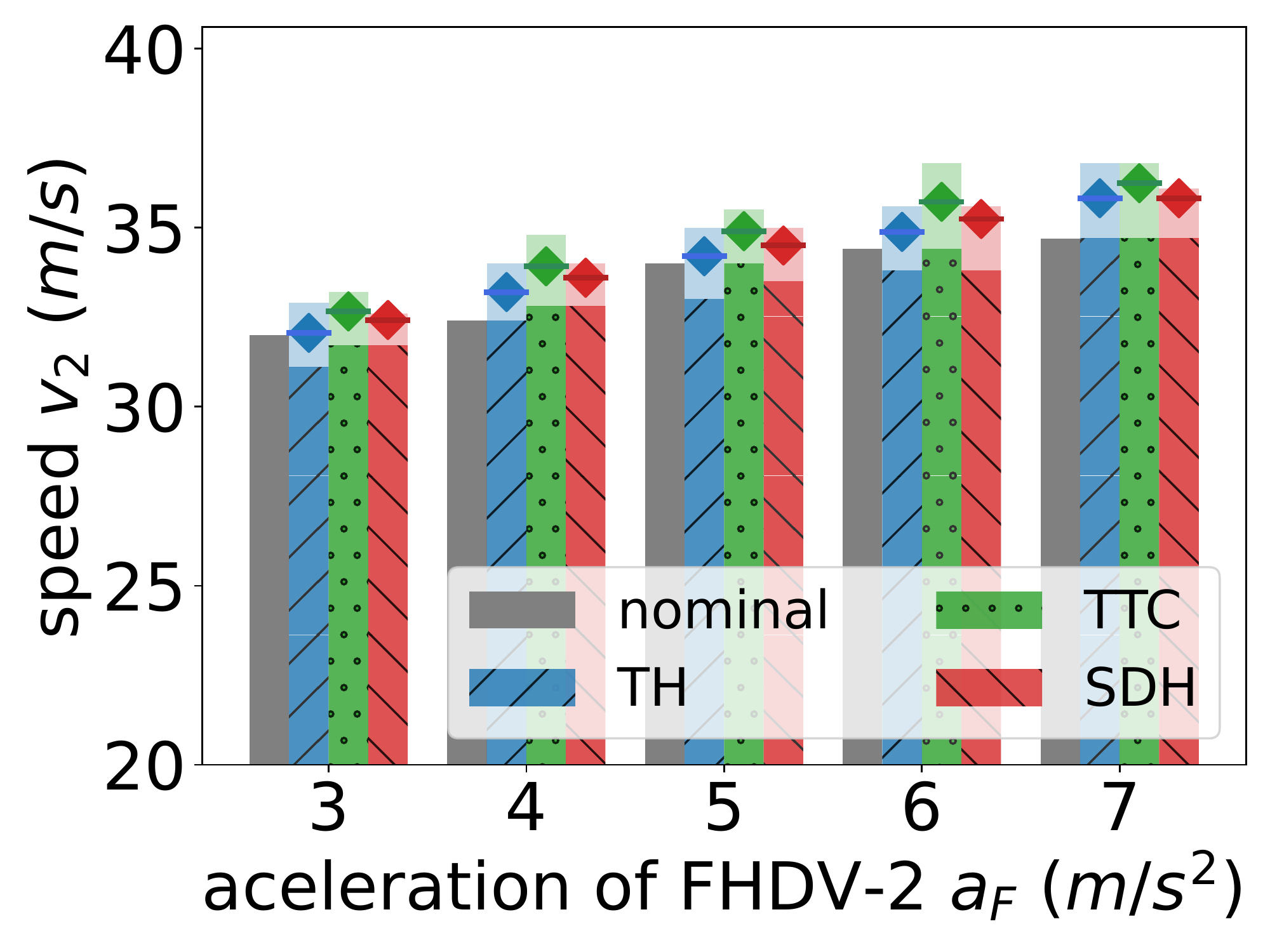}}
\caption{
Comparison of STC implemented with the time headway (TH), time-to-collision (TTC) and stopping distance headway (SDH) safe spacing policies given in~\eqref{eq:safety:TH}-\eqref{eq:safety:SDH}, using the metrics introduced in Fig.~\ref{fig:sim:region SDH}.
Considering the whole vehicle chain, the three spacing policies show similar performance (with SDH being slightly superior to the other two policies).
This shows the compatibility of STC with various safe spacing policies.
The results are shown for a range of parameter $\tau$ (that is time headway or time-to-collision).
For each safe spacing policy, dark bars and light bars represent the minimum and maximum values, whereas the horizontal line with marker indicates the average size of safe velocity perturbations over the range of $\tau$.}
\label{fig:sim:region compare}
\end{figure}

\begin{table}[t]
    \centering
    \caption{Parameters used for simulations driven by the NGSIM data}
    \label{tab:parameters_NGSIM}
    \begin{tabular}{cccccc}
    \hline
    Vehicle & Variable & Symbol & Value & Unit\\
    \hline
    \multirow{4}{*}{all}
    & equilibrium velocity & $v^\star$ & $8$ & m/s \\
    & equilibrium spacing & $s^\star$ & $14.8$ & m \\
    & braking limit & $a_{\min}$ & $-7$ & m/s$^2$ \\
    & acceleration limit & $a_{\max}$ & $7$ & m/s$^2$ \\
    \hline
    \multirow{5}{*}{HDVs}
    & speed limit & $v_{\max}$ & $46.9$ & m/s \\
    & standstill spacing & $h_{\rm st}$ & $1.6$ & m \\
    & free flow spacing & $h_{\rm go}$ & $50$ & m \\
    & \multirow{2}{*}{model coefficients}
    & $\A$ & $0.16$ & 1/s \\
    & & $\B$ & $0.63$ & 1/s \\
    \hline
    \multirow{8}{*}{CAV}
    & initial spacing & $s_0(0)$ & $50$ & m\\
    & \multirow{4}{*}{gains of nominal controller}
    & $\mu_1$ & $-2$ & 1/s$^2$ \\
    & & $\mu_2$ & $-2$ & 1/s$^2$ \\
    & & $k_1$ & $0.2$ & 1/s \\
    & & $k_2$ & $0.2$ & 1/s \\
    & \multirow{3}{*}{parameters of safety filter}
    & $\tau$ & $3$ & s \\
    & & $\gamma$ & $10$ & 1/s \\
    & & $p$ & $100$ & 1/s$^2$ \\
    \hline
    \end{tabular}
\end{table}

\begin{figure}[t]
    \centering
    Nominal Controller \\[6pt]
    \includegraphics[width=0.23\textwidth]{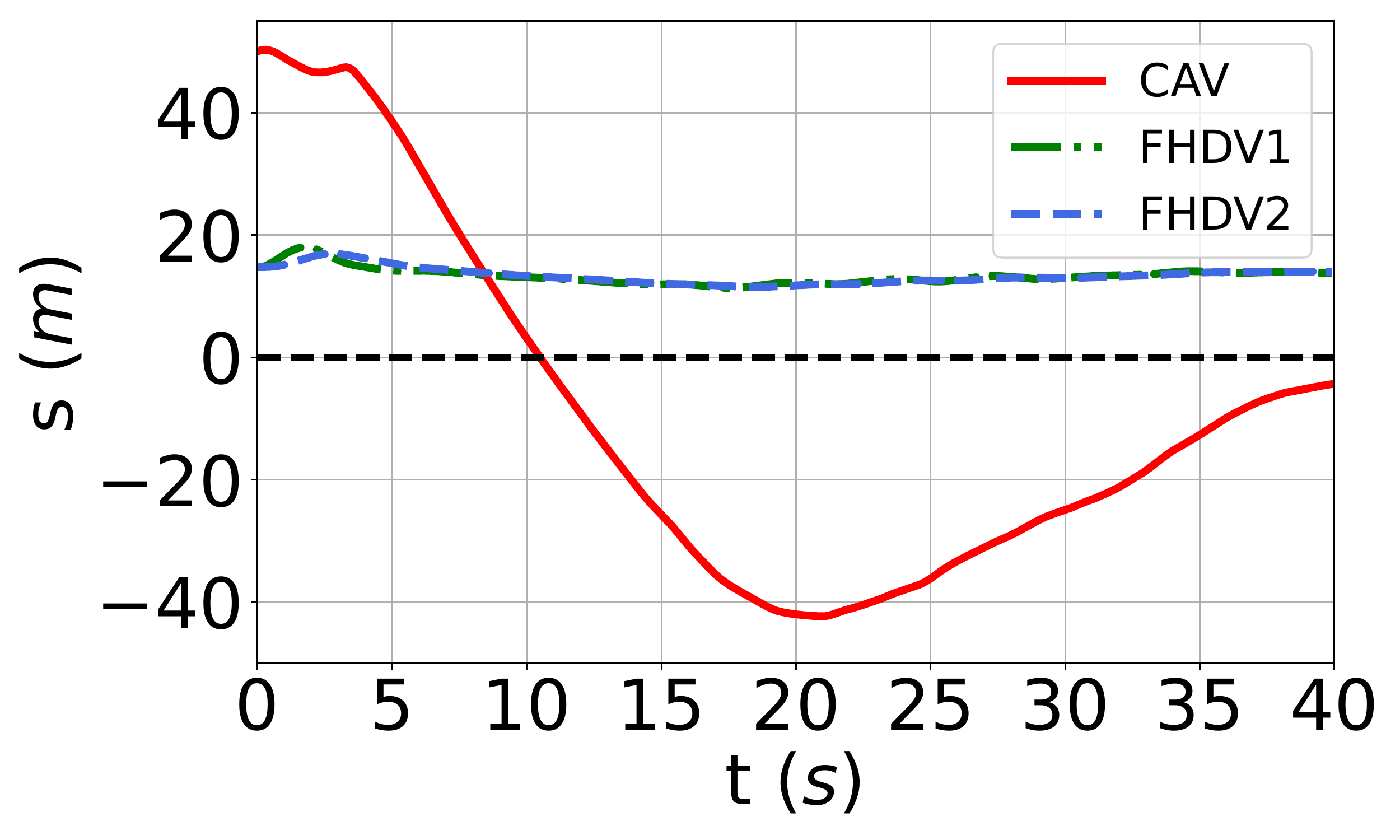}
    \includegraphics[width=0.23\textwidth]{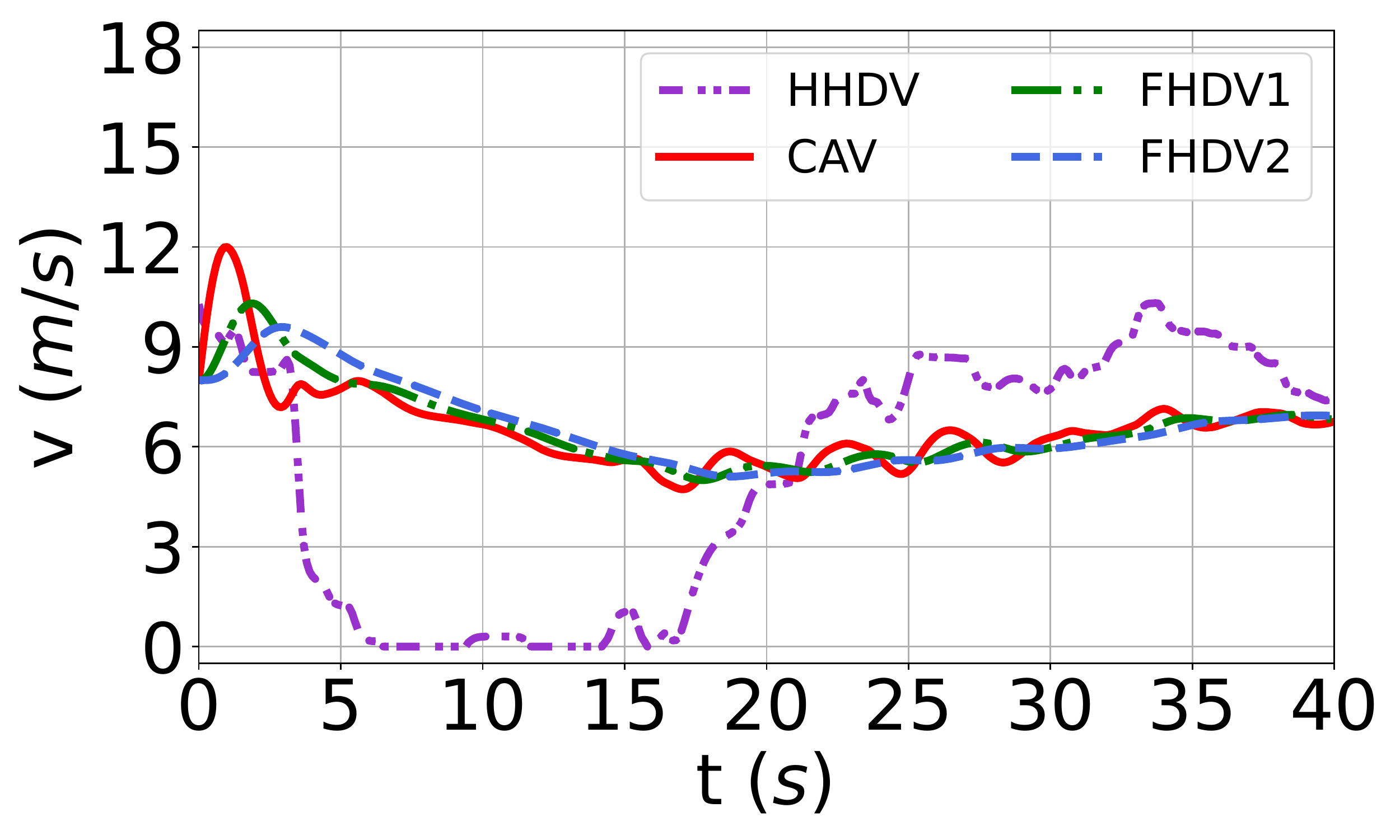}
    \includegraphics[width=0.23\textwidth]{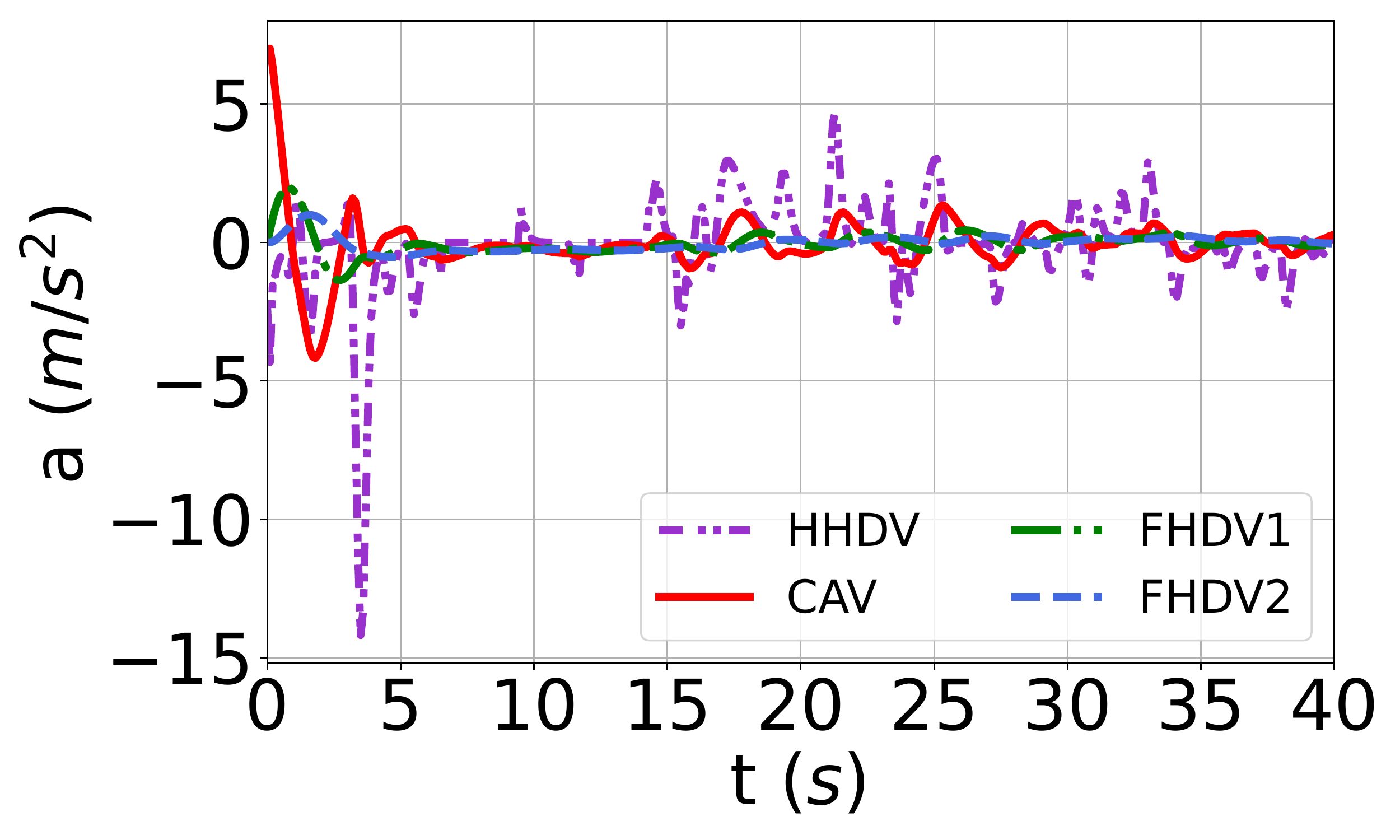}
    \includegraphics[width=0.23\textwidth]{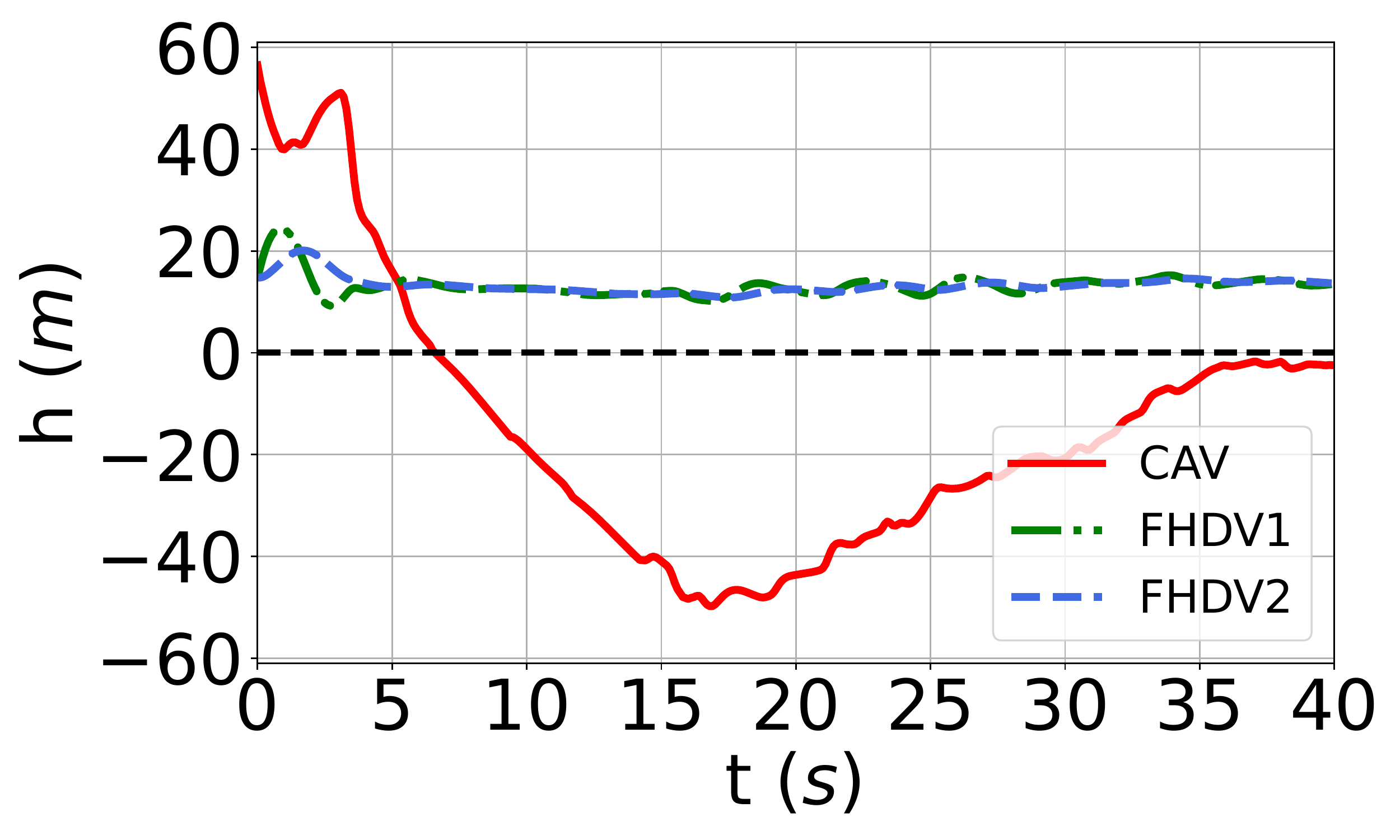}
    \\
    Proposed STC \\[6pt]
    \includegraphics[width=0.23\textwidth]{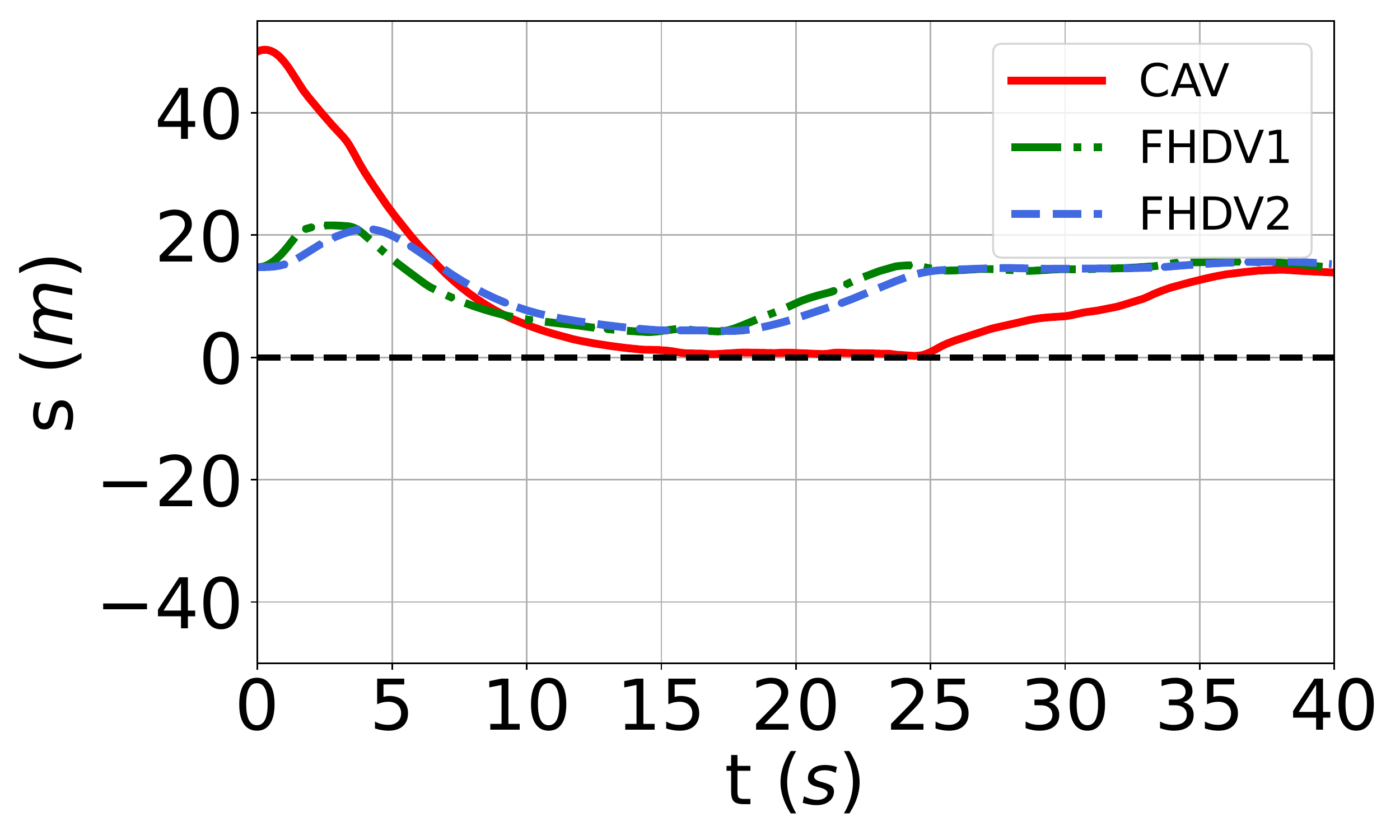}
    \includegraphics[width=0.23\textwidth]{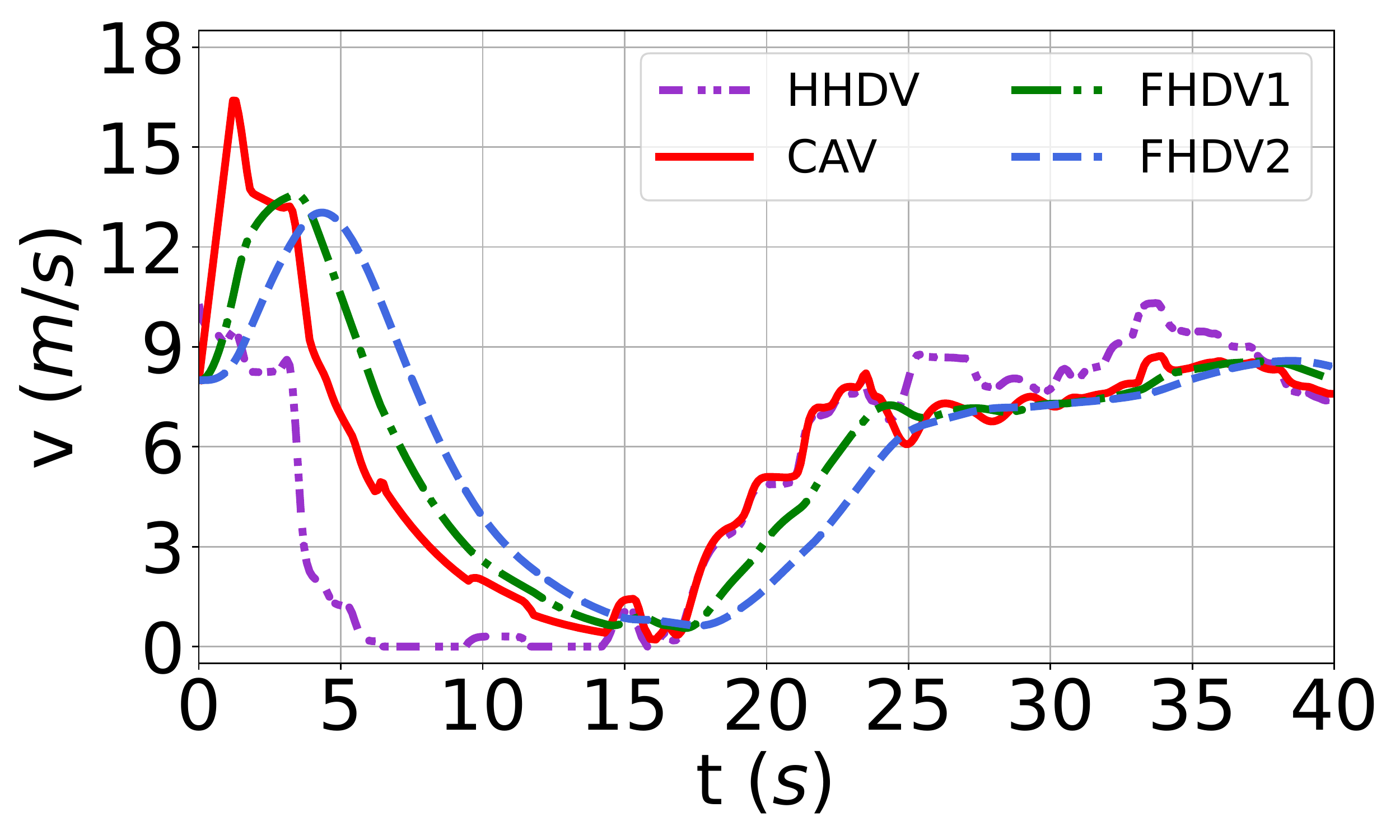}
    \includegraphics[width=0.23\textwidth]{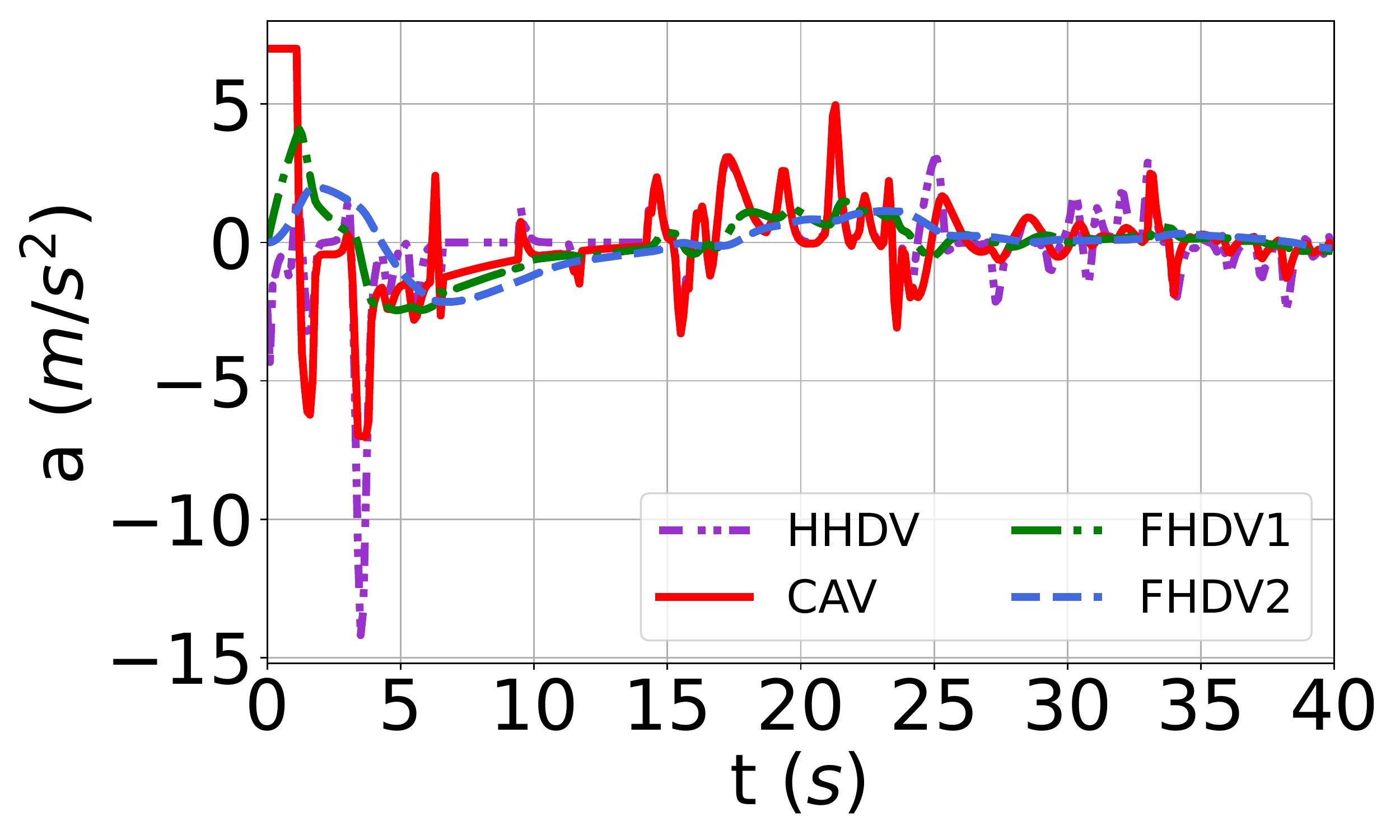}
    \includegraphics[width=0.23\textwidth]{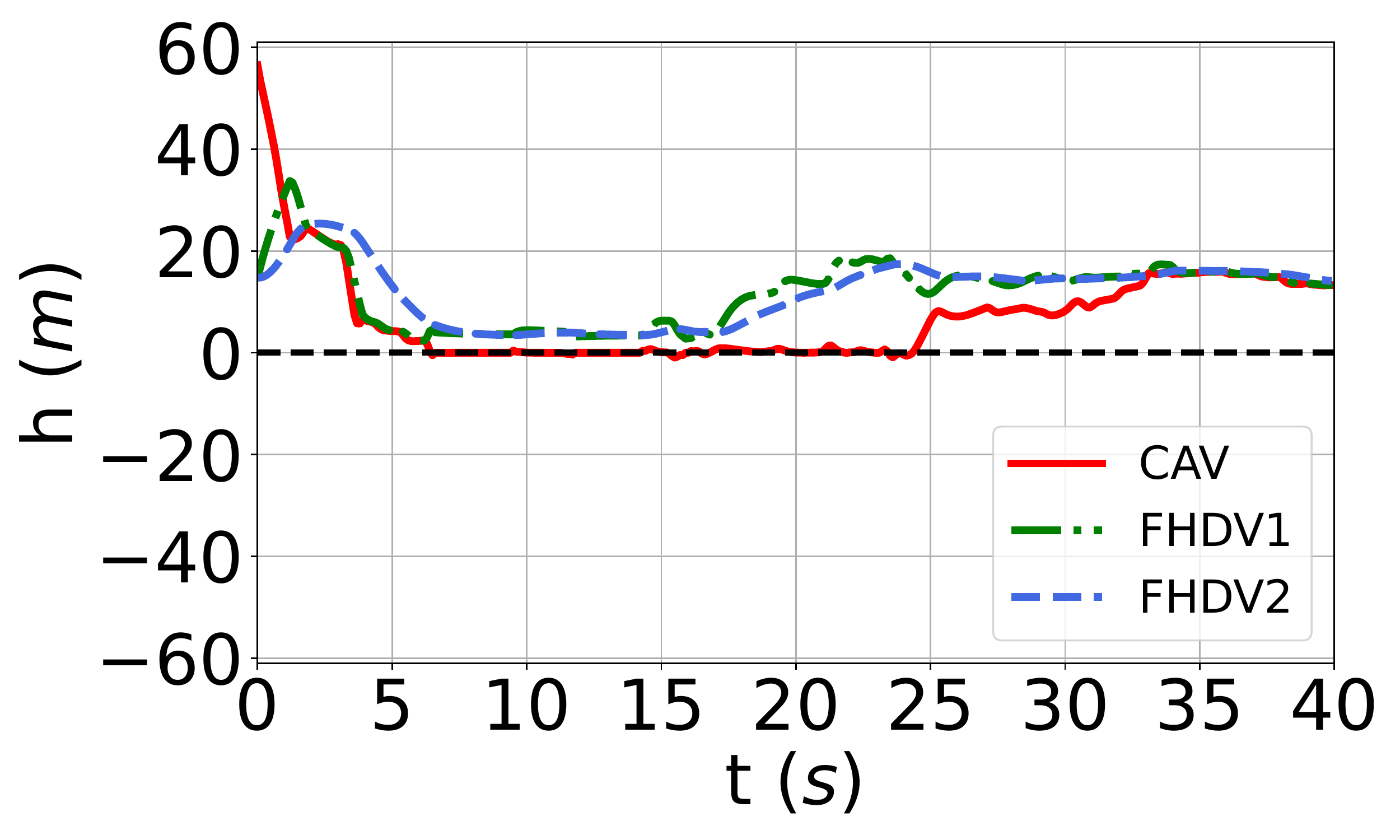}
    \caption{
    Implementation of safety-critical traffic control (STC) in a real-world traffic scenario.
    The motion of the head HDV is given by a trajectory from the NGSIM dataset, measured on highway I80 in California.
    Meanwhile the CAV and the following HDVs are simulated, using the nominal controller~\eqref{eq:nominal controller} and the proposed STC~\eqref{eq:QP}.
    Remarkably, STC (bottom) is able to maintain safety when the head vehicle performs a harsh brake that happened in real life.
    This way, STC improves safety compared to the nominal controller (top), while leveraging the beneficial properties of this controller.}
    \label{fig:sim:NGSIM trajectory}
\end{figure}

\begin{figure}[t]
    \centering
    \includegraphics[width=0.6\linewidth]{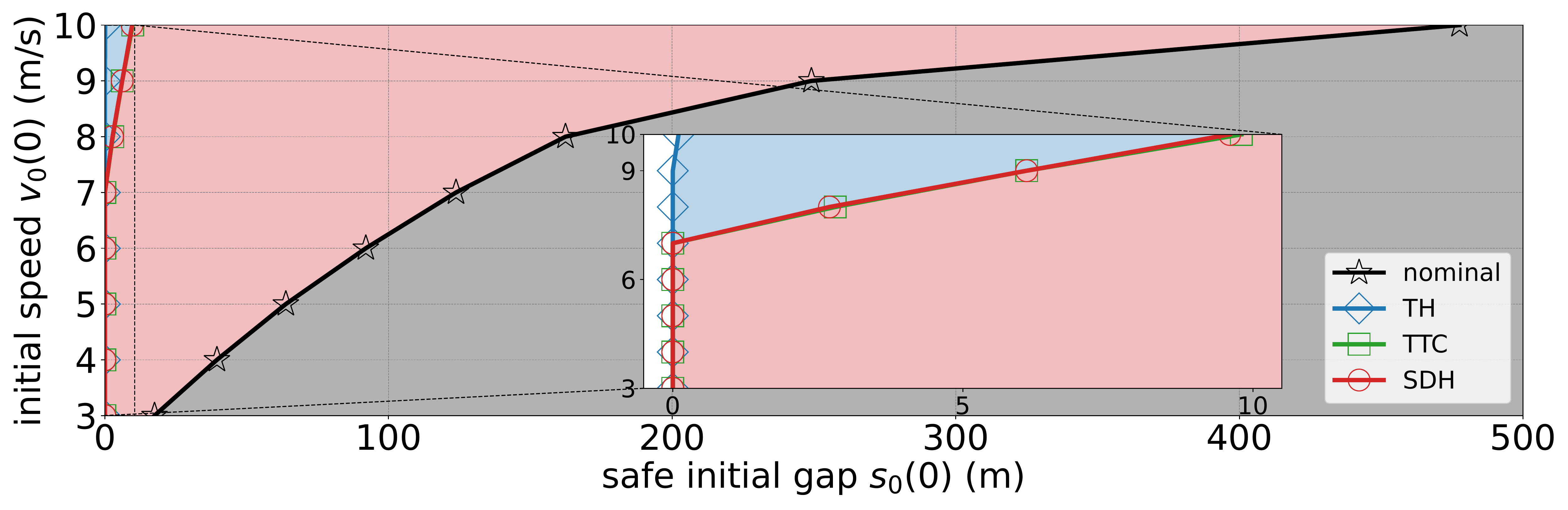}
    \caption{
    Safety regions of the nominal controller and STC with three safe spacing policies, for the real-world traffic scenario presented in Fig.~\ref{fig:sim:NGSIM trajectory}.
    If initial conditions are selected from the shaded domain (with gray, blue, green and red shadings for the cases of the nominal controller, and STC with the TH, TTC and SDH policies, respectively), then all vehicles avoid collision, even with limited acceleration capabilities.
    Clearly, the nominal controller requires large initial spacing to showcase safe behavior, while STC is able to prevent collisions for significantly smaller spacing.
    }
    \label{fig:sim:NGSIM:region}
\end{figure}

\subsection{Validation Using NGSIM Data}

Finally, we demonstrate the applicability of STC to real-world traffic scenarios by using experimental data for the trajectory of the HHDV while simulating the motions of other vehicles.

We leverage the Next Generation Simulation (NGSIM) dataset that is widely used in transportation research.
In particular, we use the reconstructed trajectory data from~\cite{montanino2015trajectory}, that contains information about the motions of vehicles on highway I80 in Emeryville, California, including their position and speed as a function of time, with a resolution of 0.1 second.
We select a vehicle (specifically, vehicle-2169) whose motion contains significant speed fluctuations, hence it poses a challenge for controlling a CAV behind it in a safety-critical fashion.
The trajectory of this vehicle, including its speed and acceleration are plotted in Fig.~\ref{fig:sim:NGSIM trajectory} by dashed purple line.
Notice the abrupt braking at around $t=4$ s, which can cause danger for the following vehicles.
We use this speed profile as the speed of the HHDV to simulate the motions of the CAV and the FHDVs.

Figure~\ref{fig:sim:NGSIM trajectory} presents simulation results that compare the nominal controller and STC (with the SDH policy).
The simulation parameters of the OVM \eqref{eq:OVM}-\eqref{eq:Vs} are given in Table~\ref{tab:parameters_NGSIM} (note that these were calibrated to the NGSIM dataset). The initial speeds of the CAV and the two FHDVs were set to $v_0(0) = v_1(0) = v_2(0) = v^\star$, and the initial spacings were chosen as $s_1(0) = s_2(0) = s^\star$ while $s_0(0)$ was set independently.
The figure clearly reveals that the proposed STC framework is able to successfully modify the nominal controller and endow it with safe behavior, while the nominal controller on its own could cause collision between the CAV and the HHDV.
Importantly, the harsh braking of the HHDV happened in real life, hence it is crucial to maintain safety guarantees while attempting to regulate the traffic flow with the motion of the CAV.

Finally, we further analyze to which extent STC improves safety.
We repeat the simulation in Fig.~\ref{fig:sim:NGSIM trajectory} for various initial conditions.
The motion of the HHDV is kept the same (given by the experimental data), and the initial velocity $v_0(0)$ and initial spacing $s_0(0)$ of the CAV are varied (while the rest of the initial conditions of FHDVs are determined by $v_0(0) = v_1(0) = v_2(0) = v^\star$ and $s_1(0) = s_2(0) = s^\star$). 
Figure~\ref{fig:sim:NGSIM:region} shows the result of executing the nominal controller and STC with the three different safe spacing policies.
For each case, thick lines with markers indicate the boundary of the shaded safety region where initial conditions are associated with collision-free motion for all vehicles.
STC significantly enlarges the safety region of the nominal controller and, as a result, improves the safety of mixed traffic.
Importantly, STC achieves this by incorporating and leveraging the nominal controller, while providing formal guarantees of safety.

\section{Conclusions}
\label{sec:conclusions}

In this paper, we proposed a safety-critical traffic control (STC) framework, in which the longitudinal motions of connected automated vehicles (CAVs) are regulated amongst connected human-driven vehicles (HDVs) in mixed-autonomy traffic flows.
The end goal of STC is to mitigate traffic congestions (i.e., achieve string stability), while guaranteeing safe, collision-free behavior.
We formulated STC as a safety filter using control barrier functions (CBFs), that endows a nominal controller capable of attaining string stability with formal safety guarantees.
To this end, we considered the safety of both the CAV and the following HDVs.
Moreover, we employed state observer-based CBFs to establish STC when the CAV does not have access to all the states of the following HDVs.
In our future work, we plan to investigate the effects of partial connectivity and address feedback, communication and actuation delays for the safety-critical control of mixed-autonomy traffic.

\bibliographystyle{plain}

\bibliography{ref.bib}

\end{document}